\documentclass[11pt]{article}

\usepackage[bf,small]{caption}
\usepackage[export]{adjustbox}
\usepackage[sort&compress,square,numbers]{natbib}
\usepackage[title]{appendix}
\usepackage{amsmath,amsthm,amssymb,setspace,enumitem,epsfig,titlesec,verbatim,color,array,multirow,comment,graphicx}
\usepackage{amssymb,amsfonts,amsmath,amsthm}
\usepackage{authblk}
\usepackage{bbm}
\usepackage{bm}
\usepackage{ifthen}
\usepackage{xspace}
\usepackage{complexity}
\usepackage{url}
\usepackage{fullpage}
\usepackage{graphbox}
\usepackage{xcolor}
\usepackage{tikz}
\usetikzlibrary{positioning}

\definecolor{promink}{HTML}{C25E05}
\definecolor{dormink}{HTML}{E8912F}
\definecolor{unrink}{HTML}{6E6E6E}
\definecolor{promwash}{HTML}{FDEBD8}
\definecolor{dormwash}{HTML}{FEF5EA}
\definecolor{unrwash}{HTML}{F1F1F1}
\definecolor{dormfill}{HTML}{FFE3CA}
\definecolor{classcrimson}{HTML}{DC143C}

\usepackage{graphicx}
\usepackage{hyperref}
\usepackage[capitalize]{cleveref}
\crefname{figure}{Figure}{Figures}
\usepackage{multicol}
\usepackage{subcaption}
\usepackage{times}
\usepackage{nicefrac}

\usepackage[top=2.5cm,left=2.7cm,right=2.7cm,bottom=3.2cm]{geometry}

\usepackage[nomarkers,figuresonly,nofiglist]{endfloat}

\hypersetup{
  colorlinks=true,
  linktoc=all,
  linkcolor=blue,
  citecolor=blue,
  urlcolor=blue,
}

\newtheorem{theorem}{Theorem}

\theoremstyle{definition}

\let\leq=\leqslant
\let\geq=\geqslant
\let\le=\leqslant
\let\ge=\geqslant

\newcommand{\classfontshape}{\fontfamily{cmss}\upshape\selectfont}

\newcommand{\classtext}[1]{{\color{classcrimson}\mbox{\classfontshape #1}}\xspace}
\newcommand{\NPcomplete}{\classtext{NP-complete}}
\newcommand{\NPhard}{\classtext{NP-hard}}
\newcommand{\NPhardness}{{\color{classcrimson}\mbox{\classfontshape NP-hard}}ness\xspace}

\makeatletter
\def\blfootnote{\gdef\@thefnmark{}\@footnotetext}
\makeatother

\titleformat{\section}{\sffamily \fontsize{12}{12}\bfseries}{\thesection}{1em}{}
\titleformat{\subsection}{\sffamily \fontsize{10}{10.5}\bfseries}{\thesubsection}{1em}{}

\newcommand{\SI}{{\bf SI}}

\newcommand{\figcaption}[1]{\refstepcounter{figure}\par\vspace{0.3cm}\noindent{\small\textbf{Figure \thefigure:} #1}\par}

\title{\bfseries\sffamily \LARGE Efficiently classifying shocks in complex systems requires dormant reporters}
\date{}

\author[a,b,c]{David A. Brewster}
\author[a,b]{Philippe Cluzel}

\affil[a]{John A. Paulson School of Engineering and Applied Sciences, Harvard University, Cambridge, MA~02138,~USA}
\affil[b]{Department of Molecular and Cellular Biology, Harvard University, Cambridge, MA~02138,~USA}
\affil[c]{Department of Organismic and Evolutionary Biology, Harvard University, Cambridge MA~02138,~USA}

\begin{document}
\maketitle
\onehalfspacing

\vspace{-1.1cm}

\noindent
\textbf{%
Many natural and engineered systems are large complex networks of interacting components, and external perturbations drive them along different dynamical paths.
Identifying which perturbation occurred matters for diagnosis, control, and prediction.
Yet often times only a few components can be jointly monitored.
Which components should be monitored?
Experimental practice usually favors placing reporters at the most sensitive sites, where perturbations produce the largest effects.
Using a simple dynamical model for complex systems with heterogeneous connectivity,
we ask how sparse reporter panels should be chosen to classify shocks from partial trajectories when repeated trials only approximately reproduce an ideal initial condition.
Once that reproduction is imperfect, sensitivity ranked panels fall far short of optimal, and the shortfall grows with the noise.
We find that the best panels mix two kinds of reporters.
A promiscuous reporter responds to most shocks, so it separates them mainly by degree, and degree fluctuates from trial to trial.
A dormant reporter responds to only a few shocks but strongly, and its answers do not scatter as much between trials.
As noise grows, the cost of losing a dormant reporter rises to meet the cost of losing a promiscuous one.
Panels of either kind alone classify worse than the mixture, and no property of the members collected individually explains the ordering.
Most of all, we find that only a minuscule number of reporters are needed on a panel to accurately identify which shock hit the system.
We implement an efficient algorithm to assemble such a panel.
Together these results provide a low cost practical design principle for monitoring large complex dynamical systems.
}

{\bf \blfootnote{Corresponding author e-mail address: dbrewster@g.harvard.edu}}

\section*{Introduction}

Many natural and engineered phenomena can be understood as complex systems composed of many interacting parts.
Ecological communities, gene regulatory networks, neural circuits, economic systems, cities, and social groups all exhibit collective behavior that is not obvious from any one component alone
\cite{kauffman1969metabolic,kauffman1969homeostasis,may1972stable,levin1998ecosystems,newman2003structure,menczer2020first,delriochanona2020supply}.
Network science provides a common framework for relating structure to collective dynamics.
The same ideas have been used to study small world organization, heterogeneous connectivity, spreading, collective behavior, synchronization, and epidemic transitions across physical, biological, and social systems
\cite{watts1998collective,barabasi1999emergence,strogatz2001exploring,albert2002statistical,boccaletti2006complex,pastor2015epidemic,menczer2020first}.
In biology, Boolean network models provide a simple way to study collective dynamics, attractors, robustness, evolvability, and the phase transitions between ordered and chaotic behavior in large regulatory systems
\cite{kauffman1969metabolic,kauffman1969homeostasis,aldana2003boolean,aldana2003natural,oikonomou2006topology,parmer2022influence}.

A central difficulty is that complex systems are almost never observed in full.
Instead, experiments rely on a limited set of observations.
In cell and molecular biology, fluorescent proteins, time lapse measurements, microfluidic platforms such as the mother machine, and transcriptional reporter libraries made it possible to follow gene expression and physiological state in living cells and even in single bacteria
\cite{chalfie1994green,cluzel2000ultrasensitive,golding2005realtime,le2005realtime,wang2010robust,zaslaver2006library,potvintrottier2016synchronous,balleza2018systematic}.
Such experiments also made clear that stochastic gene expression and variability between cells are intrinsic features of living systems
\cite{korobkova2004molecular,paulsson2004summing,paulsson2005models,lestas2010fundamental,hilfinger2011separating}.
In neuroscience, even large scale recordings sample only a small fraction of the underlying circuit, yet collective network states can still be inferred from partial populations.
Low order statistical models can capture large scale resting state brain activity
\cite{schneidman2006weak,watanabe2013pairwise}.
Recordings in behaving animals also show that neurons without a classical stimulus response carry information about the stimulus and about the animal's decision,
and that neurons which are not themselves selective can improve the fidelity of a population code when read out as part of an ensemble
\cite{insanally2019spike,leavitt2017correlated}.
In human systems, smartphones and other passive sensing technologies provide sparse but informative traces of behavior, mobility, and social interaction
\cite{wang2014studentlife}.
More generally, recent work has shown that carefully chosen subsets of nodes can reveal important aspects of global network behavior even when most nodes are unobserved
\cite{hilfinger2016constraints,maclaren2025observing}.
In all of these settings, one faces the same practical question:
if only a few components can be measured, which ones should they be?
In practice the choice is usually settled by prior knowledge, by sensitivity screens, by teleological assumptions, or by what the instruments happen to export
rather than by a quantitative placement criterion (see \SI{} for many examples of this in the scientific literature).

Partial observation is further complicated by noise.
In real systems, repeated trials rarely begin from exactly the same state.
Cells with the same genotype differ because of stochastic gene expression, growth fluctuations, and uncontrolled microenvironmental variation
\cite{paulsson2004summing,paulsson2005models,lestas2010fundamental,wang2010robust}.
In bacterial chemotaxis, the same fluctuations that drive behavioral variability also shape the response to small stimuli, so variability and response cannot be treated as independent properties
\cite{emonet2008relationship,park2010interdependence}.
Neural and behavioral measurements face analogous uncertainties from ongoing background activity, latent variables, finite sampling, and imperfect experimental control
\cite{schneidman2006weak,watanabe2013pairwise,wang2014studentlife}.
As a result, the observed response to a perturbation is inevitably a mixture of perturbation specific effects and variability in the underlying system state.
This kind of uncertainty motivates our focus on noisy initial conditions.

This question becomes sharper when the goal is not just to detect that a perturbation occurred, but to further identify which perturbation occurred.
In bacteria, for example, antibiotics trigger distributed transcriptional and physiological responses that depend on drug identity, dose, and combination
\cite{bollenbach2009nonoptimal,wood2014uncovering}.
Related inference problems arise in other settings as well,
from distinguishing environmental disturbances in ecological or cellular networks to identifying neural or behavioral responses from partial trajectories,
or even anticipating impending tipping points from sparse statistics
\cite{masuda2026tipping}.
Even there the most strongly fluctuating components are not automatically the most informative, because a large signal can carry proportionately large noise,
and rankings by response magnitude degrade when components differ in how noisy they are \cite{masuda2024anticipating}.
Existing approaches often emphasize nodes that are highly connected, highly central, highly influential, or maximally sensitive.
This emphasis is reinforced by the literature on influence maximization and spreading,
including recent extensions to Boolean network control which often seeks nodes that maximize reach or steer the system most effectively
\cite{kempe2003maximizing,kitsak2010identification,morone2015influence,parmer2022influence}.
But the best nodes for spreading influence need not be the best nodes for diagnosing which shock occurred.

Quantitative sensor placement has mostly served other goals than the one we study in this paper.
Classical methods choose measurements to reconstruct or estimate the state of a dynamical system \cite{joshi2009sensor,liu2013observability,summers2016submodularity}.
A related line asks how much of a network becomes structurally observable as sensors are added, and finds that the largest observable component appears through a percolation like transition \cite{yang2016observability}.
A related instinct is to reduce dimensionality first.
Principal component analysis (PCA) and its relatives summarize high dimensional statistics with a few latent coordinates \cite{pearson1901lines,jolliffe2016principal}.
But each latent coordinate is a weighted combination of possibly all of the original variables, so evaluating it still requires measuring the whole system.
A monitoring panel faces the opposite constraint.
It reports a few physical components exactly and says nothing about the rest, so the reduction must be dimension preserving.
Closer to our question, sparse sensors have been optimized to classify static labeled data \cite{brunton2016sparse},
and minimal sensor sets have been derived that discriminate the attractors of a Boolean network under a worst case number of corrupted nodes \cite{cheng2021discrimination}.
Our problem differs in both ingredients and complements this work.
The observations are settled trajectories of a perturbed dynamical system rather than static snapshots of labeled classes.
The uncertainty is stochastic variation in the initial state rather than a worst case bound.

Here we use Boolean threshold networks with heterogeneous connectivity and ask how to choose a small measurement node set to classify shocks when initial conditions are noisy.
Our finding yields a practical design principle for placing reporters in large complex systems.

\section*{Model}

We use a Boolean threshold network with $N$ nodes to model complex system dynamics.
Node $j$ has binary state $\sigma_j(t) \in \{0,1\}$ at discrete time $t$.
When $\sigma_j(t)=1$ the node is in the active state,
and when $\sigma_j(t)=0$ the node is in the inactive state.
The full system state is $\boldsymbol{\mathbf \sigma}(t) = (\sigma_1(t),\dots,\sigma_N(t))$.
Each directed edge $i \to j$ carries a weight $w_{ij} \in [-1,1]$.
All nodes are updated synchronously according to
\begin{equation}
\sigma_j(t+1) =
\begin{cases}
1, & \sum_i w_{ij}\sigma_i(t) > 0,\\
0, & \sum_i w_{ij}\sigma_i(t) < 0,\\
\sigma_j(t), & \sum_i w_{ij}\sigma_i(t) = 0.
\end{cases}
\end{equation}
Thus each node follows the sign of its total weighted input with ties leaving the node unchanged (see \textbf{Figure 1b}).
We use power law out-degree networks which provide a simple and realistic model of heterogeneous connectivity.
Edge weights are drawn independently and uniformly from $[-1,1]$.

To model exogenous shocks, we generate $d$ perturbed versions of the same underlying network.
For each shock, $g$ target nodes are chosen uniformly at random, and every outgoing edge weight of each target is independently redrawn to $-1$ or $+1$ with equal probability.
We use the two extreme values deliberately.
A shock models a dramatic intervention that either fully enables or fully suppresses each affected interaction, in contrast to the finely graded weights of the unperturbed network.
A shocked node therefore acts on its downstream neighbors at full strength with a new random pattern of activation and repression (see \textbf{Figure 1c}).
Each shock produces a distinct pattern of weight changes while leaving the rest of the network unchanged.
The unperturbed network serves as a control.

To model experimental noise, each network is assigned a single base initial condition,
generated by setting each node's state independently to $0$ or $1$ with equal probability.
The base state plays the role of the ideal preparation of an experiment.
In practice no preparation is exact, and each repeated trial starts from a slightly different state.
We quantify this with a noise parameter $\varepsilon \in [0,1]$.
Each replicate initial condition reproduces each component of the ideal state faithfully with probability $1-\varepsilon$ and re-randomizes it otherwise.
At $\varepsilon = 0$ every trial starts exactly from the ideal state.
At $\varepsilon = 1$ the replicates carry no memory of the preparation (see \textbf{Figure 1e}).
For each control and shocked network, we develop the dynamics for $T$ synchronous updates 
and retain the final portion of the trajectory (see \textbf{Figure 1d}).

The observation available to the classifier is not the full network state, but only the states of a selected subset of $m$ nodes, which we call reporters.
For a reporter panel $\{j_1,\dots,j_m\}$, the observed state at time $t$ is
$\mathbf r(t) = (\sigma_{j_1}(t),\dots,\sigma_{j_m}(t))$.
The inference task is to identify which of the $d$ shocks generated an observed reporter trajectory (see \textbf{Figure 1f}).
The unperturbed control is an eleventh alternative, so there are $d+1$ classes.
In practice, we classify shock identity from the retained reporter states using a random forest classifier.
We aggregate predictions across the retained samples to assign a single shock label to each trial
(see \SI{} for details).

The central question of the paper is how to efficiently choose a small reporter panel that maximizes the accuracy with which shocks can be identified.

\section*{Results}

\subsection*{The ensemble spans the transition between order and chaos}

The dynamical regime of the ensemble is controlled by the network's power law degree exponent $\gamma$ and the network size $N$ (\textbf{Figure 2}).
When averaged over possible states, a simple calculation proves that a single flipped threshold input changes the output of a node with in-degree $K$ with probability proportional to $1/\sqrt{K}$.
Thus noise or damage corrupting a state propagates when the mean connectivity exceeds $K_c = 4\pi/3 \approx 4.19$
(see \SI{} Section 2 for details).
Because in-degrees in this ensemble are binomially distributed (see \textbf{Figure 2c}),
$K_c$ fixes the critical exponent $\gamma_c(N)$ at every size.
We measure the steady state Hamming distance deviation when small noise $\varepsilon=0.02$ is introduced to the initial condition.
Our analytic prediction matches simulation results from $\gamma_c \approx 1.68$ at $N=50$ to $\gamma_c \approx 2.09$ at $N=5000$ (see \textbf{Figure 2a}).
Our working ensemble at $\gamma = 1.8$ and $N = 5000$ has mean connectivity $K=12.2$ and sits inside this finite size transition region, on the chaotic side of its boundary.
In this regime trajectories from nearby initial conditions decorrelate at a fraction of about $0.2$ of the nodes, so the dynamics are rich but far from random.
See \SI{} Section 2.5 for more details about the dynamics.

\subsection*{Shock identification from partial dynamics}

For each network we consider $d=10$ shocks per network and $g=50$ targets per shock.
Thus $1\%$ of the network's nodes are targeted.
These shocks do not modify the topology or completely destroy the dynamics.
The basins of attraction change, but the attractors typically remain similar to the control attractors.
The inference problem is to assign an observed partial trajectory to the correct shocked network.
In particular, we consider the problem of assigning one time snapshot of a trajectory to a shock.
We sample trajectories after $T=1000$ time steps and collect $10$ time snapshots.
This far into the dynamics, the trajectory is on or near an attractor.
Since only a small subset of nodes can be measured, different reporter panels preserve different aspects of the dynamics.
See \SI{} Section 3 for details.

\subsection*{Evolved panels beat random placement with $10$ times fewer reporters}

We next asked how the placement of a fixed number of reporters affects classification accuracy.
We compared randomly chosen panels with panels selected by a genetic algorithm to maximize classification performance (\textbf{Figure 3a}).
Every accuracy we report is the score of a finished panel on fresh trials.
The elbow of the evolved accuracy curve sits at $m = 8$ reporters at every noise level,
and beyond it each doubling of the panel size buys a few accuracy points at most.
At $m = 8$, evolved panels reach an accuracy of $1.00$ with a perfectly reproduced initial condition and $0.94$ at $\varepsilon = 1$,
against $0.47$ and $0.35$ (respectively) for random panels and a chance level of $\nicefrac{1}{11}$.
A random panel needs more than $10$ times as many reporters to match an evolved panel when $m=8$.

To understand the importance of the members of high accuracy reporter panels,
we focused on evolved panels with $m=8$ members.
First, the sensitivity of a node $S_j$ is its average Hamming distance from the control, averaged over time and initial conditions.
Evolved panel members with relatively higher sensitivity have an effective number of shocks to which it is sensitive between $8$ and $10$,
whereas evolved panel members with relatively lower sensitivity have an effective number of shocks between $3$ and $7$
depending on the noise.
The effective number of shocks is the number shocks needed if each shock produced the same sensitivities to the node with the equivalent total sensitivity
(see \textbf{Figure 3b}).
The sensitivity of nodes falls into a bimodal distribution:
some nodes have very low sensitivity (near zero) and some nodes have very high sensitivity (near $0.5$).
We consider nodes that reside closer to the latter mode as highly sensitive (see \SI{} Section 5.2 for details)
The fraction of highly sensitive nodes in a reporter panel is at least double as high in the evolved panels versus the random panels for $m=8$.
This difference disappears as $m$ becomes larger (see \textbf{Figure 3c}). 
On average, nodes are around twice as sensitive in the $m=8$ evolved panels versus the remaining nodes that were not selected by the genetic algorithm
(see \textbf{Figure 3d}).

\subsection*{Reporters separate into $3$ disjoint classes}
The composition of the evolved panels explains part of the evolved panel advantage.
Nodes fall into $3$ disjoint classes.
Promiscuous nodes are at least moderately sensitive to most shocks.
Dormant nodes are highly sensitive to only a few shocks and quiet under the rest.
Unresponsive nodes are not very sensitive to any shock.
We call a node responsive when it is either promiscuous or dormant.
A node answers a shock when its deviation reaches the cutoff $\theta$ that separates the two modes of the sensitivity distribution,
and we write $n_j$ for the number of the $d = 10$ shocks that node $j$ answers.
A node is unresponsive when $n_j = 0$, dormant when $1 \leq n_j \leq 5$, and promiscuous when $n_j \geq 6$.
The distribution of $n_j$ is U shaped and its interior minimum falls at this majority split, so the boundary sits in a valley rather than on a slope (\SI{}).
Only promiscuous nodes have high mean sensitivity.
Dormant and unresponsive nodes both have low mean sensitivity, and what separates them is that a dormant node answers at least $1$ shock while an unresponsive node answers none
(see \textbf{Figure 3e,f}).
We say that highly sensitive (or just sensitive) nodes exceed $\theta$ whereas low sensitivity (or just insensitive) nodes do not.

For the $50$ networks we examined, evolved panels of $m=8$ reporters on average contain about $4$ or $5$ promiscuous nodes, against about $2$ expected by chance.
The evolved panels almost never select a purely promiscuous strategy.
We find that $47$ of $50$ evolved panels contain at least $1$ dormant member at $\varepsilon = 1$.
There are barely any unresponsive nodes selected despite them making up roughly half of the network.
In the $8$ out of $50$ panels that have unresponsive nodes, 
$6$ have only $1$ unresponsive node, and $2$ have $2$ unresponsive nodes.
The rest of the panels are filled with dormant nodes (see \textbf{Figure 4}).
See \SI{} Section 5.3 for more details.

\subsection*{Promiscuous and dormant reporters encode complementary information}

To understand the roles of the classes of reporters, we examined how each member of an evolved panel responds to each individual shock (\textbf{Figure 3f}).
The dynamics pin a dormant node's state.
It holds its control value under most shocks and flips completely under a few.
So its reading is the part of the observation that repeats when the trial does not, with about half the trial to trial scatter of a promiscuous reading.
Dormant nodes are not uninformative.
Their states track the coarse dynamical state of the network.
Dormant nodes do not report persisting attractor basins since a shock leaves none of the control attractors intact---%
knowing the control basin says almost nothing about where a trajectory ends up.
Rather, the node level fingerprint of the dynamic survives the shock.
A node that is stable in the control landscape stays stable and near its control value in the shocked one,
and a thin subset of such nodes differs reproducibly across basins and across shocks (see \SI{} Section 3.2).
This pinning is what makes dormant reporter answers repeatable, and it becomes valuable once the initial state is not perfectly controlled.
A trajectory can then differ from trial to trial even before the response to the shock is fully expressed.
Thus a reading that survives that variation is exactly what the classifier can still lean on.

Promiscuous reporters can also have extremely high sensitivities to particular shocks, similar to dormant reporters (see \textbf{Figure 4}).
Though the discrimination power that dormant reporters have is lost when too many shocks cause a significant response.
Thus the roles that promiscuous and dormant reporter play differ.

The information carried by a panel is also structured across members.
Evolved panels have higher marginal entropy per member than random panels matched for sensitivity composition,
$0.81$ to $0.88$ bits against $0.68$ to $0.76$.
At the same time, evolved panel members share less pairwise mutual information.
The genetic algorithm therefore selects members that are individually informative and mutually complementary in the dynamical space rather than in the topological space
(see \SI{} Section 5.3 for more details).

\subsection*{The balance between the two classes shifts with noise}

We quantified the division of labor by removing members from evolved panels (ablation)
and measuring the loss in accuracy across $15$ noise levels,
using every removal subset of up to $m-1=7$ members (\textbf{Figure 5}).
Panel composition itself changes only mildly with noise,
but the number of sensitive nodes decreases with noise (\textbf{Figure 5b}).
At nearly perfect control of the initial state,
removing one sensitive member costs $0.06$ in accuracy, against $0.04$ for one insensitive member.
As noise grows both penalties climb,
to $0.10$ and $0.08$ (respectively) at $\varepsilon = 1$,
because every member matters more once the panel's margins have thinned (\textbf{Figure 5c}).
The insensitive share of the total removal penalty sits near $40\%$ at low noise and rises to $44\%$ at $\varepsilon = 1$ (\textbf{Figure 5d}).
Comparing removal subsets of the same panel and the same size that differ in exactly one member,
the extra cost of removing a sensitive member rather than a member below the sensitivity cutoff peaks in the middle of the panel,
reaching $0.067$ per member at depth $3$ with a perfectly reproduced initial condition against $0.032$ at $\varepsilon = 1$ (\textbf{Figure 5e}).

Resolving the ablations by composition sharpens this picture in $4$ ways (\textbf{Figure 6}).
First, dormant reporters are never dispensable.
Removing only dormant members costs $0.02$ in accuracy at a single removal even with a perfectly reproduced initial condition, and $0.20$ by 3 removals, so panels keep no dead weight (\textbf{Figure 6a}).
Second, the promiscuous premium is spent by the time noise is full.
Within one panel, removing one promiscuous instead of one dormant member costs $0.012$ extra at $\varepsilon = 0$,
about as much at $\varepsilon = 0.02$ and $\varepsilon = 0.5$ though no longer resolved from zero there,
and $-0.001$ at $\varepsilon = 1$ (see \textbf{Figure 6b}).
Under full noise a dormant reporter is worth as much as a promiscuous one.
The premium never exceeds a fifth of the cost of the removal itself,
so what separates the classes is not how much accuracy one member carries but when it carries it.
Third, the two classes are a discretization of a continuous quantity.
The accuracy drop from removing a single member rises with that member's sensitivity $S$ and then plateaus (\textbf{Figure 6c}).
It shows no jump at the class cutoff.
Fourth, the classes interact.
At $\varepsilon = 0$, removing one dormant member from an intact panel costs $0.02$,
but the same removal costs $0.06$ once a promiscuous member is already gone,
nearly a tripling that replicates in every cohort.
Under noise the intact cost is already near $0.07$ and climbs to $0.09$ after a promiscuous loss (\textbf{Figure 6d}).

The dormant members are a reserve.
When trials are reproducible their information is partly redundant with the promiscuous readout, and it becomes load bearing exactly when that readout degrades, whether through noise or through the loss of promiscuous reporters.
Promiscuous nodes dominate when trials are reproducible.
Dormant nodes grow in importance as variability between trials increases.

The dormant members of a panel also avoid answering the same shocks.
Comparing each panel against draws of other dormant nodes of the same network with the same breadth,
the answered sets overlap less than chance at every noise level (mean pairwise Jaccard $0.198$ against $0.235$, $0.129$ against $0.247$, and $0.169$ against $0.267$ respectively),
and cover more distinct shocks.
The effect is modest and the coverage is far from complete---about $7$ of the $10$ shocks---so the dormant members spread out rather than tile.
They do not, however, specialize in the shocks that the promiscuous members confuse.
The fraction of confusable shock pairs a panel's dormant members separate matches the fraction of already separable pairs they separate.
The difference is $+0.001$, $-0.001$, and $+0.027$ at $\varepsilon = 0$, $0.5$, and $1$.
A greedy control built to target the confusable pairs reaches $+0.27$, $+0.13$, and $+0.16$ (see \SI{} for details).
See \SI{} Section 5.3 for more details, experiments, and analyses.

\subsection*{Information based selection nearly matches the search}

Finally, we compared the evolved panels against simpler placement strategies, all scored with the same classifier (\textbf{Figure 7a,b,c}).
Strategies based solely on network topology fail.
Selecting hubs or covering distinct neighborhoods performs no better than random selection at every noise level.
There is a precedent for this kind of failure in spreading processes~\cite{radicchi2017superblockers}.
Selecting the most sensitive nodes is better,
but its gap to the evolved panels widens with noise,
reaching $0.70$ against $0.94$ at $m=8$ reporters under high noise.
The strongest heuristic is greedy information gain,
which adds the reporter that most reduces the empirical uncertainty about the shock given the panel.
At high noise it reaches $0.92$ at $m=8$ reporters, within about $2$ points of the genetic algorithm at a small fraction of the evaluations.

The wiring diagram alone does not reveal which nodes make good reporters.
See \SI{} Section 5.2 for more details.

\subsection*{An efficient high-accuracy selection rule rediscovers the promiscuous-dormant mix}

The classes explain what evolved panels contain, but knowing the classes is not enough to build one.
We tried, scoring $6$ designed panels at full noise with the same classifier (see \SI{} for details).
A panel of $8$ dormant nodes chosen to cover as many distinct shocks as possible covers all $10$ shocks in every network,
yet reaches an accuracy of $0.550$, barely above the $0.507$ of $8$ dormant nodes drawn at random.
Coverage is not identification, because two shocks that trip the same members still cannot be told apart.

What orders the designs is not any property of the mean response profile.
Rather it is how the members' separations stand against their trial to trial variability (see \SI{} for details).
Computing the discriminability of every pair of alternatives from single trials with each member's separation measured against its variance pooled across all conditions,
the low tail of that distribution over pairs orders all $6$ designs exactly as their accuracy does.
Further, the order correlates with accuracy at $\rho = 0.84$ across the $300$ panels, against $0.50$ for the best statistic of the mean response profile.

A dormant reporter is variable exactly where it flips, so judged against its local noise it looks unreliable at its moments of action,
while judged against its overall variability it is a node that barely fluctuates at all.
The classifier's behavior follows the second judgment.
Panel quality is in this sense collective, not because members interact, but because the criterion is a property of the set.
It is the worst separated pairs of the whole panel that bind.
This is why rankings of nodes one at a time fall increasingly short as noise grows
while the selectors that approach the search optimize set level objectives.
The same criterion applied at the set level succeeds at every noise level (\textbf{Figure 7} and \SI{}).

Greedily adding whichever node most raises the panel's low tail discriminability, computed from the same pilot replicates the search consumes, reaches
$0.97$, $0.92$, and $0.90$ at $\varepsilon = 0$, $0.5$, and $1$, within $3$ to $4$ points of the genetic algorithm's $1.00$, $0.95$, and $0.94$,
and it never calls the classifier: the search scores $3000$ panels, the rule scores one.
We call this selector the spring rule: 
choosing the next member is like tying the candidate to every current member with a spring in response space and letting it settle,
so each pick lands at least some minimum distance from all of them.

The panels the spring rule builds mix promiscuous and dormant reporters without fed knowledge that the classes exist (see \textbf{Figure 7d-f}).
It performs on par with greedy information gain, $0.90$ against $0.92$ at full noise.
The two are not the same criterion: conditional entropy is an average over readout patterns, so a panel can score well while leaving one pair of shocks confused, whereas the low tail scores the binding pair.
The spring rule is also the cheaper one to feed, needing only each candidate's mean and spread per condition, about $100$ replicate readouts per node, so candidates can be screened separately.
In contrast, information gain needs the joint pattern of all candidates in the same snapshot and a number of snapshots that doubles with every member added, $2^m$ for a panel of $m$.
With the $1100$ snapshots retained per network that estimate saturates near $m=10$ members,
which is why the entropy strategies stop at $m = 16$ while the rule runs to $m = 32$ and carries over unchanged to continuous readouts.
Ranking nodes by the identical quantity computed one at a time reaches $0.58$ at full noise, so the ingredient is worthless outside the set, and the rule uses that ranking only to prune the candidate pool.

The reporter panel selection problem is NP-complete in the worst case,
thus likely computationally intractable (see \SI{} Section 5.4 for details).
On typical instances of this ensemble polynomial time rules land within a few points of the search that our results are built on.
So the hardness lives in the worst case and not in the problem an experimenter actually faces.

\section*{Discussion}

Our results show that reporter placement in a partially observed dynamical system is not simply a problem of finding the nodes that respond most strongly to perturbation.
That strategy works only when the initial state of the system is reproduced with nearly perfect precision.
Once even modest uncertainty is introduced, classification depends on two kinds of reading that fail in opposite ways.
A promiscuous node reads out most shocks but typically indiscriminately between imperfect trials.
A dormant node reads out few shocks, but the dynamics pin its state so what it says survives the variation.
In this sense the problem is to secure readings that mean the same thing from one trial to the next.

This perspective helps explain why optimized reporter panels differ from panels chosen to maximize perturbation response alone.
In many settings, the most obvious candidate nodes are those that amplify, spread, or most visibly register an external perturbation.
That logic is natural, and it is closely related to ideas from influence maximization and network control.
But diagnosis under partial observation is a different task.
A node may be highly sensitive and still fail to situate a trajectory within the system's broader dynamical landscape.
Our results suggest that successful measurement design must therefore balance information specific to the perturbation against information specific to the underlying state.
The best reporters are the ones that jointly resolve both.

A concrete scenario makes the recommendation tangible.
A city wants to know, from a few monitored sites, which disruption is underway, a marathon, a stadium event, a transit strike, a storm.
Live monitoring traffic intersection volume is affordable at only a handful of sites.
The natural picks are the downtown hubs that react to everything, and they do spread the event types apart, but they mostly report that something is happening, and their day to day variability is largest.
The spring rule we proposed keeps a couple of hubs and then recruits quiet peripheral sites,
the road along the marathon route,
the street by the stadium,
sites that sit at a predictable level on almost every day,
so that on the rare days they do move the reading is unambiguous.
The watch list that identifies the disruption is not the list of busiest intersections.
Even in the case when all components of a system are available for live tracking---city bike ridership activity per station, for example---%
reducing analysis and identification to a few components of the system can lead to more explainable predictions and more efficient analyses.
This sort of dimensionality reduction is widespread among the sciences in the form of PCA and related techniques.

The Boolean threshold model is intentionally minimal, and that simplicity is both a strength and a limitation.
It lets us isolate a clear principle without committing to the details of a particular biological, neural, or behavioral mechanism.
At the same time, real systems can have graded responses, asynchronous updates, multiple timescales, structured noise, and perturbations that act through more specific pathways than the shocks considered here.
The literature has made clear that discrete network dynamics depend strongly on modeling choices.
Synchronous, asynchronous, and stochastic update schemes can generate different attractors and different perturbation responses
\cite{garg2008synchronous,albert2014boolean,park2026succession}.
There are approaches that expose attractor structure and control logic directly from logical models,
and recent work has extended these ideas to asynchronous attractor landscapes
\cite{zanudo2015cell,park2026succession}.
More broadly, recent reviews and assessments have emphasized that the dynamical repertoire supported by a regulatory network depends not only on topology,
but also on the choice of logical functions, threshold formalism, multivalued rules, and continuous realizations
\cite{gedeon2024network,kadelka2025critical}.
We therefore use a Boolean threshold model not as a claim of unique realism
but as a deliberately simple and analyzable setting in which to study measurement design under partial observation.

The promiscuous and dormant labels should therefore not be viewed as immutable properties attached to nodes in isolation,
but rather as properties of a node relative to a perturbation ensemble, a noise level, and an inference objective.
We also did not consider classifiers that operate directly on time-spanning windowed trajectories rather than single timepoint snapshots.
Our setting is closer to snapshot data, which makes the framework generic across experimental systems where only sparse observations are available.
When full single cell time series can be collected, however, trajectory-based representations may support stronger classification, as recent work in \textit{E. coli} suggests \cite{fraisse2025representation}.
Extending these ideas to models with continuous states, experimentally constrained perturbations, and adaptive or time dependent measurement schemes will be an important next step.

Even with these caveats, the main message is clear.
Sparse measurements work best when they combine nodes that respond strongly with nodes that preserve global dynamical context.
We expect the same principle to hold in experimental settings such as bacterial drug response, neural recordings, and human behavioral data.
Testing this prediction in those systems is a natural next step.
Our results suggest practical strategies for reporter design in systems biology, neuroscience, and behavior, especially when measurement budgets are limited and variability between trials cannot be ignored.
More broadly, our findings suggest a simple principle for observing complex systems:
to infer which shock occurred, one must also know what state the system was in when the shock arrived.

\hfill

{\small

\section*{Methods}

\noindent%
Here we specify the simulation, classification, and selection procedures.
The \SI{} contains full derivations, extended analyses, and implementation notes.

\medskip
\noindent\textbf{Network ensemble.}
All experiments use $50$ independent Boolean threshold networks with $N = 5000$ nodes.
Out-degrees are drawn from the power law $P(k) \propto k^{-\gamma}$ with $\gamma = 1.8$,
truncated at $k = N$, which gives mean connectivity $K \approx 12.24$.
Each node connects to distinct targets chosen uniformly at random (self-connections allowed).
Each edge weight is drawn uniformly from $[-1,1]$.
States update synchronously by the threshold rule in the Model section.

\medskip
\noindent\textbf{Shocks.}
Each network receives $d=10$ independent shocks (a.k.a. perturbations).
A shock selects $g = 50$ target nodes uniformly at random ($1\%$ of the network)
and redraws every outgoing weight of each target to $-1$ or $+1$ with equal probability.
The unperturbed network is the control, so the classification task has $d+1 = 11$ classes.
Target nodes and redrawn weights are fixed once per network and shared across all trials.

\medskip
\noindent\textbf{Trajectories.}
For each network and each shock,
dynamics run from $n_{\textrm{IC}}=10$ initial conditions:
the base state itself and $9$ noisy copies at noise $\varepsilon$.
Each trajectory runs for $T = 1000$ synchronous updates and the last $L = 10$ states are retained.
$n_{\textrm{IC}}$ counts initial condition replicates, not distinct initial condition states.

\medskip
\noindent\textbf{Classifier.}
A reporter panel of size $m$ reduces each retained snapshot to an $m$ bit observation.
We train a random forest classifier of $100$ trees \cite{breiman2001random},
implemented in \texttt{scikit-learn} \cite{pedregosa2011scikit} with default settings otherwise.
We took the retained last $L$ states, per each of the $n_{\textrm{IC}}$ trajectories associated with each noisy trial, and per each of the $d$ shocks and control (i.e., $d+1$ labels).
Then each of these $s:=L\times n_{\textrm{IC}}\times (d+1)=10\times 10\times 11=1100$ snapshots is individually placed into a basket along with its label.
Then half of the items in the basket are randomly removed and stowed for the training data,
such the number of training representatives is the same for each shocks (and the control).
Equivalently, the training data is assembled by removing $s/(2(d+1))=50$ snapshots per class from the basket, sampled uniformly without replacement.
The snapshots that remain in the basket are withheld for testing.
The random forest classifier predicts each withheld snapshot of a class,
and the plurality of those predictions is the label assigned to that class.
The accuracy of one such split is the fraction of the $d+1 = 11$ classes labeled correctly.
Accuracies are averaged over $10$ independent train and test splits.
Additionally,
we found that sampling the train and test split so that no initial condition has a snapshot that appears in both training and testing data
does not significantly affect classification performance.
See \SI{} for more details.

\medskip
\noindent\textbf{Genetic algorithm.}
Panels are evolved separately for each network and each panel size $m$.
An individual is a set of exactly $m$ distinct nodes, and its fitness is the classifier accuracy from $10$ independent train and test splits.
Fitness guides the search only.
Every accuracy reported in this paper is the score of a finished panel on fresh random splits.
The best fitness met during optimization is the maximum of a noisy score and can overstate the panel.
Fitness is accuracy.
The population holds $100$ individuals and evolves for up to $30$ generations. %
The top $10\%$ of the population (i.e., the elites) carry over into the next generation unchanged, with their fitnesses re-evaluated per generation.
For the remaining population, there are two modes of reproduction.
Parents are sampled with probability proportional to $e^{r/\tau}$, where $r$ is fitness and $\tau = 0.1$ is the selection temperature.
Both modes of reproduction draw members by one rule.
With probability $0.9$ the draw is a uniformly random member of an independently sampled parent.
Otherwise it is a uniformly random node.
Half of the non-elite offspring are built by crossover.
Crossover fills all $m$ member slots by this rule.
The other half copy a single parent.
Each of their members is then replaced with probability $0.1$, using the same rule.
Duplicate members are refilled with distinct random nodes to ensure each individual always consists of $m$ members.
Thus, there is no selection pressure to have fewer members in a panel.
The evolved panel is the fittest individual of the final generation, and the analyses use panels with exactly $m=8$ distinct members.

\medskip
\noindent\textbf{Baselines.}
Random panels draw $m$ nodes uniformly without replacement, $10$ panels per network and size.
Heuristic panels take the first $m$ nodes of a ranking:
by sensitivity,
by in-degree,
by out-degree,
by greedy minimum mean squared error (MMSE) column selection on the control covariance,
or by greedy selection of the node whose downstream target set is most dissimilar, in Jaccard distance, from the targets already covered.
Ranking ties are resolved by seeded random jitter and repeated draws are averaged.
All strategies are scored with the same classifier. %

\medskip
\noindent\textbf{Selection rules.}
Greedy information gain adds the node that most reduces the conditional entropy of the shock label given the panel.
The entropy is a plug in estimate over the retained snapshots.
The spring rule adds the node that most raises the panel's low tail discriminability.
That quantity is the $10^{\text{th}}$ percentile, over all pairs of conditions, of the Mahalanobis separation between condition means.
The separation uses the within condition covariance pooled across conditions, with shrinkage toward its diagonal.
Candidates are pre screened to the $400$ nodes of highest mean single node separation.
Neither rule calls the classifier while it builds a panel.
Both are scored afterwards with the same classifier as every other strategy.

\medskip
\noindent\textbf{Sensitivity.}
The deviation of node $j$ under shock $q$ is $\Delta_{j,q} = \bigl\langle |\sigma_j^{(q)} - \sigma_j^{(0)}| \bigr\rangle$, the mean absolute difference between its shocked state $\sigma_j^{(q)}$ and its control state $\sigma_j^{(0)}$, averaged over all replicates and all retained snapshots.
The sensitivity of node $j$ is the mean deviation over the $d=10$ shocks, $S_j = \langle \Delta_{j,q} \rangle_q$.
The cutoff $\theta$ that separates the two modes is placed at the minimum of the smoothed histogram pooled over all $50$ networks, searched on $[0.05, 0.40]$. %
Node $j$ answers shock $q$ when $\Delta_{j,q} \geq \theta$, and $n_j$ counts the number shocks it answers.
Nodes with $n_j = 0$ are unresponsive,
those with $1 \leq n_j \leq 5$ are dormant,
and those with $n_j \geq 6$ are promiscuous.
The $3$ classes are disjoint, and we call a node responsive when it is either promiscuous or dormant (see \textbf{Figure 3e}).
A highly sensitive (or just sensitive) node has (average) sensitivity that exceeds $\theta$ whereas a low sensitive (or just insensitive) node does not.

\medskip
\noindent\textbf{Ablation.}
For every evolved panel with exactly $8$ distinct members we re-estimate the baseline accuracy, remove every subset of $1$ to $7$ members, and retrain the classifier on the remaining members.
Over the $50$ networks, each evaluation averages $30$ independent train and test splits.
The penalty of a subset is the baseline accuracy minus the ablated accuracy.

\medskip
\noindent\textbf{Reproducibility.}
Random seeds are assigned deterministically per network, perturbation, and replicate.
The submission scripts in the code repository record every invocation and its parameters.

\medskip
\noindent\textbf{Code availability.}
All simulations were performed using \texttt{Rust}.
All numerical calculations, optimization algorithms, machine learning, data analyses, plotting, and visualizations were performed using \texttt{Python}.
Our code is available at the following web address: \url{https://github.com/davidb2/boolean-threshold-network/}.

\medskip
\noindent
{\bf Use of generative artificial intelligence (AI).} \\
We used AI assistants (Anthropic Claude, the Opus 5 and Fable 5 models)
for help with coding, data analysis, writing the code that generates the data figures, literature search, and editing of the text.
No figure contains AI generated imagery.
Simulation code in \texttt{Rust} was manually written,
with the exception of one line of code that an AI assistant added that fixed a critical bug.
All research questions, hypotheses, modeling choices, and design decisions are the authors' own.
The authors reviewed and verified all AI assisted output and take full responsibility for the content of this work.

\medskip
\noindent
{\bf Acknowledgments.} \\
D.A.B. is supported by a Harvard Graduate School of Arts and Sciences Prize Fellowship.
The authors are grateful for Harvard's FASRC clusters where much simulation and analyses were performed.

\medskip
\noindent
{\bf Author contributions.} \\
D.A.B. contributed ideas, designed and implemented the model and the simulations, performed the experiments and the analyses, made the figures, and wrote the manuscript.
P.C. conceived the study, supervised the project, and secured funding.
All authors discussed the results and approved the final manuscript.

\medskip
\noindent
{\bf Competing interests.} \\
The authors declare no competing interests. \\

}

\clearpage
\singlespacing
\newgeometry{top=1cm,left=1cm,right=1cm,bottom=1.5cm}

\noindent
\begin{center}
\includegraphics[width=0.88\textwidth]{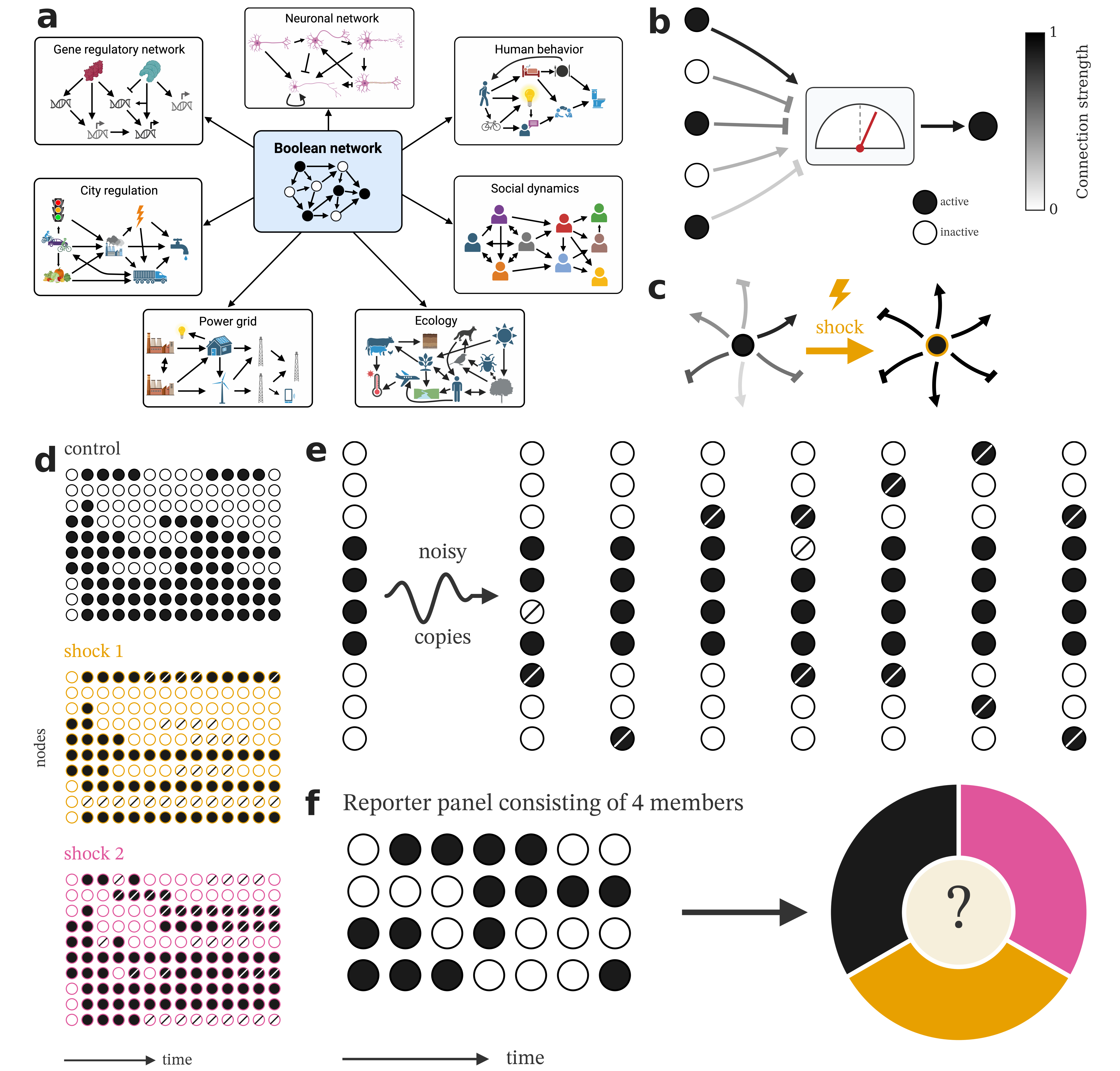}
\end{center}
\figcaption{
\textbf{Shock classification from partially observed dynamics in a Boolean threshold network.}
\textbf{a}, Several examples of domains where complex systems can be modeled as a Boolean network.
Created in BioRender \url{https://BioRender.com/46jpe5e}.
\textbf{b}, Our Boolean networks use threshold updating. 
Input nodes can either be active or inactive.
Connections can be enabling (positive weight, pointed arrowhead) or repressing (negative weight, flat arrowhead), each with varying strengths.
The active input nodes are weighed by their connections and summed.
If the sum is positive (negative), the output of the threshold gate becomes active (inactive, respectively).
At each time step, each node in the network synchronously updates their state using this threshold rule.
\textbf{c}, A node targeted during a shock.
A shock selects target nodes and independently redraws each of their outgoing weights to $-1$ or $+1$ with equal probability.
One active target node is drawn twice with $6$ of its outgoing edges, shaded as in \textbf{b}.
Before the shock the edges carry mixed signs and magnitudes.
After it every edge has magnitude $1$ and some of the signs have switched.
\textbf{d}, Dynamics of a network with $10$ nodes under the rule in \textbf{b}, for a control and two shocked copies of it, all started from the same initial condition.
Circles are filled for active nodes and open for inactive ones, each copy is outlined in the color of its shock, and a diagonal marks every state that differs from the control at that time.
The differences grow with time.
\textbf{e}, Noisy copies.
The initial condition of \textbf{d}, then $7$ noisy copies of it at $\varepsilon = 0.3$, where each node is flipped independently with probability $\varepsilon/2$.
A diagonal marks every node that flipped.
One trial starts from the initial condition itself and the rest start from copies like these.
\textbf{f}, The inference task.
A panel of $4$ reporters is watched for $7$ time steps in one noisy trial, and a classifier must decide which of the $3$ alternatives produced it: the control, shock $1$, or shock $2$.
}
\label{fig:main-1}
\clearpage

\noindent
\begin{center}
\includegraphics[width=0.98\textwidth]{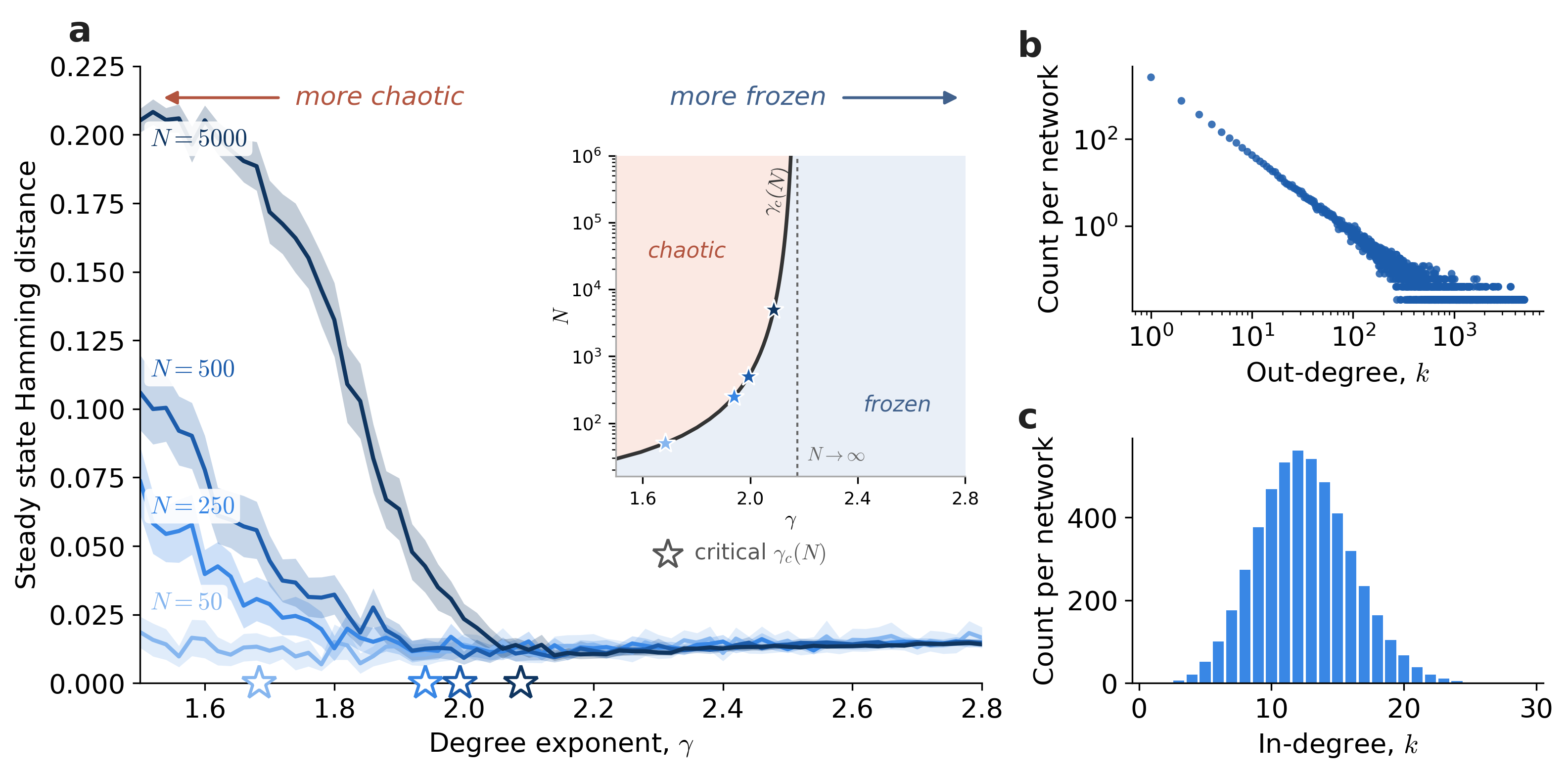}
\end{center}
\figcaption{
\textbf{Phase transitions between order and chaos.}
\textbf{a}, Steady state Hamming distance between trajectories started from imperfect copies of the same initial condition at noise $\varepsilon = 0.02$, as a function of the degree exponent $\gamma$ for $4$ different network sizes.
Each network contributes every pair of its $10$ replicates.
Curves show the mean over $100$ networks per point with $95\%$ confidence bands.
White stars on the horizontal axis mark the critical exponent $\gamma_c(N)$ predicted by the annealed calculation in the \SI{}, where criticality occurs when the mean connectivity crosses $K_c \approx 4\pi/3$.
In the inset panel, the phase diagram is in the $(\gamma, N)$ plane.
The boundary curve $\gamma_c(N)$ separates the chaotic and frozen regions and approaches $\gamma_c(\infty) \approx 2.17$ as $N \to \infty$.
Stars mark the same $4$ network sizes as in panel \textbf{a}.
\textbf{b}, Out-degree distribution pooled over $50$ simulated networks at $\gamma = 1.8$, $N = 5000$.
\textbf{c}, The corresponding in-degree distribution is narrow and binomial.
}
\label{fig:main-2}
\clearpage

\noindent
\begin{center}
\includegraphics[width=0.95\textwidth]{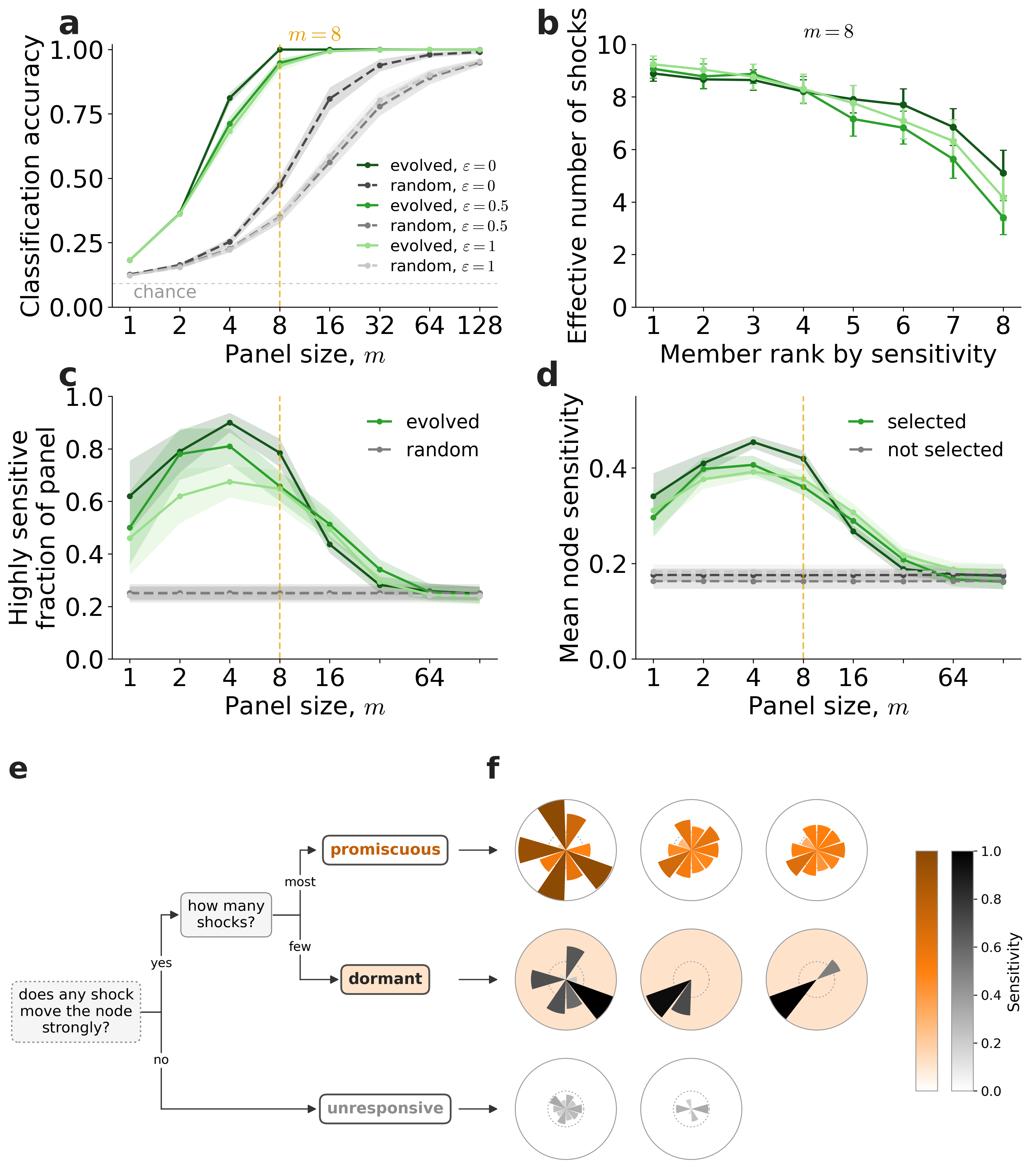}
\end{center}
\figcaption{
\textbf{The elbow of the accuracy curve sits at $8$ reporters, and evolved panels mix $3$ classes of them.}
\textbf{a}, Classification accuracy versus panel size $m$ for panels selected by the genetic algorithm and for random panels, at noise $\varepsilon = 0$, $0.5$, and $1$.
The elbow of the evolved curve sits at $m = 8$ at every noise level, and a random panel needs more than $10$ times as many reporters to match an evolved panel of $8$.
\textbf{b}, The effective number of shocks each member responds to, by member rank, across all $50$ networks.
It falls from about $9$ for the most sensitive member to between $3$ and $5$ for the least, with darker shades lower noise as in \textbf{a}.
\textbf{c}, The highly sensitive fraction of the panel versus $m$, against the random expectation.
Darker shades are lower noise, as in \textbf{a}.
\textbf{d}, Mean sensitivity of selected and unselected nodes versus $m$.
Enrichment is strongest for small panels and dissolves by $m \approx 64$.
\textbf{e}, How a node is assigned a class.
Moving strongly means the deviation reaches the sensitivity cutoff $\theta$ defined in Methods.
A node that no shock moves strongly is unresponsive.
A node that most shocks move strongly is promiscuous.
A node that only a few shocks move strongly is dormant.
An arrow runs from each leaf to the row of \textbf{f} it classifies. %
\textbf{f}, The $m = 8$ members of one evolved panel at high noise ($\varepsilon = 1$), grouped by class and ordered by mean sensitivity within each group.
Each disk is one member and each wedge is one of the $10$ shocks, with wedge length and color depth equal to the sensitivity of that member to that shock. %
The dotted circle inside each disk marks the sensitivity cutoff $\theta$, so a wedge that crosses it is an answered shock.
Promiscuous members respond to nearly every shock and differ only in how strongly, so their profiles overlap heavily.
Dormant members, drawn on a translucent orange ground, concentrate a large response in $1$ or $2$ shocks and stay quiet under the others.
The $2$ unresponsive examples show wedges under many or a few shocks, none of which reaches the cutoff.
}
\label{fig:main-3}
\clearpage

\noindent
\begin{center}
\includegraphics[width=0.88\textwidth]{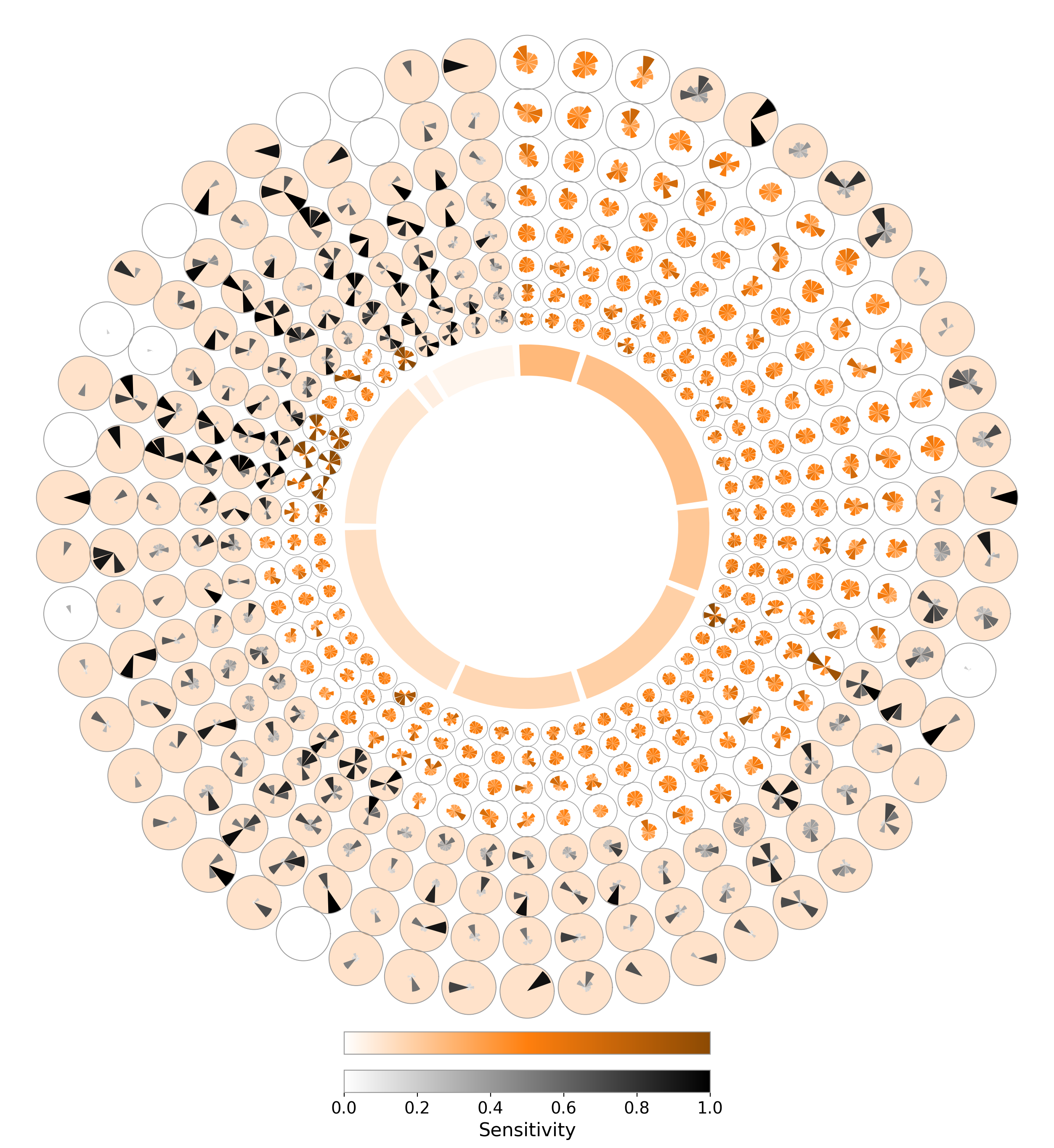}
\end{center}
\figcaption{
\textbf{Nearly every evolved panel mixes promiscuous reporters with dormant reporters.}
Each of the $50$ rays is one network's evolved panel of $8$ reporters at high noise ($\varepsilon = 1$).
Each small disk shows one member's sensitivity to the $10$ shocks, as in Figure 3f.
Each wedge is one shock, and wedge length and color depth both give the sensitivity of that member to that shock, on the shared scale below.
The classes are those defined in Figure 3e from a node's response across shocks.
Promiscuous members respond to nearly every shock and are drawn with orange wedges on white.
Dormant members concentrate their response in only a few shocks and are drawn with black wedges on a translucent orange ground. %
Unresponsive members respond to no shock and appear as nearly empty disks on white.
Across the $400$ members shown, $54.3\%$ are promiscuous, $42.5\%$ dormant, and $3.3\%$ unresponsive.
Rays are ordered by the number of promiscuous members, which therefore is non-increasing as the figure is read clockwise from the top, ranging from $8$ to $0$.
The inner ring encodes the same count: the intensity of orange is proportional to it, and a gap marks each change in it, so arc length shows how many networks share each composition.
Disk size grows with radius for layout only and carries no meaning.
Of the $50$ panels, $47$ contain at least $1$ dormant member.
The other $3$ contain none.
Of the same $50$, $4$ contain no promiscuous member.
}
\label{fig:main-4}
\clearpage

\noindent
\begin{center}
\includegraphics[width=0.92\textwidth]{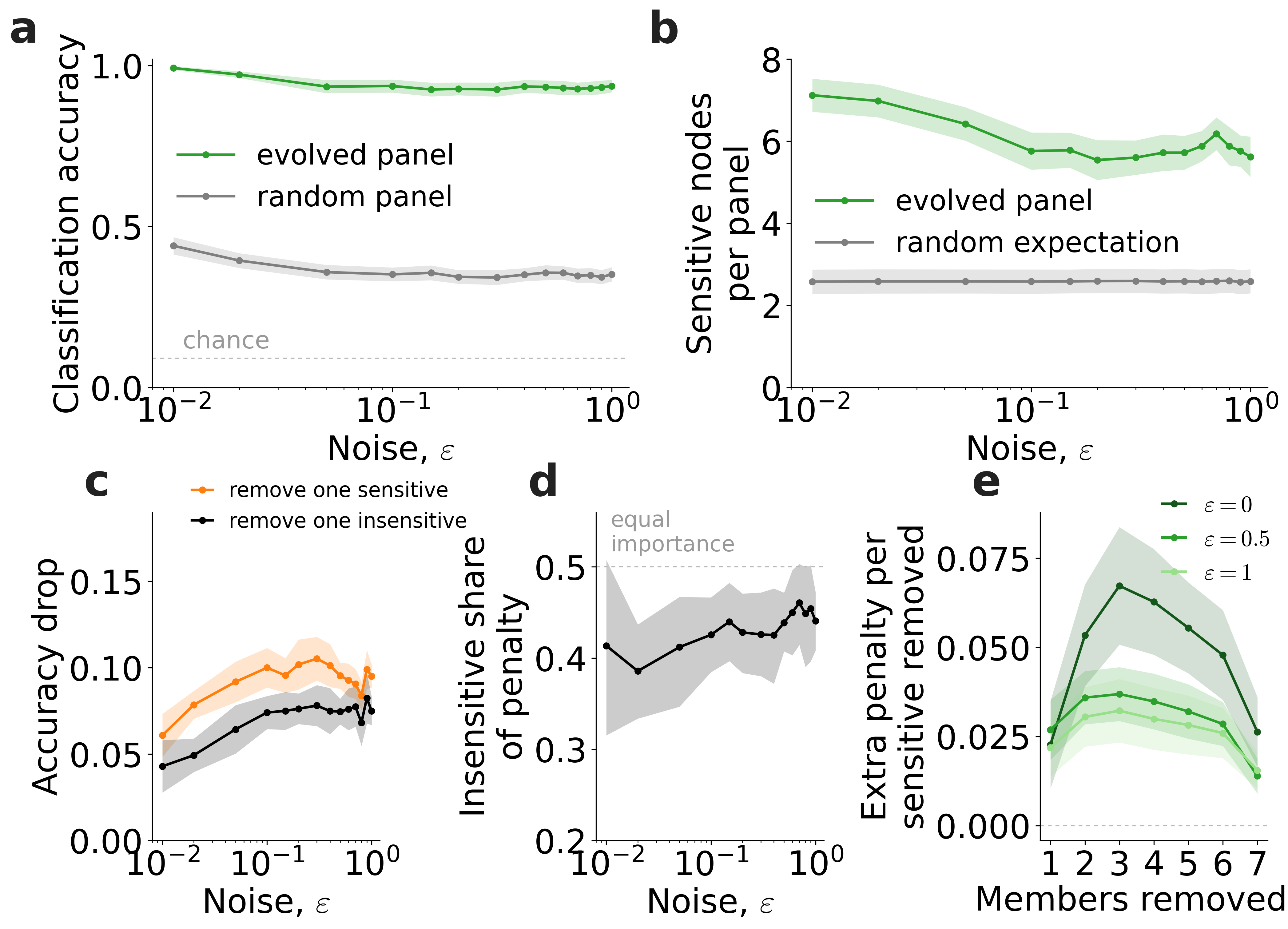}
\end{center}
\figcaption{
\textbf{Evolved panels nearly always mix sensitive and insensitive reporters, and the balance of their importance shifts with noise.}
\textbf{a}, Classification accuracy of evolved and random panels of $8$ reporters as a function of the initial condition noise $\varepsilon$.
\textbf{b}, Number of sensitive nodes per evolved panel against the random expectation.
\textbf{c}, Accuracy drop after removing one sensitive or one insensitive member from an evolved panel.
Both penalties climb as noise grows, and losing a sensitive member costs more at every level.
\textbf{d}, The low sensitivity share of the total removal penalty slightly rises with noise, suggesting that the curves from \textbf{c} are getting slightly closer to each other.
\textbf{e}, The extra penalty per sensitive node removed, estimated within each panel at each removal depth, at noise $\varepsilon = 0$, $0.5$, and $1$.
The premium peaks at intermediate depth and is largest with a perfectly reproduced initial condition.
}
\label{fig:main-5}
\clearpage

\noindent
\begin{center}
\includegraphics[width=0.98\textwidth]{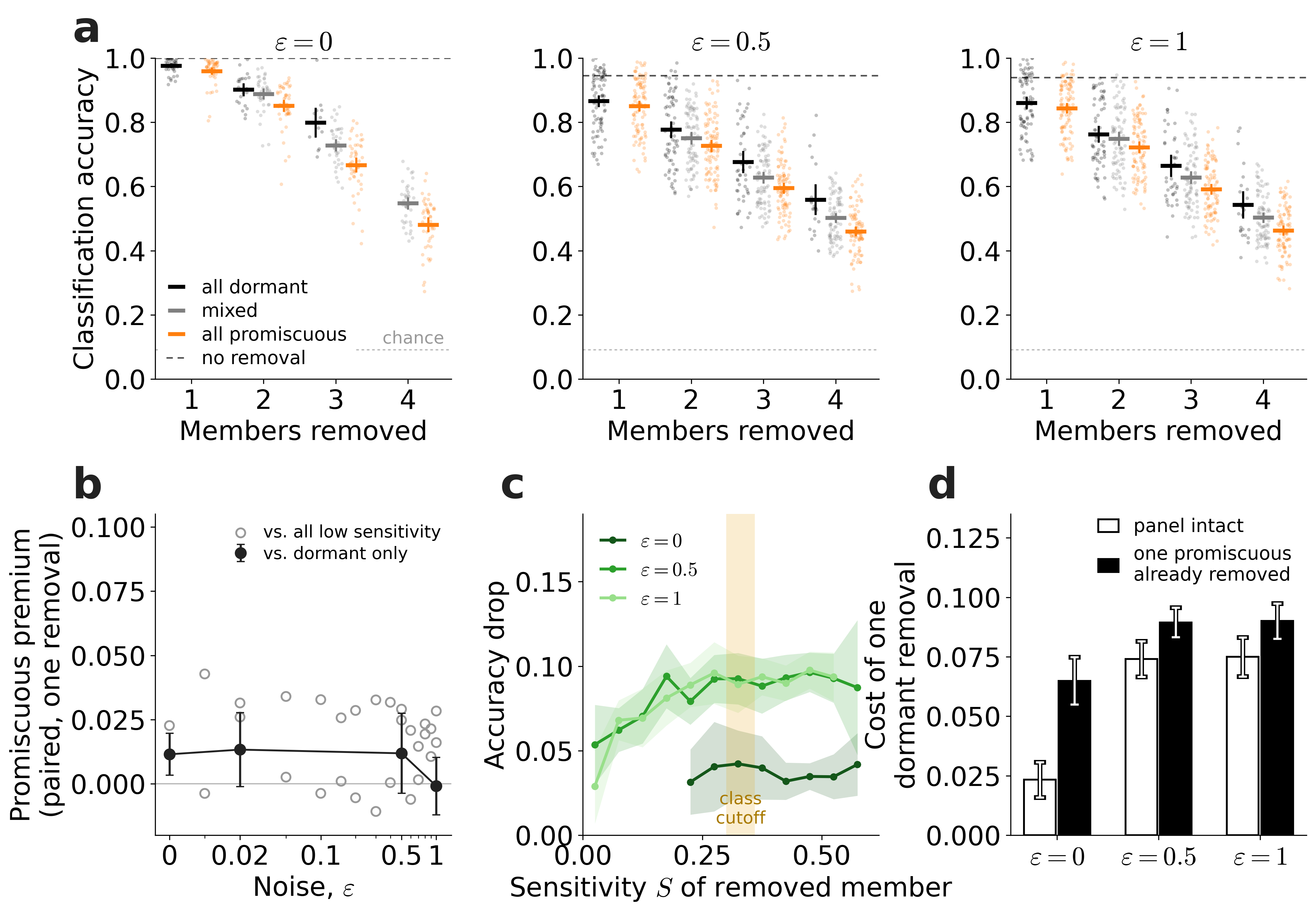}
\end{center}
\figcaption{
\textbf{What it costs to remove promiscuous versus dormant reporters, and how noise changes the answer.}
A continuation of Figure \ref{fig:main-5}, but now looking at the specific reporter classes.
\textbf{a}, Classification accuracy after removing $1$ to $4$ members from evolved $8$ member panels, at noise $\varepsilon = 0$, $0.5$, and $1$.
Each removal subset is classified by its composition: all removed members promiscuous, all dormant, or mixed.
Small points are per panel means, bars are class means with $95\%$ confidence intervals, the dashed line is the unablated baseline, and cells backed by fewer than $10$ panels are not drawn.
Removing only dormant members is never free.
At $\varepsilon = 0$ and one removal it already costs $0.02$ in accuracy, and $0.20$ by $3$ removals.
\textbf{b}, The premium for removing one promiscuous instead of one dormant member, paired within each panel.
Open circles compare promiscuous members against every member below the sensitivity cutoff (insensitive), which is available at all $15$ noise levels but pools unresponsive members into the reference and so slightly overstates the premium.
Filled circles compare them against genuinely dormant members only, at the $4$ noise levels for which per shock deviations were retained, with $95\%$ confidence intervals across panels.
The noise axis is logarithmic above $0.02$ and linear below, so the zero point can be shown.
Corrected this way the premium is $0.012$ at $\varepsilon = 0$ and stays near that level. %
Under full noise a dormant reporter is worth nearly as much as a promiscuous one, something that the sensitivity analysis in Figure \ref{fig:main-5} could not capture.
\textbf{c}, The accuracy drop from removing a single member as a function of that member's sensitivity $S$.
The accuracy drop rises with $S$ and then quickly plateaus. %
\textbf{d}, %
Cost of removing one dormant member from an intact panel (white bars) and after one promiscuous member is already gone (black bars).
At $\varepsilon = 0$ the dormant cost nearly triples once a promiscuous member is lost.
Under noise the dormant members carry that load from the start.
Bars are means with $95\%$ confidence intervals across panels.
}
\label{fig:main-6}
\clearpage

\noindent
\begin{center}
\includegraphics[width=0.98\textwidth]{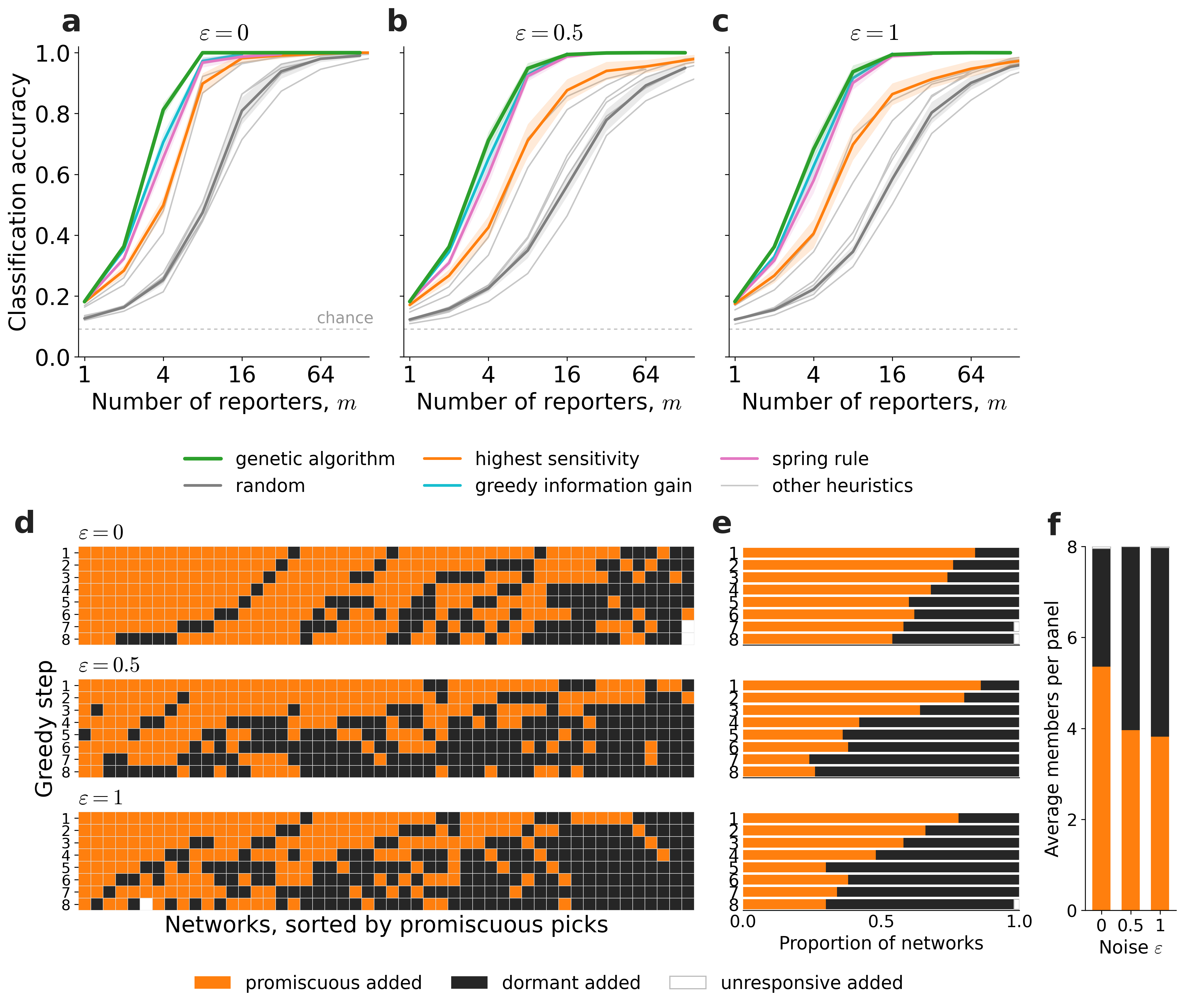}
\end{center}
\figcaption{
\textbf{The genetic algorithm sets the accuracy ceiling, set level rules nearly reach it, and single node rankings fall behind.}
Classification accuracy versus panel size for the genetic algorithm, random selection, $8$ placement heuristics, and the spring rule at noise $\varepsilon = 0$ (\textbf{a}), $0.5$ (\textbf{b}), and $1$ (\textbf{c}).
All strategies are scored with the same random forest evaluator, and every curve reports the accuracy of finished panels on fresh trials. %
The genetic algorithm reaches $1.00$, $0.95$, and $0.94$ at $m = 8$.
Greedy information gain and the spring rule follow within $2$ to $4$ points at a small fraction of the evaluations. %
Selecting the most sensitive nodes is next, with a gap that widens with noise.
Structural strategies perform like random selection or worse.
The thin gray curves are the remaining heuristics:
degree based selection,
greedy minimum mean squared error (MMSE) reconstruction,
Jaccard coverage,
upstream coverage,
and entropy with a diversity penalty.
The full labeled curves for every heuristic are in the \SI{}.
\textbf{d}, Which class the spring rule adds at each greedy step, at $\varepsilon = 0$, $0.5$, and $1$, top to bottom.
Each column is one network, sorted left to right by its number of promiscuous picks, each row is one step of the greedy, and each cell is colored by the class of the reporter added, orange for promiscuous, black for dormant, white for unresponsive.
\textbf{e}, The proportion of networks whose pick at each step is promiscuous, dormant, or unresponsive, row aligned with the grids.
The spring rule buys broad coverage first and then recruits dormant reporters.
This is the same shift toward dormant reporters that the genetic algorithm discovers.
Unresponsive nodes are essentially never chosen, $3$ picks out of $1200$.
\textbf{f}, The average composition of the $8$ member panels the spring rule builds.
Promiscuous members fall from $5.4$ of $8$ with a perfectly reproduced initial condition to $4.0$ at $\varepsilon = 0.5$ and $3.8$ at $\varepsilon = 1$. %
The selection problem is NP-complete in the worst case (see \SI{} for details),
yet on typical instances of this ensemble polynomial time rules land within a few points of the search.
}
\label{fig:main-7}
\clearpage

\clearpage
\setcounter{figure}{0}
\setcounter{table}{0}
\setcounter{equation}{0}
\setcounter{section}{0}
\renewcommand{\thefigure}{S\arabic{figure}}
\renewcommand{\thetable}{S\arabic{table}}
\renewcommand{\theequation}{S\arabic{equation}}

\begin{center}
{\sffamily\bfseries\LARGE Supplementary Information}\\[0.7em]
{\large Efficiently classifying shocks in complex systems requires dormant reporters}\\[0.6em]
{David A. Brewster and Philippe Cluzel}
\end{center}
\vspace{0.5em}

\clearpage
\newpage
\tableofcontents
\clearpage
\newpage

\part{The Problem}

\section{Complex networks}

This Supplementary Information (\textbf{SI}) gives the full model definitions, algorithmic details,
and mathematical notes behind the main text.
Boolean threshold network models are used here as deliberately minimal descriptions
of interacting systems with many degrees of freedom, finite state spaces, and rich attractor
structure \cite{kauffman1969metabolic,aldana2003boolean,albert2014boolean}.
The goal of this supplementary text is to explain the dynamics in more detail
and make explicit which features of the setup matter for the paper's main claims.

\subsection{Network structure}

We model a complex system as a directed weighted network on $N$ nodes and some edges between the nodes.
Node $j$ has binary state $\sigma_j(t) \in \{0,1\}$ at discrete time $t$, and the full state of the system at time $t$ is
\begin{equation}
\bm \sigma(t) = \bigl(\sigma_1(t),\dots,\sigma_N(t)\bigr).
\end{equation}

\subsubsection{Topology}
The in-degree of a node $j$ is the number of nodes $i$ such that $i\to j$ is an edge.
The out-degree of a node $i$ is the number of nodes $j$ such that $i\to j$ is an edge.
For power law networks, the out-degree is sampled from a power law distribution with exponent $\gamma$,
\begin{equation}
P\!\left(K^{\mathrm{out}}_i = k\right)
= \frac{k^{-\gamma}}{\zeta_N(\gamma)},
\qquad
k=1,\dots,N,
\end{equation}
where
\begin{equation}
\zeta_N(\gamma) = \sum_{k=1}^{N} k^{-\gamma}
\end{equation}
is the truncated Riemann zeta function.
There are $K^{\mathrm{out}}_i$ nodes $j_1,\ldots,j_{K^{\mathrm{out}}_i}$ selected uniformly at random without replacement from the network
and connected to $i$ such that $i\to j_{\ell}$ is an edge for $\ell=1,\ldots,K^{\mathrm{out}}_i$.
The corresponding expected out-degree is
\begin{equation}
\left\langle K^{\mathrm{out}}_i \right\rangle
=\sum_{k=1}^N P( K_i^{\mathrm{out}} =k)\cdot k
= \frac{\zeta_N(\gamma-1)}{\zeta_N(\gamma)}.
\end{equation}
The in-degree distribution is $\mathrm{Binomial}(N, p)$ where
\begin{align}
  p
  &= \sum_{k=0}^N P(i\to j \mid K_i^{\textrm{out}}=k)\cdot P(K_i^{\textrm{out}}=k) \\
  &= \sum_{k=0}^N P(K_i^{\textrm{out}}=k) \cdot \binom{N-1}{k-1}\Bigg/\binom{N}{k} \\
  &= \frac{1}{N}\sum_{k=0}^N P(K_i^{\textrm{out}}=k)\cdot k \\
  &= \langle K_i^{\textrm{out}}\rangle / N.
\end{align}
Thus we have
\begin{equation}
  \langle K_i^{\mathrm{in}} \rangle = Np = \langle K_i^{\mathrm{out}} \rangle.
\end{equation}
Further, we can denote
\begin{equation}
  K := \langle K_i^{\mathrm{out}} \rangle = \langle K_j^{\mathrm{in}} \rangle
\end{equation}
for all $i,j$ as the expected degree without reference to in- or out-degrees and without reference to a specific node.

\subsubsection{Weights}
Each directed edge $i \to j$ has a weight $w_{ij} \in [-1,1]$.
Each edge weight is sampled independently from the uniform distribution on $[-1,1]$.

\subsection{Regulatory dynamics}

Dynamics are updated synchronously.
Given state $\bm\sigma(t)$, node $j$ receives the total weighted input
\begin{equation}
\sum_i w_{ij}\cdot \sigma_i(t),
\end{equation}
where the sum ranges over all $i$ such that $i\to j$ is an edge in the network.
The next state is described by the threshold rule
\begin{equation}
\sigma_j(t+1) =
\begin{cases}
1 &\text{if } \sum_i w_{ij}\cdot \sigma_i(t) > 0,\\
0 &\text{if } \sum_i w_{ij}\cdot \sigma_i(t) < 0,\\
\sigma_j(t) &\text{if } \sum_i w_{ij}\cdot \sigma_i(t) = 0.
\label{eq:threshold-updating}
\end{cases}
\end{equation}
If a node has no input, its state never changes over time. 

For each network, a single \emph{base} initial condition $\bm \beta \in \{0,1\}^N$ is sampled by setting each $\beta_j$ to $0$ or $1$ uniformly at random, independently.
The base state plays the role of the ideal preparation of an experiment, and it is reused across the control and all shocked versions of the same network.
No preparation is exact in practice, so each repeated trial starts from an imperfect copy of the ideal state.
Replicate initial conditions are generated with a noise parameter $\varepsilon \in [0,1]$.
Each component of the base state is reproduced faithfully with probability $1-\varepsilon$ and re-randomized to a fair coin otherwise,
\begin{equation}
\sigma_j^{(r)}(0) =
\begin{cases}
\beta_j, & \text{with probability } 1-\varepsilon,\\
0 \text{ or } 1 \text{ with equal probability}, & \text{with probability } \varepsilon.
\end{cases}
\end{equation}
The parameter $\varepsilon$ is the fraction of the ideal state that escapes experimental control.
At $\varepsilon = 0$ every trial reproduces the ideal state exactly, and at $\varepsilon = 1$ the replicates carry no memory of it.

\medskip
\noindent\textbf{Equivalence with the copy probability parameterization.}
A re-randomized bit lands on the value opposite to $\beta_j$ with probability $1/2$, so the total flip probability per bit is $\varepsilon/2$.
The same replicate ensemble is therefore produced by copying each bit with probability $\rho$ and flipping it with probability $1-\rho$, provided
\begin{equation}
1-\rho = \frac{\varepsilon}{2},
\qquad\text{or in other words,}\qquad
\varepsilon = 2(1-\rho).
\end{equation}
The code simulations are parameterized by the copy probability $\rho$, and every noise level reported in this paper corresponds to $\rho = 1-\varepsilon/2 \in [1/2, 1]$.
Flip probabilities above $1/2$---corresponding to $\rho < 1/2$---cannot arise from loss of control, because they correlate replicates with the complement of the base state.
The model is symmetric under replacing $(\rho, \bm\beta)$ with $(1-\rho, 1-\bm\beta)$, and base states are sampled uniformly.
So restricting to $\rho \ge 1/2$ loses no generality.

The initial state is $\bm\sigma(0)$.
The ordered states $\bm\sigma(0),\bm\sigma(1),\bm\sigma(2),\ldots$ are collectively
called a \emph{trajectory}.
Since all updating is synchronous and deterministic,
$\bm\sigma(0)$ uniquely defines a trajectory.
We define
\begin{equation}
  F_{\bm \sigma(0)}:\{0,1\}^N\to\{0,1\}^N
\end{equation}
as the dynamical map for the network with initial state $\bm\sigma(0)$.
Thus for all $t$,
\begin{equation}
  \bm\sigma (t+1) = F_{\bm\sigma(0)}(\bm\sigma(t)).
\end{equation}

Because the state space is finite, every trajectory eventually enters a periodic orbit.
We therefore use the terms \emph{attractor} and \emph{basin of attraction} in the usual finite state sense:
an attractor is a fixed point or cycle, and its basin is the set of initial states that eventually flow into it.

The synchronous threshold rule is a deliberate modeling choice rather than a claim of universality.

\noindent\textbf{Finite-state attractors and baseline noise scales.}

Because the update rule $F_{\bm\sigma(0)}$ of each perturbation acts on the finite state space $\{0,1\}^{N}$, every trajectory is eventually periodic.
That is, for every initial condition there exist integers $t_{\star} < 2^N$ and $1 \le p \le 2^N$ such that
\begin{equation}
\bm\sigma (t+p) = \bm\sigma(t)
\qquad \text{for all } t \ge t_{\star}.
\end{equation}
If a trajectory reaches a periodic orbit by time $t_{\star} \le T-L+1$ and the orbit has period $p \le L$, then the stored window
\begin{equation}
\bm \sigma(T-L+1),\dots,\bm\sigma(T)
\end{equation}
contains at least one full cycle.
This is the basis for the language of attractors and attractor basins.

The noisy initial condition model also sets a natural baseline scale before any shock is applied.
For a base state $\bm x \in \{0,1\}^{N}$ and an imperfect copy $\bm y$ of $\bm x$ at noise $\varepsilon$, the normalized Hamming distance is
\begin{equation}
d(\bm x,\bm y) := \frac{1}{N}\sum_{j=1}^{N} |x_j-y_j|,
\end{equation}
with expectation
\begin{equation}
\left\langle d(\bm x,\bm y) \right\rangle = \frac{\varepsilon}{2}.
\end{equation}
If $\bm y$ and $\bm z$ are two independent copies of the same base state at noise $\varepsilon$, then
\begin{equation}
P(y_j \neq z_j) = \varepsilon\left(1-\frac{\varepsilon}{2}\right),
\end{equation}
so
\begin{equation}
\left\langle d(\bm y,\bm z) \right\rangle = \varepsilon\left(1-\frac{\varepsilon}{2}\right).
\label{eq:d0}
\end{equation}
These expressions quantify the variability between trials induced by imperfect preparation of the initial state.

\section{The order to chaos transition}
\label{sec:phase-transition}

Boolean networks exhibit a transition between an ordered regime
where small perturbations of the state die out,
and a chaotic regime
where they spread through the system \cite{kauffman1969metabolic,derrida1986random,aldana2003boolean}.
Here we locate the transition for our ensemble using the annealed approximation of Derrida and Pomeau \cite{derrida1986random}, adapted to threshold units \cite{rohlf2002criticality} and to heterogeneous out-degrees \cite{aldana2003boolean}.

\subsection{Damage spreading}

Fix one network and let $\bm\sigma(t)$ and $\bm\sigma'(t)$ be the trajectories of two independent replicate initial conditions built from the same base state $\bm\beta$ at noise $\varepsilon$.
Their normalized Hamming distance at time $t$ is
\begin{equation}
d_t := \frac{1}{N}\sum_{j=1}^{N}\left|\sigma_j(t)-\sigma'_j(t)\right|.
\label{eq:dt}
\end{equation}
At $t=0$ the two trajectories are exactly the pair of independent copies treated in \cref{eq:d0}.
Without loss of generality, we can look at the state of node $j=1$ in the two initial condition replicates and observe the probability that they differ from each other.
In order for the node state to differ, it must be the case that exactly one of the initial conditions got flipped by noise.
That probability is $2\times (1-\varepsilon/2)\times (\varepsilon/2)$, so
\begin{equation}
\left\langle d_0 \right\rangle = \varepsilon\left(1-\frac{\varepsilon}{2}\right),
\label{eq:d0-avg}
\end{equation}
which rises from a minimum of $0$ at $\varepsilon = 0$ to a maximum of $1/2$ at $\varepsilon = 1$.

Identical trajectories stay identical, so $d = 0$ is a fixed point.
The phases of the map $d_t$ are distinguished by its stability under the smallest perturbation the network admits---a single flipped node.
Let $\bm\sigma$ and $\bm\sigma'$ agree everywhere except at one node $i$, so that $d_t = 1/N$, and let
\begin{equation}
\Lambda := N\cdot \left\langle d_{t+1} \right\rangle.
\label{eq:orderparam}
\end{equation}
The ensemble is frozen when $\Lambda < 1$,
critical when $\Lambda = 1$,
and chaotic when $\Lambda > 1$.
Note that $d_t \approx d_0 \Lambda^{t}$ while $d_t$ stays small,
thus $\ln \Lambda$ is a Lyapunov exponent.
The expectation runs over the in-degree, the weights, the states of the inputs, the node's own state, and which input carries the damage.

The quantity plotted in main text \textbf{Figure 2} is an approximation
of $\langle d \rangle := \lim_{t\to\infty}\left\langle d_{t} \right\rangle$.
Each trajectory is run for $T$ steps and its last $L$ states are retained.
The distance is averaged over that retained window,
then over every unordered pair of the $n_{\mathrm{IC}}$ replicates of a network,
then over the $M$ sampled networks,
giving the approximation $\langle d \rangle \approx \hat{d}$ where
\begin{equation}
\hat{d}
:= \frac{1}{M}\sum_{n=1}^{M}
  \binom{n_{\mathrm{IC}}}{2}^{-1} \sum_{1 \leq r_1 < r_2 \leq n_{\mathrm{IC}}}
  \frac{1}{L}\sum_{\ell=0}^{L-1} d^{\,(n,r_1,r_2)}_{\,T-L+1+\ell}.
\label{eq:dbar}
\end{equation}
The main text \textbf{Figure 2} uses
$T = 1000$, $L = 10$, $\varepsilon = 0.02$, $M = 100$ networks per point,
$n_{\mathrm{IC}} = 10$ replicates, and therefore $45$ pairs per network.

Using \eqref{eq:orderparam} in the limit of minuscule amounts of noise (i.e., $\varepsilon \ll 1$):
in the frozen regime $\langle d\rangle \ll 1/N$,
in the critical regime $\langle d\rangle \approx 1/N$,
and in the chaotic regime $\langle d\rangle \gg 1/N$.
Note that it is possible that some nodes have in-degree zero.
The consequence is that their node states will never change, per the threshold updating rules in \eqref{eq:threshold-updating}.
Thus since in-degrees are distributed as $\mathrm{Binomial}(N, K/N)$ and using \eqref{eq:d0-avg}, for every $t$ we have %
\begin{equation}
\left\langle d_t \right\rangle \geq \left(1-\frac{K}{N}\right)^{\!N} \varepsilon\left(1-\frac{\varepsilon}{2}\right).
\label{eq:floor}
\end{equation}
In the limit of large $N$, this gives us that $\langle d_t \rangle \geq e^{-K}\cdot\varepsilon(1-\varepsilon/2)$.
The bound is not tight, because freezing propagates:
a node whose inputs are all frozen is itself frozen from the next step on
at a value that can differ between the replicates.
At $\varepsilon = 0.02$ where $\left\langle d_0 \right\rangle = 0.0198$,
the actual minimum is close to $0.014$;
all $4$ curves of main text \textbf{Figure 2} share something close to this minimum for $\gamma \gtrsim 2.2$.
Then the curves all ever so slightly increase.

\subsection{Response of a single threshold node}

We consider an initial original state drawn uniformly at random.
We first ask how likely a single flipped input is to change the output of one node.
Fix a focal node with in-degree $k$.
Single out one of its inputs, labeled $1$, with incoming weight $w_1$.
The rest of the inputs have labels $2,\ldots,k$ with corresponding incoming weights $w_2,\ldots,w_k$.
Split the node's total pre-threshold sum as
\begin{equation}
w_1\,\sigma_1 + R
\qquad\text{where}\qquad
R := \sum_{i=2}^{k} w_{i}\,\sigma_i.
\label{eq:field-split}
\end{equation}
In the original trajectory and the trajectory with the flipped input,
both focal nodes agree on inputs $2,\dots,k$ and disagree on input $1$.
So the sum equals $R$ in one node and $R+w_1$ in the other.
The threshold rule sends a positive sum to $1$ and a negative sum to $0$, and holds the previous state when the sum vanishes.

Write $a_t := P(\sigma_i(t) = 1)$ for the activity of node $i$,
with the probability taken over initial states, networks, weights, etc.
A node with at least one active input (i.e. where an input node state has value $1$) receives a sum that is a sum of weights drawn symmetrically about the origin.
Thus, such a sum has a continuous symmetric distribution and the node becomes active with probability exactly $1/2$.
A node with no active input retains its previous state, which is active with probability $a_t$.
Writing $z_t^{(j)} := \prod_{i\to j} (1-a^{(i)}_t)$ for the probability that every input is inactive for node $j$,
one synchronous update maps
\begin{equation}
a^{(j)}_{t+1} = \frac{1-z^{(j)}_t}{2} + z^{(j)}_t\,a^{(j)}_t.
\end{equation}
So the point $\mathbf{a}\in[0,1]^N$ such that $a_j = 1/2$ for all $j$
is a fixed point since $a\,(1-z) = (1-z)/2$ is true when $a=1/2$.
Initial states are drawn uniformly, so $a = 1/2$ holds at $t = 0$ and is preserved at every step afterwards.

\textbf{Case when $R=0$.}
When $R=0$, all of the nonfocal inputs are inactive.
So one of the trajectories and holds its previous state value
while the other takes the sign of $w_1$ which is uniform on $[-1,1]$.
Thus after the threshold updating,
these two focal node states disagree with probability $a\times (1/2) + (1-a)\times (1/2) = 1/2$, giving $q(1) = 1/2$.

\textbf{Case when $R\neq 0$.}
First, suppose $R \neq 0$.
If $w_1 > 0$ the two outputs differ exactly when $R < 0 < R+w_1$, that is when $R$ lies in $(-w_1,\,0)$.
If $w_1 < 0$ they differ exactly when $R+w_1 < 0 < R$, that is when $R$ lies in $(0,\,-w_1)$.
Either way $R$ must land in an interval of width $|w_1|$ with one endpoint at the origin.

Let $p_W$ denote the probability density of a single weight.
The summands of $R$ in \cref{eq:field-split} are independent and identically distributed products of an independent weight and state, with
\begin{equation}
\left\langle w_i\,\sigma_i \right\rangle = a\int_{-1}^{1} w\,p_W(w)\,\mathrm{d}w = 0
\qquad\text{and}\qquad
\left\langle \left(w_i\,\sigma_i\right)^{2} \right\rangle = a\int_{-1}^{1} w^{2}\,p_W(w)\,\mathrm{d}w = s^{2}a,
\end{equation}
the first vanishing because $p_W$ is symmetric, and the second using $\sigma_i^{2} = \sigma_i$ for states in $\{0,1\}$, with $s^{2}$ the variance of the weight distribution.
A local central limit theorem applied to these $k-1$ summands of $R$ gives the probability density $f_R$ of $R$ at the origin,
\begin{equation}
f_R(0) = \frac{1}{\sqrt{2\pi\,(k-1)\,s^{2}a}}\left(1 + O(k^{-1})\right).
\label{eq:fR0}
\end{equation}

Fix $w_1$. %
Since $f_R$ is symmetric about the origin,
the window $(-w_1,\,0)$ and the window $(0,\,-w_1)$ carry the same probability.
So for either sign of $w_1$ we have that
\begin{equation}
P\!\left(\text{threshold outputs differ}\mid w_1\right) = \int_{0}^{|w_1|} f_R(r)\,\mathrm{d}r.
\end{equation}
Averaging this against the weight density gives the single input sensitivity
\begin{equation}
q(k)
:= \int_{-1}^{1}\left[\,\int_{0}^{|w|} f_R(r)\,\mathrm{d}r\,\right] p_W(w)\,\mathrm{d}w
= \int_{0}^{1}\left[\,\int_{0}^{w} f_R(r)\,\mathrm{d}r\,\right]\mathrm{d}w,
\label{eq:qK-integral}
\end{equation}
where the second equality uses $p_W(w) = 1/2$ on $[-1,1]$ together with the invariance of the inner integral under $w \mapsto -w$.
The probability density $f_R$ varies on the scale $\sqrt{k s^{2} a}$ while the inner integral runs over a range of at most $1$.
So expanding about the origin---while keeping in mind $f_R'(0)=0$---gives $f_R(r) = f_R(0)\left(1 + O(k^{-1})\right)$ throughout that range. 
Thus, we have
\begin{equation}
q(k) = f_R(0)\int_{0}^{1} w\,\mathrm{d}w\,\left(1 + O(k^{-1})\right)
= \mu f_R(0)\left(1 + O(k^{-1})\right)
\qquad
\text{where}
\qquad
\mu := \langle |w_1|\rangle = \int_{-1}^{1} |w|\,p_W(w)\,\mathrm{d}w.
\end{equation}
Substituting \cref{eq:fR0} and $a = 1/2$,
\begin{equation}
q(k) = \frac{c}{\sqrt{k}}\left(1 + O(k^{-1})\right)
\qquad\text{where}\qquad
c = \frac{\mu}{s\sqrt{2\pi a}} = \frac{\mu}{s\sqrt{\pi}}.
\label{eq:qK}
\end{equation}
The variance in \cref{eq:fR0} carries $k-1$ and not $k$,
but $(k-1)^{-1/2}$ and $k^{-1/2}$ differ at relative order $k^{-1}$ which is why we use $k$ in \cref{eq:qK}.
For our uniform weights on $[-1,1]$ the density is $p_W(w) = 1/2$ on that interval, so
\begin{equation}
\mu = \int_{0}^{1} w\,\mathrm{d}w = \frac{1}{2},
\qquad
s^{2} = \int_{-1}^{1} \frac{w^{2}}{2}\,\mathrm{d}w = \frac{1}{3},
\qquad
c = \sqrt{\frac{3}{4\pi}} \approx 0.4886.
\end{equation}

Direct simulation of a single threshold unit under the exact update rule reproduces both values and matches \cref{eq:qK} to within $2\%$ for $k \geq 8$ (\cref{fig:si-qk}).
The largest departure is at $k = 2$, where \cref{eq:qK} gives $\sqrt{3/(8\pi)} = 0.3455$ against the exact $3/8$, a shortfall of $7.9\%$.
Indeed, at $k = 2$, the single background input (from the node labeled $2$) is inactive with probability $1/2$,
which reproduces the case $k = 1$ (i.e., $q(1)=1/2$).
Otherwise the background input node is active, making $R$ uniform on $[-1,1]$ so a window of width $|w_1|$ is hit with probability $|w_1|/2$.
averaging over $w$ as in \cref{eq:qK-integral},
\begin{equation}
q(2) = \frac{1}{2}\cdot\frac{1}{2} \;+\; \frac{1}{2}\int_{0}^{1}\frac{w}{2}\,\mathrm{d}w
= \frac{1}{4}+\frac{1}{8} = \frac{3}{8}.
\end{equation}

The essential feature is the decay $q(k) \sim k^{-1/2}$.
In a Kauffman network with random lookup tables the same probability is a constant, $2p(1-p)$, independent of $k$ \cite{derrida1986random,aldana2003boolean}.
A threshold unit sums up its inputs, so one flipped input must overcome a sum whose typical magnitude grows as $\sqrt{k}$.

\subsection{Branching parameter and criticality}

In our ensemble the out-degrees follow the power law and each edge lands on a uniformly chosen target.
The in-degree of a node is distributed as $\mathrm{Binomial}(N, K/N)$. %
Degree heterogeneity therefore enters the damage calculation only through the mean connectivity $K$.

Take the single damaged node $i$ of \cref{eq:orderparam} and ask which other nodes can differ at the next step;
only a node that reads $i$ as an input can.
Consider a node $j$ of in-degree $k$.
Its inputs are drawn uniformly, so it reads $i$ with probability
$1 - \left(1 - 1/N\right)^{k} = k/N + O(N^{-2})$. %
By \cref{eq:qK}, if $i\to j$ is an edge, node $j$ then flips with probability $q(k)$.
Summing over the nodes of the network and averaging over the in-degree distribution,
\begin{equation}
\Lambda = N\cdot \left\langle \left(\frac{k}N+O(N^{-2})\right)\cdot q(k) \right\rangle = \langle k\,q(k)\rangle + O(N^{-1})
\approx c\,\sqrt{K},
\label{eq:branching}
\end{equation}
where the last step uses \cref{eq:qK} and by expanding the square root its the binomial distribution mean $K$.
The transition sits at $\Lambda = 1$.
This gives the critical mean connectivity
\begin{equation}
K_c = \frac{1}{c^2} = \frac{\pi s^2}{\mu^2} = \frac{4\pi}{3} \approx 4.19.
\label{eq:Kc}
\end{equation}

With Kauffman networks,
$\Lambda = 2p(1-p)K$, so criticality at $p=\tfrac12$ requires $K = 2$ \cite{derrida1986random,aldana2003boolean}.
The $k^{-1/2}$ suppression in \cref{eq:qK} makes threshold networks more ordered, pushing the transition to larger connectivity.

\subsection{Finite $N$ and the large $N$ limit}

The mean connectivity of the power law ensemble depends on both $\gamma$ and $N$,
\begin{equation}
K(\gamma, N) = \frac{\zeta_N(\gamma-1)}{\zeta_N(\gamma)}.
\end{equation}
The critical exponent $\gamma_c(N)$ solves $K(\gamma_c, N) = K_c$.
Because the truncated sum $\zeta_N(\gamma-1)$ grows with $N$ whenever $\gamma \leq 2$,
the mean connectivity at fixed $\gamma$ increases with $N$ and the critical exponent moves to the right,
\begin{equation}
\begin{array}{r|ccccc}
N & 50 & 250 & 500 & 5000 & \infty \\
\hline
\gamma_c(N) & 1.68 & 1.94 & 1.99 & 2.09 & 2.17
\end{array}
\label{eq:gammacN}
\end{equation}
In the limit $N \to \infty$ the truncated sums become Riemann zeta functions and $\gamma_c(\infty)$ solves
\begin{equation}
\frac{\zeta(\gamma_c - 1)}{\zeta(\gamma_c)} = \frac{4\pi}{3},
\qquad
\gamma_c(\infty) \approx 2.17.
\end{equation}
For $\gamma \leq 2$ the mean connectivity diverges with $N$, so an infinite network with such a heavy tail is always chaotic.
These predictions are compared with simulation in main text \textbf{Figure 2}.

\begin{figure}
\centering
\includegraphics[width=0.5\textwidth]{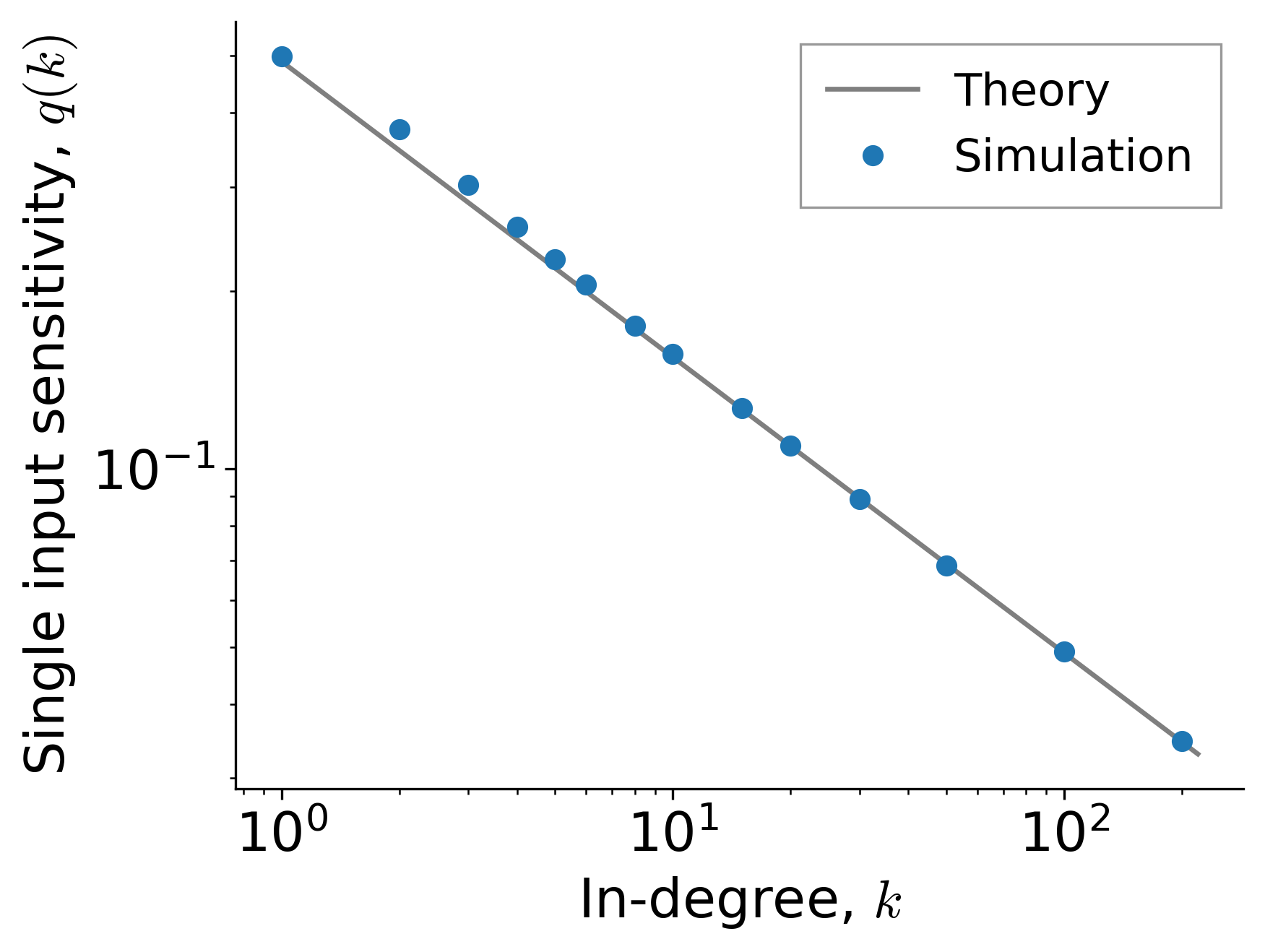}
\caption{
\textbf{Simulation validation of the single input sensitivity.}
The probability $q(k)$ that flipping one input changes the output of a threshold node with in-degree $k$, estimated from $5\times 10^5$ samples per point under the exact update rule with uniform weights on $[-1,1]$ and unbiased inputs.
The solid line is the annealed calculation prediction, $q(k) = \sqrt{3/(4\pi k)}$.
}
\label{fig:si-qk}
\end{figure}

\subsection{Attractor landscape across the transition}

Whether the number of steps $T$ in a simulation samples transient or attractor dynamics depends on $\gamma$.
So we took a census of the attractor landscape across the transition.
For each $\gamma$ we simulated $20$ networks at $N=5\times 10^3$
with $10^3$ random initial conditions each,
followed every trajectory for up to $2\times 10^4$ steps,
and recorded the first exact recurrence of the full network state (\cref{fig:si-census}).

The census resolves $3$ regimes.
Deep in the chaotic phase the accessible state space is effectively transient.
At $\gamma=1.5$ no trajectory recurs within $2\times 10^4$ steps,
and at $\gamma=1.6$ only $5\%$ do, with a median cycle length above $2.5\times 10^3$ among those that recur.
In the frozen phase the landscape is trivial.
From $\gamma=2.0$ on,
every initial condition reaches a fixed point or a two cycle within at most roughly $10^2$ steps,
distinct initial conditions reach distinct states,
and the number of attractors saturates the number of sampled initial conditions.

The working ensemble at $\gamma=1.8$ sits between these extremes.
There, $95\%$ of trajectories recur, $78\%$ are already on their attractor by $T=10^3$, the median cycle length is $54$,
and a typical network funnels its initial conditions into about $5$ attractors, with the largest basin absorbing $71\%$ of them.
The dynamics in near $T=10^3$ steps are therefore dominated by a small set of large cycles rather than by free transients.
This supports the reserve picture.
A small number of large basins makes it such that different initial conditions can lead to the same attractor.
Our regime is where trials are hard to reproduce and readings are taken mid relaxation.
Waiting for full relaxation (i.e. steady state dynamics only) instead would force a choice between
(i) the frozen phase where the transient time is short due to trivial dynamics and
(ii) the necessity of exponentially long time windows.
Sampling expression profiles hours after a perturbation without verifying that the system has completely settled
also matches common experimental practice.

\begin{figure}
\centering
\includegraphics[width=0.9\textwidth]{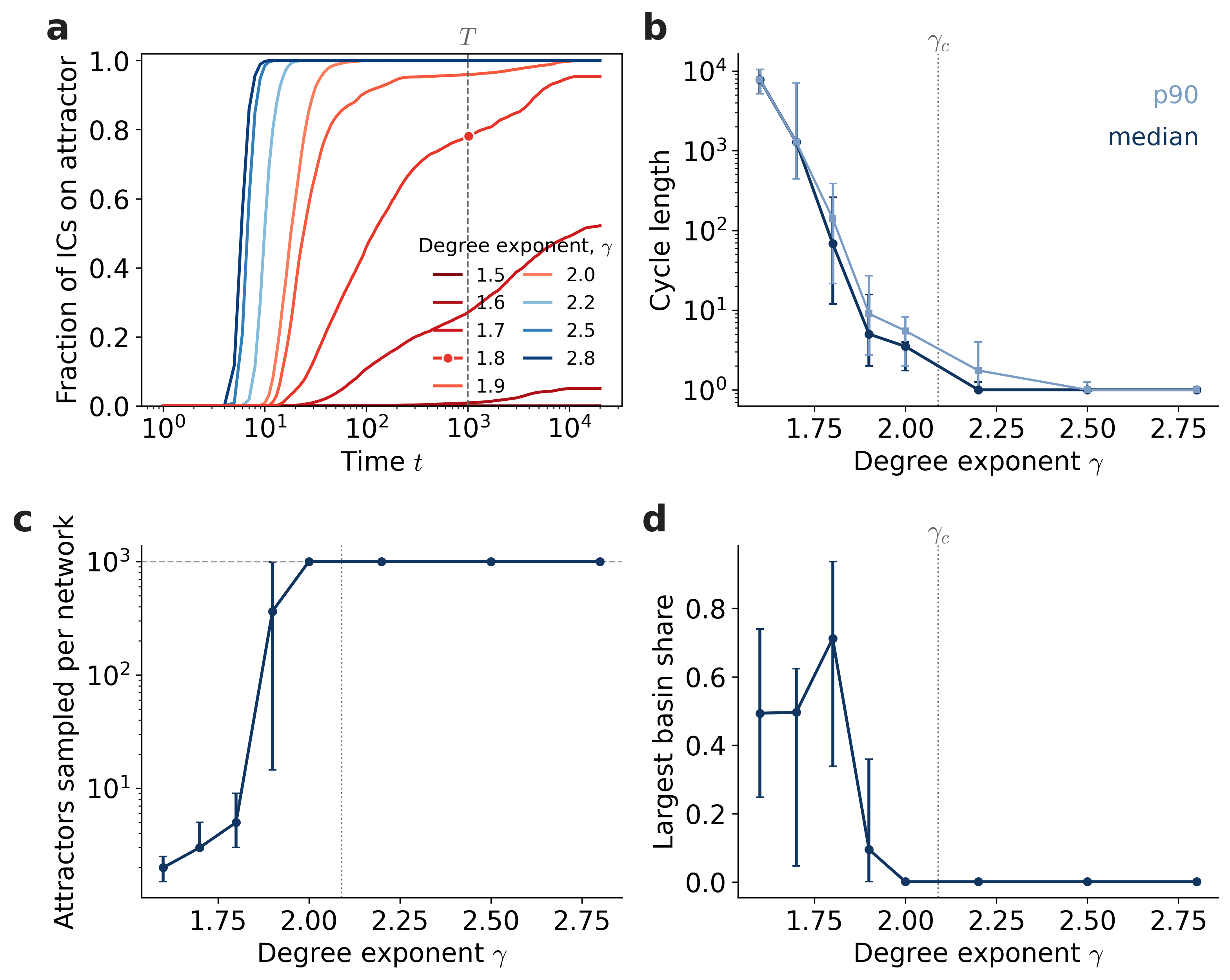}
\caption{
\textbf{Attractor landscape across the phase transition.}
\textbf{a}, Fraction of initial conditions that have reached their attractor by time $t$, one curve per $\gamma$.
Trajectories that do not recur within $2\times 10^4$ steps are counted as not yet settled, so the curves are lower bounds.
The dashed line marks the trajectory length $T=10^3$ used throughout the paper.
At $\gamma=1.5$ no trajectory recurs.
\textbf{b}, Median and $90^{\text{th}}$ percentile cycle length per network.
\textbf{c}, Number of distinct attractors found per network.
Our census samples $10^3$ initial conditions per network, so the dashed ceiling marks saturation.
\textbf{d}, Share of initial conditions absorbed by the largest basin.
Points are medians across $20$ networks and bars span quartiles.
Panels \textbf{b} to \textbf{d} omit $\gamma=1.5$ because no trajectory recurs there.
The dotted vertical line in those panels marks $\gamma_c(N)$ at $N = 5000$.
}
\label{fig:si-census}
\end{figure}

\section{Shocks}

\subsection{Shock construction}

For each network we generate multiple independently perturbed copies.
The unperturbed network serves as the matched control.
Each perturbation selects $g$ target nodes uniformly at random and modifies all outgoing edges of those nodes.
If $w_{ij}$ is an affected edge, the perturbed weight is
\begin{equation}
w'_{ij} = (1-c)w_{ij} + c\xi_{ij},
\end{equation}
where $c \in [0,1]$ is the shock strength and $\xi_{ij} \in \{-1,+1\}$ is chosen independently with equal probability.
All experiments in this work use $c = 1$.
In that limit the perturbed weight is simply $w'_{ij} = \xi_{ij}$, so each affected edge is redrawn to a random sign of full magnitude.
We keep the general form because graded shocks are a natural direction for future work.
Equivalently, the shock can be written as an edge increment
\begin{equation}
\Delta w_{ij} = w'_{ij} - w_{ij} = c(\xi_{ij} - w_{ij}).
\end{equation}

This perturbation rule has several useful properties.
Firstly, $w'_{ij}$ always remains in $[-1,1]$ because it is a convex combination of $w_{ij}$ and $\pm 1$.
Secondly, conditional on $w_{ij}$,
\begin{equation}
\left\langle w'_{ij} \right\rangle_{w_{ij}}  = (1-c)w_{ij}
\qquad\text{and}\qquad \left\langle (w'_{ij})^2 \right\rangle_{w_{ij}}  = [(1-c)w_{ij}]^2 + c^2.
\end{equation}
Thus,
\begin{equation}
\mathrm{Var}\!\left(w'_{ij}\mid w_{ij}\right) = c^2.
\end{equation}
Likewise, conditional on $w_{ij}$,
\begin{equation}
\left\langle \Delta w_{ij} \right\rangle_{w_{ij}} = -cw_{ij},
\qquad\text{and}
\qquad
\mathrm{Var}\!\left(\Delta w_{ij}\mid w_{ij}\right) = c^2.
\end{equation}
The expected absolute change in an affected edge is
\begin{equation}
\left\langle \left|\Delta w_{ij}\right| \right\rangle_{w_{ij}}
= \frac{c}{2}\bigl(|1-w_{ij}| + |-1-w_{ij}|\bigr)
= c.
\end{equation}
At the maximal perturbation $c=1$, every affected edge is reassigned to a random sign.

In the fixed target experiments used for the main text figures, the perturbation configuration is
\begin{equation}
d = 10,
\qquad
g = 50,
\qquad
c = 1.
\end{equation}
The target count $g$ corresponds to $1\%$ of the nodes at $N=5000$.

\noindent\textbf{Perturbation sparsity and overlap.}

If $U_q$ denotes the target set of shock $q$ and $M_q$ the number of directly perturbed edges under that shock, then
\begin{equation}
M_q = \sum_{u \in U_q} k_u^{\mathrm{out}},
\end{equation}
where $k_u^{\mathrm{out}}$ is the out-degree of node $u$.
Conditioned on the realized degree sequence,
\begin{equation}
\left\langle M_q \right\rangle_{\{k_i^{\mathrm{out}}\}_{i=1}^{N}}
= \sum_{i=1}^{N}P(i \in U_q)\,k_i^{\mathrm{out}}
= \frac{g}{N}\sum_{i=1}^{N} k_i^{\mathrm{out}}
= \frac{g}{N}M,
\end{equation}
where $M$ is the total number of edges in the network.
Thus the direct perturbation touches a fraction $g/N$ of the realized edge set in expectation over target choice, regardless of degree heterogeneity.
For the canonical setting this fraction is $50/5000 = 0.01$.

Across $d$ independently sampled shocks, the probability that a given node is never chosen as a target is
\begin{equation}
\left(1-\frac{g}{N}\right)^d,
\end{equation}
so the probability that it is targeted at least once is
\begin{equation}
1-\left(1-\frac{g}{N}\right)^d.
\end{equation}
For the canonical parameters $d=10$ and $g/N=0.01$, this becomes
\begin{equation}
1-(0.99)^{10} \approx 0.0956.
\end{equation}
Even across all $10$ shocks, most nodes are therefore never directly targeted.

\subsection{How shocks reshape the attractor landscape}

We repeated the attractor census under shocks, at $\gamma = 1.8$.
For each of $50$ networks, we followed the same $n_{\textrm{IC}}=10^3$ initial conditions under the control weights and under $d=10$ independent shocks drawn by the construction above (\cref{fig:si-shock-census}).
The initial condition seeds do not depend on the perturbation,
so every comparison below is paired at the level of individual initial conditions.
This census generates its own networks and shocks and is not linked to the classification experiments.
The question is: what of the control landscape survives a shock? %

At the level of exact attractor cycles, a shock rebuilds the landscape completely.
Across $7088$ tests of a control attractor against a shocked landscape,
not one control cycle was ever reached by any of the $n_{\textrm{IC}}=10^3$ trajectories of the shocked system, in any of the $50$ networks.
In contrast, the overall shape of the landscape changes only mildly.
Shocks lengthen transient time (median $+72$ steps across networks)
and slightly increase the number of sampled attractors (median $+1.2$),
but the largest basin share shows no systematic change (median $-0.01$).
Each line in \cref{fig:si-shock-census}a is one network,
connecting the share of initial conditions its largest control basin absorbs
to the mean of that share over its $10$ shocked landscapes.
The lines tilt in both directions with no systematic trend.
Recurrence behaves the same way.
Each cell of \cref{fig:si-shock-census}b is one network under one perturbation
colored by the fraction of the $10^3$ initial conditions that reach an attractor within $T=10^3$ steps.
Most rows are uniformly dark,
so recurrence is mostly a property of the network rather than of the shock.
We found that $37$ of the $42$ networks whose control dynamics recur for every initial condition also recur under all $10$ shocks.
Only $13$ of their $42\times 10=420$ shocked landscapes lose any recurrence at all.
But the exceptions are dramatic.
In one network a single shock cuts recurrence from all initial conditions to $4\%$,
and some networks flip in both directions under different shocks.

Comparing the two landscapes any further needs a distance between attractors of two different systems,
and exact cycle identity is useless for this because no cycle survives.
We use activity fingerprints.
The fingerprint of an attractor is a vector of length $N$ whose entry $j$ is the fraction of one attractor period that node $j$ spends in the active state.
For a fixed point, the fingerprint is the state itself.
The distance between two attractors is the mean absolute difference between their fingerprints, averaged over nodes.
It lies in $[0,1]$,
it is exactly the per node Hamming distance when both attractors are fixed points,
and it generalizes that distance to cycles by comparing time averaged activity.
When we summarize a whole shocked landscape,
each shocked attractor is weighted by the fraction of initial conditions that reach it.
This is the basin weighting of \cref{fig:si-shock-census}c.

What survives a shock is node level structure but not basin identity.
The fingerprints of the new shocked attractors sit close to the control fingerprints.
The basin weighted distance to the nearest control attractor is a median of $0.07$ per node,
against $0.46$ for a ``shuffled'' null, a nearly $7$ fold pinning of the new landscape to the old (\cref{fig:si-shock-census}c).
The null shuffles the node coordinates of each shocked fingerprint before measuring the same distance.
So the shuffled null keeps the distribution of activity values and destroys only their assignment to nodes.
The first column of \cref{fig:si-shock-census}c shows that two distinct attractors of the same control landscape sit a median of $0.01$ apart. %
A shocked attractor is therefore genuinely displaced, several fold beyond the control landscape's own spacing, yet far closer to the control set than chance.
Restricted to the nodes that are not frozen in the control landscape,
the distance from each shocked attractor to its nearest control attractor is $0.10$ against a matched null of $0.42$.
So even the active portion of the fingerprint is largely preserved (\cref{fig:si-shock-census}d).

In contrast, basin identity carries almost nothing across the shock.
Every initial condition has a control basin label and a shocked destination;
\cref{fig:si-shock-census}e treats the first as a predictor of the second.
A baseline predictor ignores the label and always guesses the largest basin of the shocked landscape.
An informed predictor knows the control basin and guesses the shocked attractor that initial conditions from that basin reach most often.
The gain of the informed predictor over the baseline is a median of less than $0.01$ across the $44$ networks with more than one control basin
and enough paired recurrent trajectories to run the test (\cref{fig:si-shock-census}e).

Finally, we counted pinned nodes---the nodes frozen at a constant value in every attractor of every shocked landscape whose value is nonetheless not the same everywhere.
When the value differs across the attractors of a single shock, the node is basin informative.
When the value differs only between shocks, it is shock informative.
Such nodes would be a perfect steady reference because reading them once identifies a basin or a shock with no averaging over time or trials.
Pinned nodes exist but are rare:
a median of $2$ per network, none at all in $14$ networks, up to about $200$ in a few, and every one of them is drawn from the control frozen set (\cref{fig:si-shock-census}f).

\begin{figure}
\centering
\includegraphics[width=0.98\textwidth]{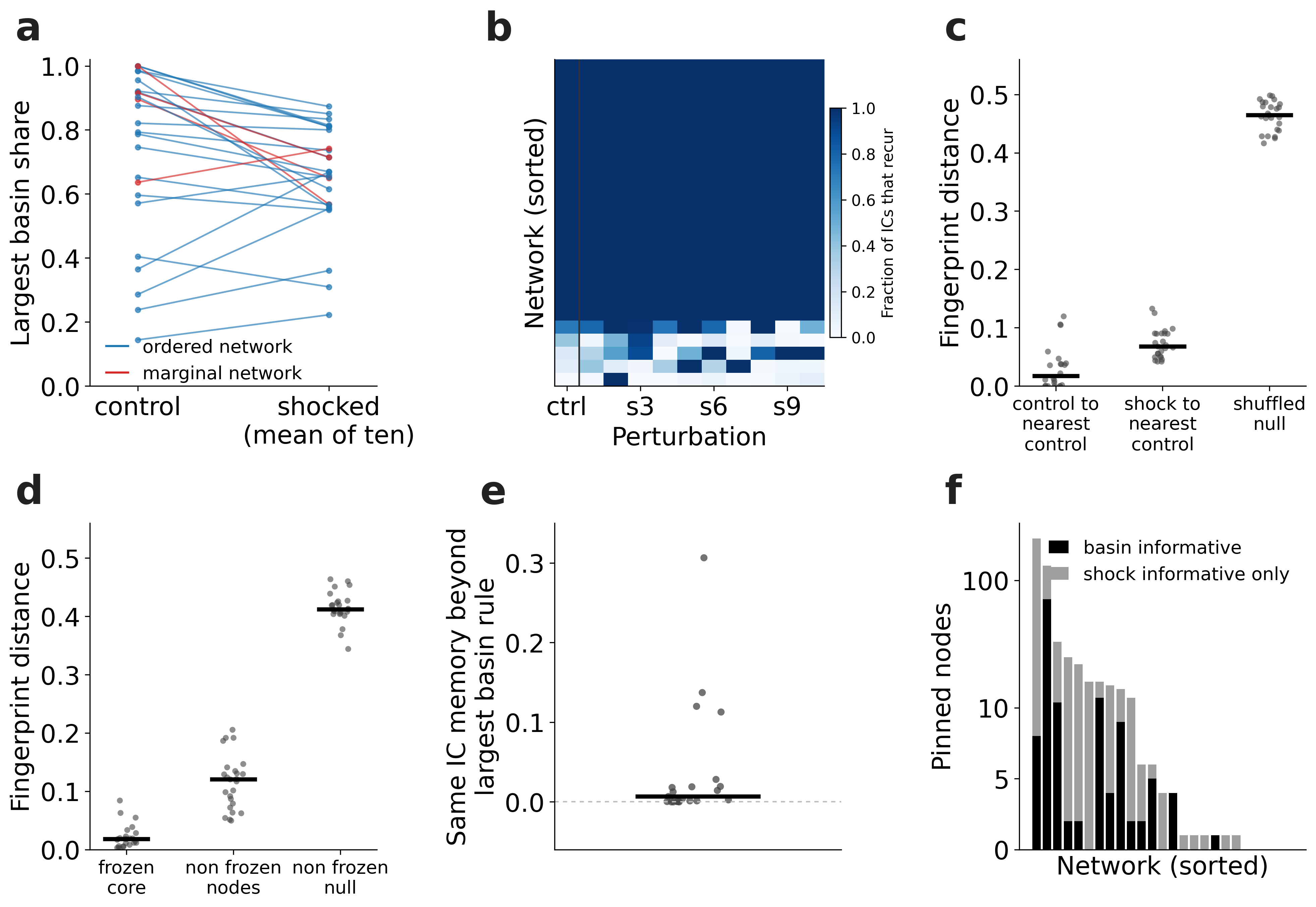}
\caption{
\textbf{Shocks rebuild the attractor landscape, but the control fingerprint still pins it.}
Shock census at $\gamma = 1.8$: $50$ networks, $10^3$ shared initial conditions, control plus $10$ independent shocks per network.
\textbf{a}, Largest basin share under control and averaged over the $10$ shocks, one line per network.
Red marks networks whose control dynamics recur for fewer than half of the initial conditions. %
\textbf{b}, Fraction of initial conditions that recur within $2\times 10^4$ steps, per network (rows) and shock (column).
Most fully recurrent networks stay fully recurrent under every shock,
a few are broken by individual shocks,
and the rest of the networks flip in both directions.
\textbf{c}, Basin weighted per node fingerprint distance from each shocked attractor to its nearest control attractor,
against a shuffled node null and against the spacing among control attractors.
No shocked attractor repeats any control cycle exactly,
yet the new attractors sit far closer to the control set than chance.
\textbf{d}, The distance from each shocked attractor to its nearest control attractor,
restricted to the control frozen core,
to the non frozen nodes,
and to the matched null for the non frozen nodes.
\textbf{e}, The gain in predicting an initial condition's shocked destination from its control basin versus always guessing the largest shocked basin.
\textbf{f}, Pinned nodes per network, the nodes frozen in every shocked attractor whose value differs across attractors within a shock (basin informative) or only across shocks.
All pinned nodes are contained in the control frozen set.
}
\label{fig:si-shock-census}
\end{figure}

\section{Classification}

\subsection{Canonical simulation sizes}

All fixed target experiments of the main text use networks of $N = 5\times 10^3$ nodes
and $50$ independent network realizations.
Each network receives $d=10$ independent shocks (in parallel) from the control network, resulting in the control plus $10$ additional networks.
Together with the matched control this defines $11$ perturbation classes.
For every class, dynamics are run from $n_{\textrm{IC}}=10$ initial conditions.
The final $L=10$ states of each trajectory of $T=10^3$ steps are retained as observation snapshots.
Each class therefore contributes $10 \times 10 = 10^2$ snapshots,
each network contributes $11 \times 100 = 1.1\times 10^3$,
and the full data set holds $5.5\times 10^4$ labeled snapshots across the $50$ networks.
Because the class sizes are balanced, the chance baseline for classification is $1/11 \approx 0.09$.

\subsection{Observation and inference}

\noindent\textbf{Reporter panels and retained observations.}

If $R=\{j_1,\dots,j_m\}$ is a selected reporter panel, the observation at time $t$ is the restricted state vector
\begin{equation}
\bm r_R(t) = \bigl(\sigma_{j_1}(t),\dots,\sigma_{j_m}(t)\bigr).
\end{equation}
The main question of the paper is how to choose $R$ when only a small number $m$ of nodes can be monitored jointly.
At any one time, a reporter panel maps the full system state into $\{0,1\}^{m}$.
A classifier is trained not on the full network trajectory but on the collection of retained reporter snapshots
pooled across replicates for a fixed network and perturbation class.
This places the emphasis on separability among the settled states of the dynamics.

\noindent\textbf{Classification pipeline.}

Classification is performed separately for each underlying network realization.
For a given network and reporter panel,
all retained snapshots are assembled into a dataset $\mathcal D \equiv \{(X_i, Y_i)\}_i$ where $(X_i,Y_i) \in \mathcal X \times \mathcal Y$
with $\mathcal X = \{0,1\}^m$ are the possible snapshots and $\mathcal Y$ are the $|\mathcal Y|=11$ perturbation class labels (including the control).

The classifier is a random forest with $100$ trees.
For each class, $50$ retained snapshots are sampled for training and the remaining snapshots are held out for testing.
To classify one test trial for a fixed class, we sample $50$ withheld snapshots from that class, predict the class of each snapshot, and assign a single label by plurality vote.
This procedure is repeated $10$ times per class, and accuracy is averaged across classes, repetitions, and networks.

\part{Our Findings}

\section{Reporters}

\subsection{Genetic optimization of reporter panels}

Reporter panels are optimized separately for each network realization and each panel size $m$.
The optimization variable is the set of selected node labels
\begin{equation}
R \subseteq \{1,\dots,N\},
\qquad
|R| = m.
\end{equation}
The fitness of a candidate panel is the mean classification accuracy returned by the random forest evaluation described above.
The search space has size
\begin{equation}
\binom{N}{m},
\end{equation}
which is already combinatorially enormous for moderate $m$.
This is why we use a genetic algorithm rather than exhaustive enumeration.

The genetic algorithm is population based.
The population holds $100$ candidate panels.
The genetic algorithm runs for $30$ generations.
Each member of a copied panel mutates with probability $0.1$.
Selection/sampling probabilities are proportional to an exponential map of accuracy,
\begin{equation}
P(\text{select individual } a)
\propto \exp\!\left(\frac{\mathrm{accuracy}(a)}{T_{\mathrm{sel}}}\right),
\end{equation}
with selection temperature $T_{\mathrm{sel}} = 0.1$.

At each generation, the top $10\%$ of individuals (i.e. the elites) are retained unchanged.
The remaining population is split evenly between two reproduction modes.
In copying with mutation, one parent is sampled and each feature is retained or replaced by a randomly sampled node.
In crossover, features are assembled by repeatedly drawing parent features from the current population.
Every offspring is rescored by rerunning the classifier.
The fitness of a panel is its mean accuracy over $10$ fresh train and test splits, redrawn at every evaluation.
The retained elites are rescored every generation. %
Fitness remains stochastic because the random forest and the splits are stochastic,
so the optimization acts on a noisy but consistently defined objective.

The optimization converges within roughly $10$ to $15$ generations at every noise level (\cref{fig:si-convergence}).

\begin{figure}
\centering
\includegraphics[width=0.6\textwidth]{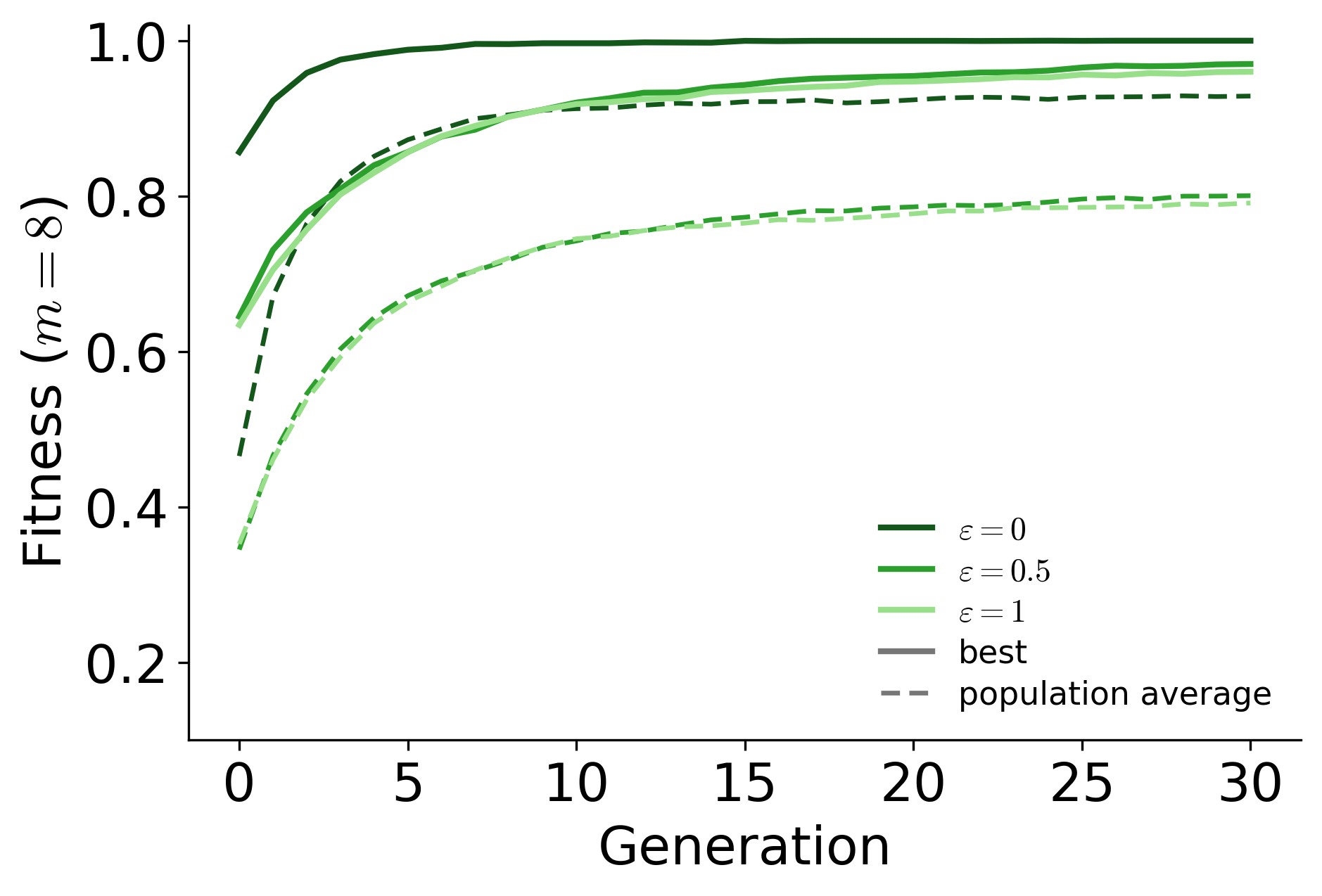}
\caption{
\textbf{Convergence of the genetic algorithm.}
Best and population average fitness at $m = 8$ by generation, averaged over $50$ networks, at the $3$ canonical noise levels.
Solid curves are the best individual and dashed curves the population average; the darker the shade, the lower the noise.
Fitness averages $10$ fresh splits per evaluation.
Elites are rescored every generation.
}
\label{fig:si-convergence}
\end{figure}

We searched over panel sizes that vary over powers of $2$ up until (but not surpassing) $N$.
In the small panel regime emphasized in the main text, the relevant sizes are $1$, $2$, $4$, $8$, $16$, and $32$.
Members of a panel are distinct nodes.

\noindent\textbf{Splitting by snapshot rather than by trajectory.}

The evaluator draws $50$ snapshots per class for training and tests on the remaining snapshots of the same network
without grouping them by initial condition,
so snapshots from one trajectory can fall on both sides of the split.
To measure what this leakage costs, we repeated the evaluation holding out whole initial conditions---%
training on $5$ replicates of each condition and testing on the other $5$---%
which leaves the training and test sets at the same $50$ rows per class and changes only the grouping.

Averaged over $50$ networks and $10$ splits at $m = 8$, evolved panels fall from $0.960$ to $0.945$ at $\varepsilon = 0.5$ and from $0.944$ to $0.931$ at $\varepsilon = 1$.
These starting values come from the $10$ splits used here rather than the $30$ splits used elsewhere in the paper.
The other selection algorithms move by no more than $2$ points in the same direction:
from $0.714$ to $0.694$ and from $0.697$ to $0.692$ for selecting the $8$ most sensitive nodes,
and from $0.364$ to $0.362$ and from $0.354$ to $0.351$ for selecting random panels.
So the ordering between selection algorithms is unchanged at every noise level.

The node class asymmetry is likewise unaffected.
Removing one sensitive member costs more than removing one insensitive member by $0.025$ under the snapshot split and by $0.028$ under the held out condition split at $\varepsilon = 0.5$,
and by $0.017$ and $0.016$ at $\varepsilon = 1$,
each difference paired within a panel and averaged across networks. %
This test splits members at the sensitivity cutoff rather than by response breadth.
The premium of \textbf{Figure 6b} uses the breadth classes and is therefore a different quantity.
The held out condition split also trains on $5$ distinct trajectories rather than $10$.

At $\varepsilon = 0$ the comparison is empty by construction because a perfectly reproduced initial condition makes all $10$ replicates of a condition identical and no trajectory can be held out.

\subsection{Sensitivity}

\noindent\textbf{Sensitivity of single nodes to single shocks.}

For a fixed network, node $j$, shock $q$, replicate $r$, and retained snapshot index $\ell$ (i.e. a relative time step), let
\begin{equation}
x_{j,q,r,\ell} :=
\left|\sigma^{(q)}_{j,r,\ell} - \sigma^{(0)}_{j,r,\ell}\right|,
\end{equation}
where $q=0$ denotes the control and $q=1,\dots,d$ denote the shocks.
The sensitivity of node $j$ to an individual perturbation $q$ is
\begin{equation}
S_{j,q}
:= \frac{1}{n_{\mathrm{IC}}L}
\sum_{r=1}^{n_{\mathrm{IC}}}
\sum_{\ell=1}^{L}
x_{j,q,r,\ell}.
\end{equation}
The sensitivity of each node used throughout the paper is
\begin{equation}
S_j
:= \frac{1}{d}
\sum_{q=1}^{d}
S_{j,q}.
\end{equation}
Equivalently, $S_j$ is the average nodewise Hamming distance between the shocked and control trajectories after averaging over shocks, replicates, and retained samples.

This definition is relative to the control trajectory on the same underlying network and with the same noisy initial condition.
It therefore solely measures shock response. %
The quantities $S_{j,q}$ are the wedges in the disk plots of main text \textbf{Figure 3f} and \textbf{Figure 4}.

\noindent\textbf{Panel summary statistics used in the main text analyses.}

When a sensitivity cutoff $\theta$ is used to polarize nodes into sensitive and insensitive classes for descriptive plots,
the fraction of sensitive nodes in a panel $R$ is
\begin{equation}
\phi_{\theta}(R)
:= \frac{1}{|R|}\sum_{j \in R} \mathbf{1}\!\left\{S_j > \theta\right\}.
\end{equation}

For ablation analyses, let $R_m^{\star}$ denote an optimized panel of nominal size $m$, and let $A \subseteq R_m^{\star}$ be the subset of removed nodes.
The loss in accuracy after removing $A$ is
\begin{equation}
\Delta_{\mathrm{acc}}(A;R_m^{\star})
:= \mathrm{Acc}(R_m^{\star}) - \mathrm{Acc}(R_m^{\star}\setminus A),
\end{equation}
where $\mathrm{Acc}(R')$ is the random forest classifier accuracy when using only the node states in $R'$ for training and testing.

\noindent\textbf{Defining sensitivity thresholds.}

The distribution of node sensitivities is bimodal at every noise level (\cref{fig:si-bdist}).
A large group of nodes has low mean deviation, responding strongly to at most $1$ or $2$ shocks.
A second group responds to nearly every shock at a saturating level.
There is a clear valley between these two groups.
The sensitivity cutoff/threshold $\theta$ is placed at the antimode---the minimum between the two modes.
The cutoff varies only between $0.27$ and $0.36$ across the full noise range.
The main conclusions of the paper do not rely on its exact value.

The two modes have a simple mechanistic origin,
and the gap is deeper within a single network than the pooled histogram suggests (\cref{fig:si-bimodality}).
The low mode is dominated by the frozen population.
About half of all nodes hold a constant state in the control dynamics, and a node that never moves without a shock has nothing to report.
Those are the unresponsive nodes, which are a class of their own and not in the dormant class.
The dormant reporters that panels recruit stay active in control (\cref{fig:si-bimodality}a).
The high mode is set by decorrelation.
Writing $v$ for the variance of a node's control state,
a shocked trajectory that decorrelates from control while keeping the same activity produces a mean absolute difference of $2v$,
which is at most $0.5$ since $v\leq 0.25$.
Active nodes lie along that line with a correlation of $r=0.92$ between the sensitivity $S$ and $2v$.
Thus the high mode is the pile up of decorrelated nodes against the ceiling at $S = 0.5$.
Nodes above the line are not merely decorrelated, they have shifted their mean state, and frozen nodes can only respond in that way (\cref{fig:si-bimodality}c).
Each network has its own antimode, ranging from $0.07$ to $0.45$, so pooling networks blurs the valley.
The relative depth of the dip is $0.46$ pooled but $0.62$ in the median individual network,
We find that $43$ of $50$ networks show a clear dip on their own (\cref{fig:si-bimodality}b).
Aligning each network on its own antimode before pooling recovers the gap (\cref{fig:si-bimodality}d).

\begin{figure}
\centering
\includegraphics[width=0.98\textwidth]{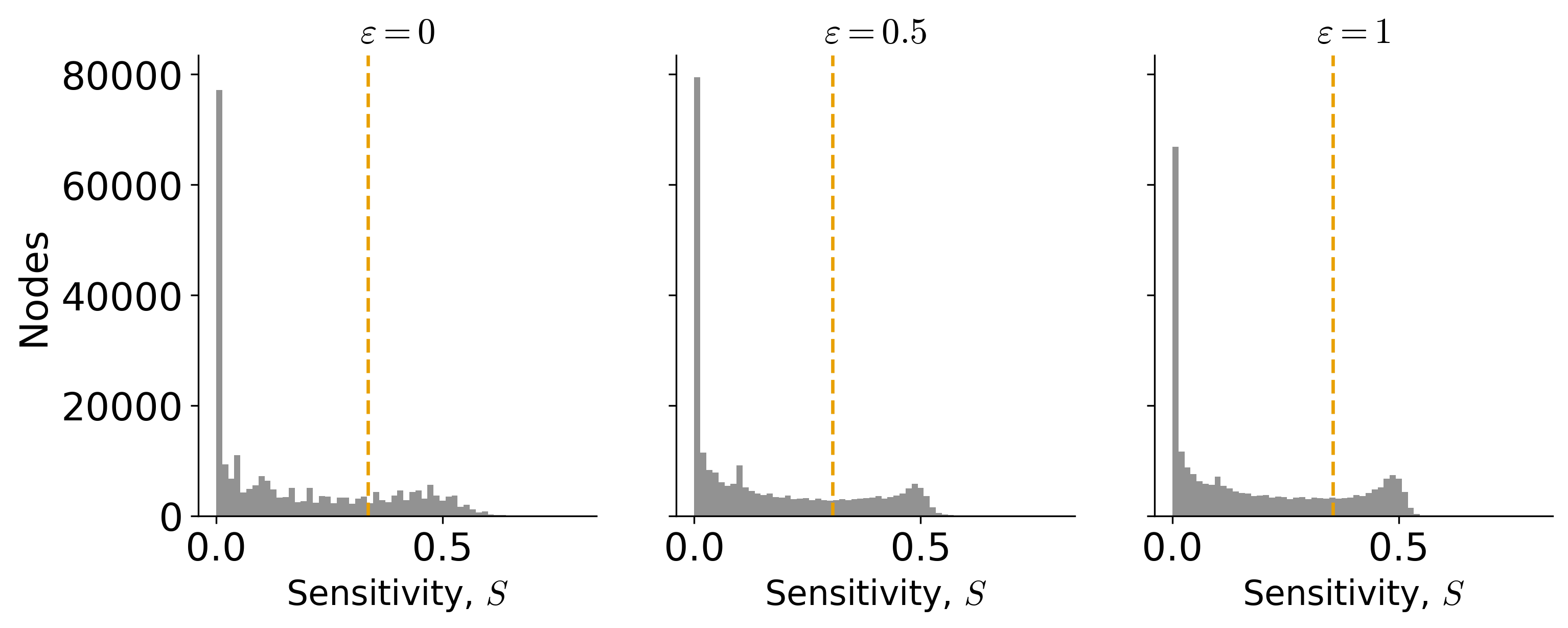}
\caption{
\textbf{Sensitivity distributions across noise levels.}
Histograms of the node sensitivity $S$ pooled over $50$ networks at $3$ noise levels $\varepsilon$. %
The frozen spike in the first bin holds about a quarter of all nodes, towering over the rest of the distribution.
Dashed lines mark the antimode cutoffs used to classify nodes as sensitive or insensitive. %
At $\varepsilon = 0$ every replicate starts from the same state, so $S$ takes fewer distinct values and the histogram is comb-like.
}
\label{fig:si-bdist}
\end{figure}

\begin{figure}
\centering
\includegraphics[width=0.98\textwidth]{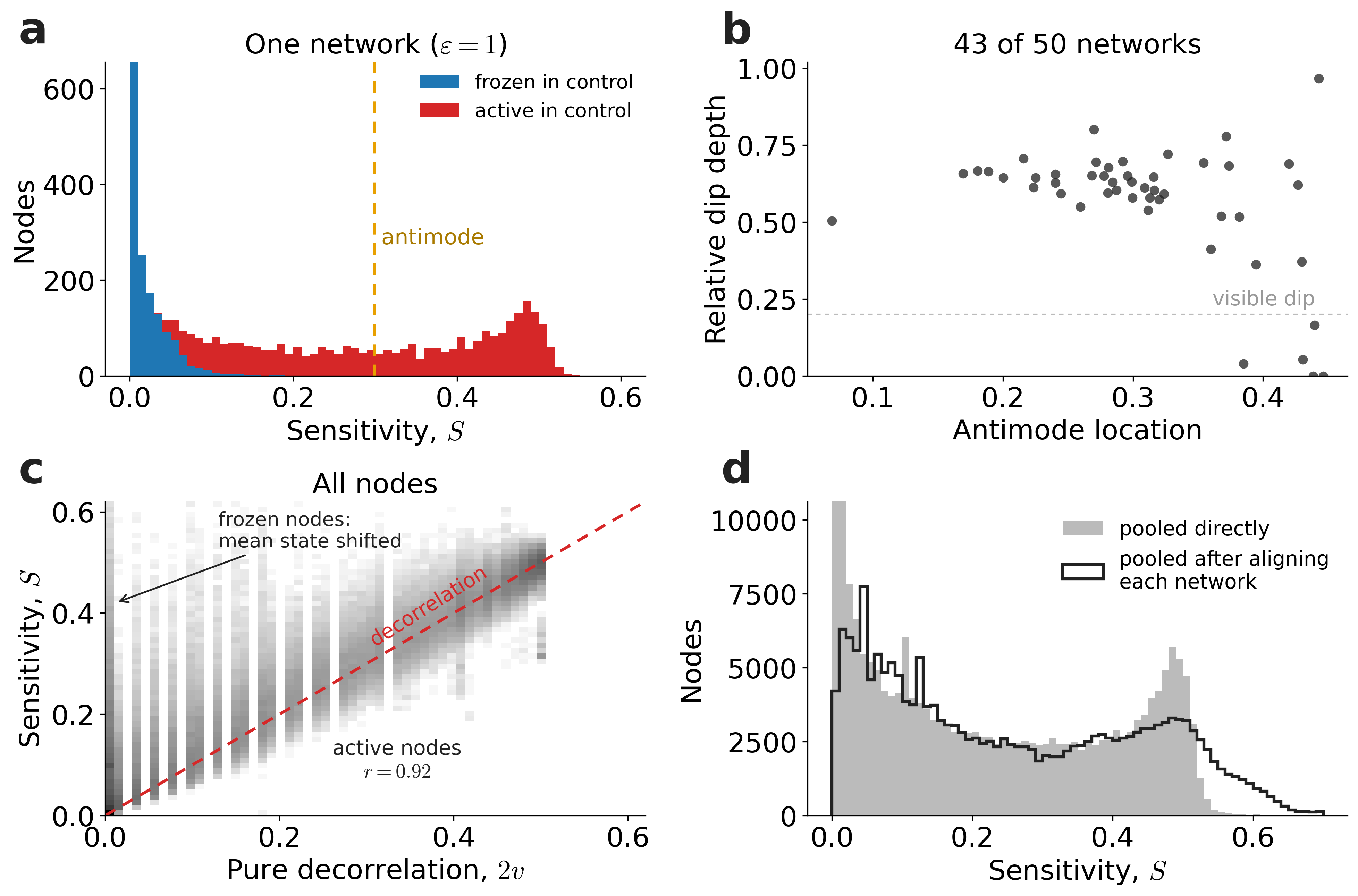}
\caption{
\textbf{Where the two modes of the sensitivity distribution come from.}
All panels at $\varepsilon = 1$.
\textbf{a}, An example network split by whether a node is frozen in the control dynamics (control variance below $0.02$) or active.
The low mode is dominated by frozen nodes and the high mode is entirely active ones.
The active nodes that fall in the low mode are the dormant reporters. %
The dashed line is this network's antimode.
\textbf{b}, Each network has its own antimode, and $43$ of $50$ networks show a dip of relative depth above $0.2$ on their own.
Writing $h$ as height, relative depth is $1 - h_{\mathrm{antimode}} / \min(h_{\mathrm{left}}, h_{\mathrm{right}})$ on a smoothed histogram,
so $0$ is no dip and $1$ is an empty antimode bin.
\textbf{c}, Sensitivity $S$ against $2v$, where $v$ is the variance of the node's control state.
A shocked trajectory that decorrelates from control at unchanged activity gives $S = 2v$, the dashed line, which saturates at $0.5$.
Active nodes lie along it, which places the high mode at the decorrelation ceiling.
Nodes above the line have shifted their mean state rather than merely decorrelated, and frozen nodes, at $v = 0$, can respond only in that way.
\textbf{d}, Pooling networks directly blurs the valley because each network's antimode sits in a different place.
Aligning each network on its own antimode before pooling recovers it.
}
\label{fig:si-bimodality}
\end{figure}

\noindent\textbf{Sensitivity reflects the dynamical regime, not the topology.}

A node's sensitivity is almost fully determined by its activity in the unperturbed control dynamics.
The variance of a node's state across control trajectories correlates with its sensitivity at the level of $r=0.87$ to $r=0.93$ (\cref{fig:si-activity}).
Between $51\%$ and $63\%$ (by noise level) of nodes are frozen in the control run.
These nodes have mean sensitivity near $0.05$ compared to $0.32$ to $0.37$ (respectively) for active nodes.

Evolved panels are topologically unremarkable in every generic sense (\cref{fig:si-connectivity}).
Node degrees carries almost no sensitivy information.
The correlation between sensitivity and in-degree is at most $+0.06$, and between sensitivity and out-degree it is $0.000$.
So there is insignificant association between degree and sensitivity.
Sensitivity is therefore a property of the dynamical configuration.

The mean graph undirected distance between panel members matches size matched random panels (\cref{fig:si-connectivity}a).
The fraction of the network within a given undirected distance (measured in hops) of the panel lies inside the random band at every radius.
The degree distributions of panel members trace the network background (\cref{fig:si-connectivity}b).
The one structural signature is specific to the shocks.
Following edge direction,
the mean directed distance from a shock's target nodes to the nearest panel member is $1.15$ for evolved panels at $\varepsilon = 1$ and $1.22$ at $\varepsilon = 0.5$.
For random panels, it is $1.41$ to $1.45$ (\cref{fig:si-connectivity}c).
Random panels matched for sensitivity composition show $1.40$ to $1.43$, indistinguishable from unconstrained random panels.
Thus the alignment to the shock geometry is a genuine placement effect and not a byproduct of selecting promiscuous nodes.
At $\varepsilon = 0$ the evolved distance rises to $1.37$ and the alignment largely disappears, consistent with reproducible trials rewarding readout strength over placement.
The worst covered shock (largest distance from panel) is also closer.

\begin{figure}
\centering
\includegraphics[width=0.85\textwidth]{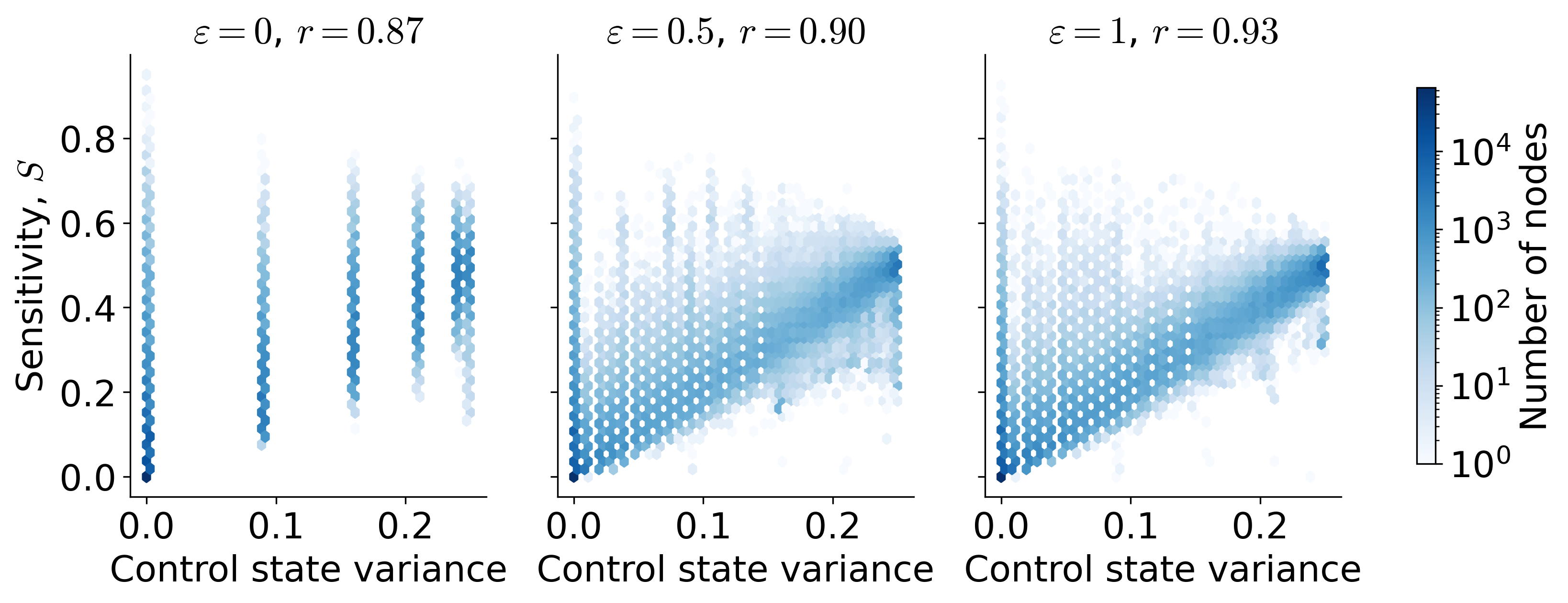}
\caption{
\textbf{Control run activity predicts shock sensitivity.}
Two dimensional histograms of the control state variance of each node against sensitivity $S$, pooled over $50$ networks, at $\varepsilon = 0$, $0.5$, and $1$.
Most nodes that are frozen in the control dynamics answer no shock.
The share is between $69\%$ and $77\%$ across the $3$ noise levels.
The dormant nodes are low sensitivity nodes that remain active in control.
}
\label{fig:si-activity}
\end{figure}

\begin{figure}
\centering
\includegraphics[width=0.9\textwidth]{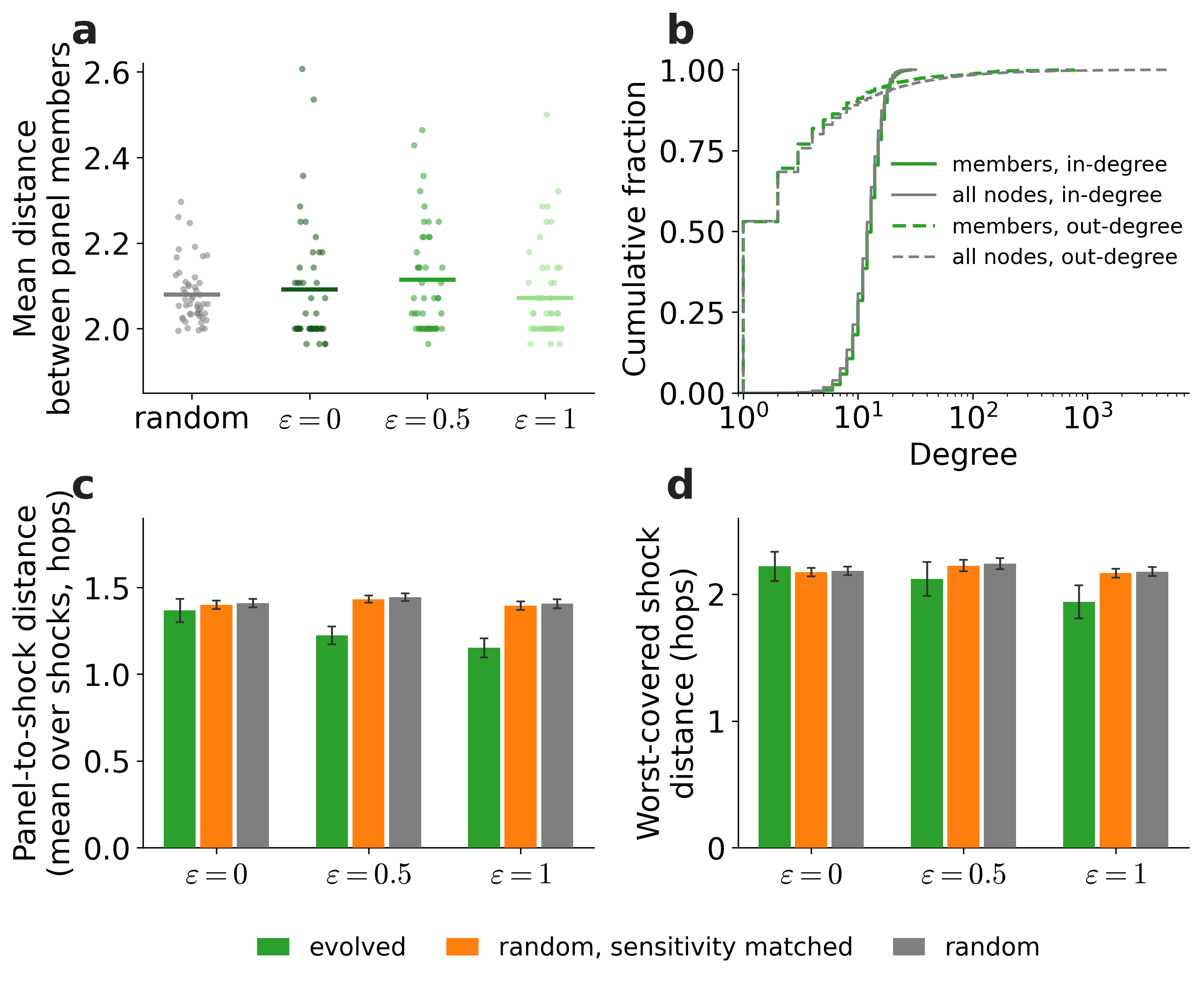}
\caption{
\textbf{Evolved panels align to the shocks rather than to the topology.}
\textbf{a}, Mean graph distance between panel members for evolved panels and size matched random panels.
\textbf{b}, Degree distributions of panel members trace the network background.
\textbf{c}, Mean distance from each shock's targets to the nearest panel member, following edge direction, for evolved panels, sensitivity matched random panels, and unconstrained random panels.
\textbf{d}, Distance of the worst covered shock, same comparison.
Error bars are $95\%$ confidence intervals over networks.
}
\label{fig:si-connectivity}
\end{figure}

\subsection{Promiscuous and dormant reporters}

\noindent\textbf{Evolved panel members are informative and complementary.}

The information structure of an evolved panel can be summarized by the marginal entropy of each member and the mutual information between member pairs,
computed over all retained snapshots (\cref{fig:si-redundancy}).
The (marginal) entropy for a node $j$ is $H(X_j)$ where $X_j$ is the random variable of the state of node $j$ in the snapshots.
Then the mean entropy is $\bar H:=\frac1N\sum_j H(X_j)$.
Evolved panels have higher marginal entropy per member than random panels matched for sensitivity composition%
---making random panels using the same proportion of sensitive versus insensitive nodes---%
$0.81$ to $0.88$ bits against $0.68$ to $0.76$ for various noise levels (\cref{fig:si-redundancy}a).
At the same time their members share less pairwise information.
The pairwise entropy between two nodes $i$ and $j$ is $I(X_i;X_j)$.
The mean pairwise mutual information is lower than in composition matched panels,
and the largest pairwise mutual information within a panel is lower still. 
Unconstrained random panels show even smaller mutual information, but only because their members are mostly frozen and carry little information of any kind.
The genetic algorithm therefore selects members that are individually informative and mutually complementary.

\begin{figure}
\centering
\includegraphics[width=0.98\textwidth]{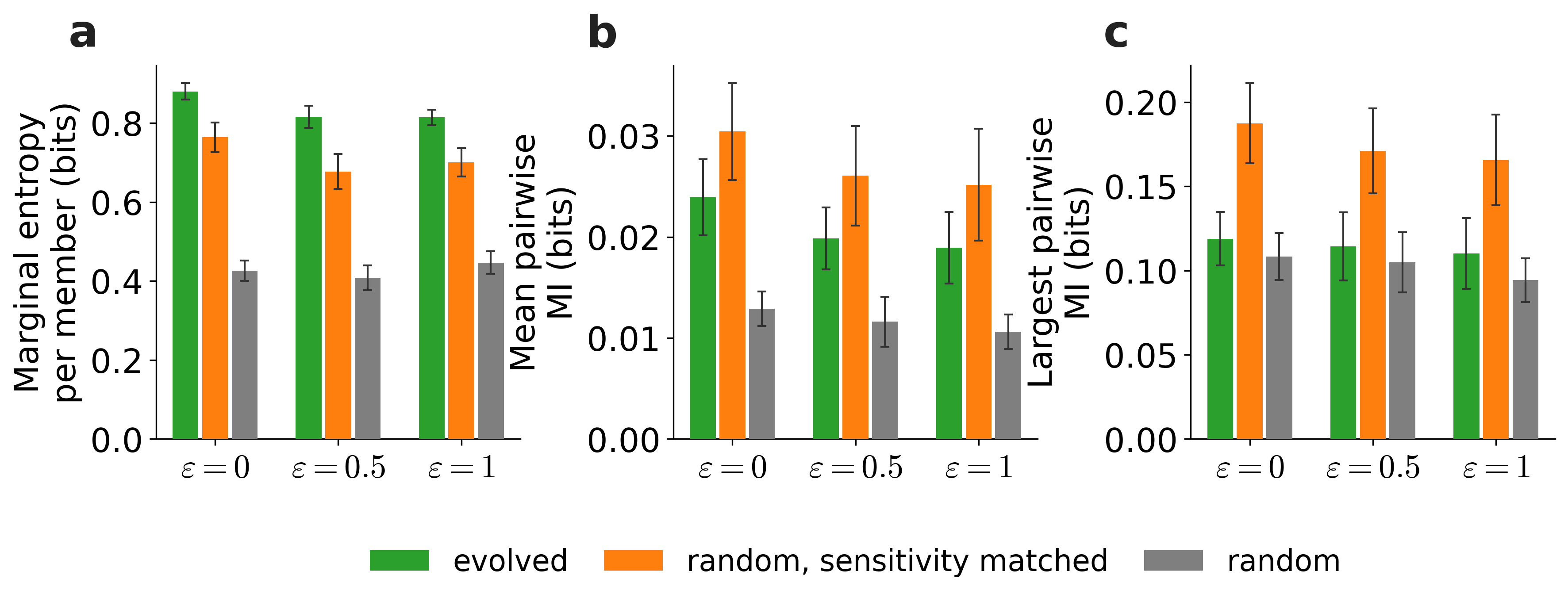}
\caption{
\textbf{Information structure within panels.}
\textbf{a}, Mean marginal entropy per panel member.
\textbf{b}, Mean pairwise mutual information between members.
\textbf{c}, Largest pairwise mutual information within the panel.
Bars show means over networks with $95\%$ confidence intervals, for evolved panels, random panels matched for the number of sensitive members, and unconstrained random panels, at $\varepsilon = 0$, $0.5$, and $1$.
}
\label{fig:si-redundancy}
\end{figure}

\noindent\textbf{Credit assignment by Shapley value.}

Our removal (ablation) experiments measure what a reporter contributes given the rest of the panel, averaged without bias over all possible removal sequences.
We use another notion of ``value added'' for the reporters to strengthen our main claim.
The Shapley value averages a member's marginal contribution over every possible sub panel \cite{shapley1953value}.
To compute the Shapley value,
we treat a panel as a cooperative game whose value is the classification accuracy of a sub panel.
More concretely, for a set of $m$ reporters $R$, the Shapley value of reporter $i$ is
\begin{equation}
  \varphi_i := \frac1m\sum_{R'\subseteq R \setminus \{i\}}{\binom{m-1}{|R'|}}^{-1}\times \left(\mathrm{Acc}(R'\cup \{i\})-\mathrm{Acc}(R')\right)
\end{equation}
Among the many nice properties of $\varphi_i$, we note that $\sum_{i=1}^m \varphi_i = \mathrm{Acc}(R) - \mathrm{Acc}(\emptyset)$.
The empty panel is assigned chance accuracy.
So the Shapley value is the credit of reporter $i$ in the panel $R$.

Our result is a sharp version of the main claim (\cref{fig:si-shapley}).
With a perfectly reproduced initial condition,
sensitive members carry more than twice the credit of insensitive ones, $0.132$ against $0.061$.
That gap closes almost immediately.
At $\varepsilon = 0.01$,
the two classes are already within $10$ percent of each other,
$0.108$ against $0.100$.
From $\varepsilon = 0.05$ to $\varepsilon = 1$, the two classes are pretty much indistinguishable (\cref{fig:si-shapley}b).
The convergence is driven from both sides.
The credit assigned to sensitive members falls with noise while the credit assigned to insensitive members rises.
Counting members rather than averaging credit gives the same picture.
The fraction of insensitive members that contribute more than the median sensitive member of their own panel rises from $0.03$ at $\varepsilon = 0$ to about $0.5$ once noise is present,
meaning a typical insensitive member is as valuable as a typical sensitive one (\cref{fig:si-shapley}c).
The insensitive group here is the whole pool below the cutoff, so it holds dormant and unresponsive members alike.
Since unresponsive members carry almost no credit, pooling them makes every comparison in this figure conservative.
This is an independent confirmation of the step reported in main text \textbf{Figure 6b}.

The same figure shows where the genetic algorithm recruits below the cutoff.
They are not drawn from the bottom of the low sensitivity pool.
Their median sensitivity is $0.231$ against $0.059$ across all nodes below the cutoff.
So the genetic algorithm reaches for the most sensitive of those and avoids the unresponsive nodes that never move at all.
This also explains why reserving picks for the least sensitive nodes fails.
Such a rule selects for minimal sensitivity and lands in the frozen pool.

\begin{figure}
\centering
\includegraphics[width=0.98\textwidth]{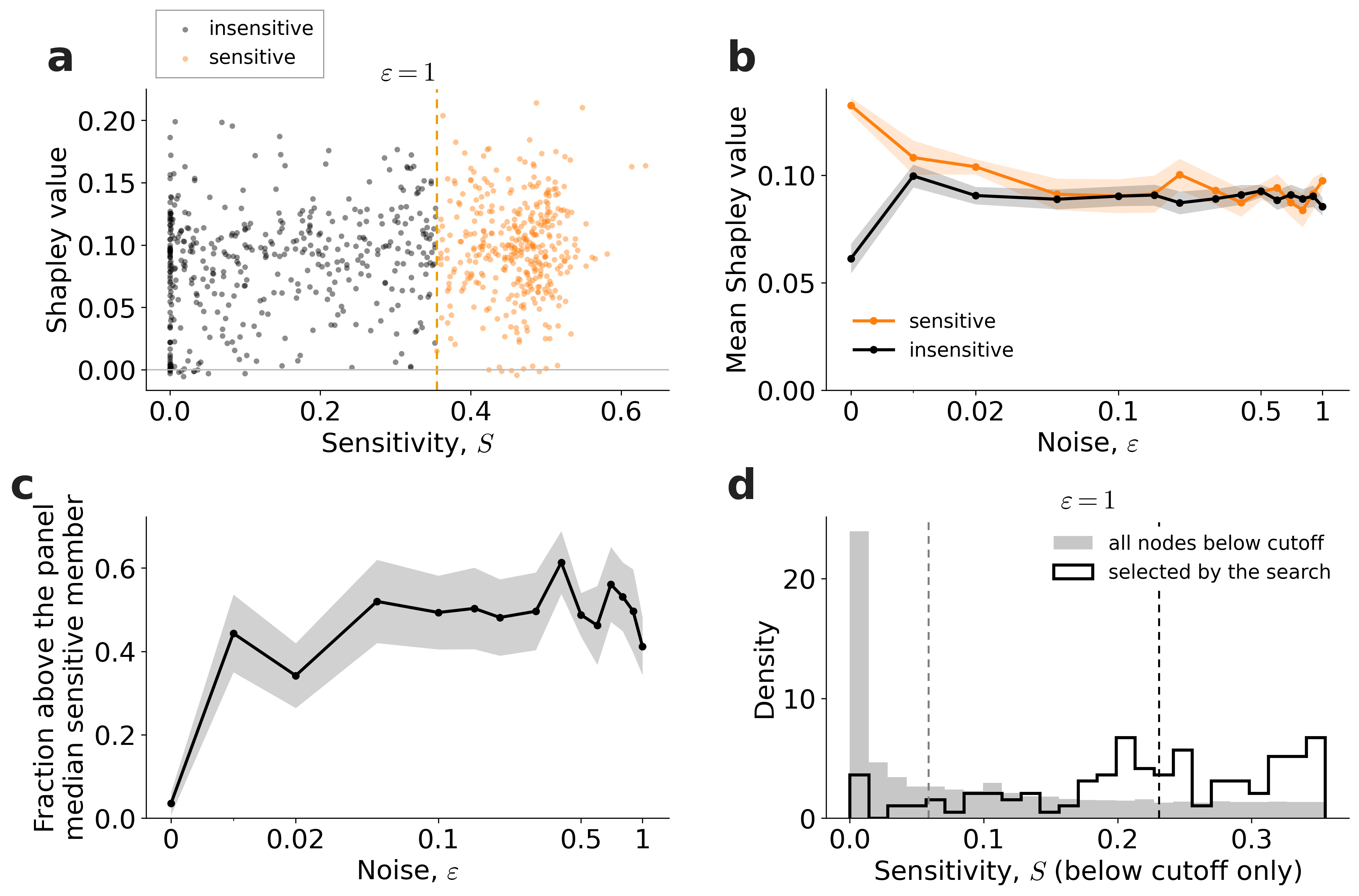}
\caption{
\textbf{Exact Shapley credit assignment within evolved panels.}
\textbf{a}, Shapley value against sensitivity for every member of every evolved panel at $\varepsilon = 1$, with the sensitivity cutoff dashed.
Insensitive members carry credit comparable to sensitive members.
The insensitive group is the whole pool below the cutoff, so the points at $S = 0$ are unresponsive members.
\textbf{b}, Mean Shapley value by sensitivity class across every noise level. %
The classes converge by $\varepsilon \approx 0.05$ and remain together thereafter.
\textbf{c}, Fraction of insensitive members whose Shapley value exceeds the median sensitive member of their own panel.
\textbf{d}, Where the genetic algorithm recruits below the sensitivity cutoff at $\varepsilon = 1$.
Selected members (outline) against all nodes below the cutoff (filled), with medians dashed.
Bands and intervals are $95\%$ confidence intervals across panels.
}
\label{fig:si-shapley}
\end{figure}

\noindent\textbf{Classifying reporters by response breadth.}

The sensitivity $S_j$ averages a node's response over the $10$ shocks,
so a node that answers every shock weakly and a node that answers one shock enormously can share the same $S$.
We therefore classify nodes by how many shocks they answer.
A node answers shock $q$ when its deviation reaches the sensitivity cutoff,
$\Delta_{j,q} \geq \theta$, where $\theta$ is the antimode of the pooled $S$ distribution defined in the main text Methods section.
Writing $n_j$ for the number of shocks a node answers,
\begin{equation}
n_j := \bigl| \{ q : \Delta_{j,q} \geq \theta \} \bigr|.
\end{equation}
A node is unresponsive when $n_j = 0$, dormant when $1 \leq n_j \leq 5$, and promiscuous when $n_j \geq 6$.
The rule introduces no threshold beyond $\theta$, which the sensitivity classes already use.
It also enforces the requirement that a dormant node must answer something strongly.
A node whose deviations never reach $\theta$ cannot be dormant however concentrated its profile is.

The majority split at $6$ is not arbitrary.
The distribution of $n_j$ is ``U'' shaped with an interior minimum, and that minimum sits at $n_j = 6$ or $7$ at every noise level (\cref{fig:si-entropy}a).
Pooled over all $50$ networks, $49.3\%$ of nodes are unresponsive, $29.5\%$ dormant, and $21.2\%$ promiscuous at $\varepsilon = 1$, with the same ordering at $\varepsilon = 0$ and $\varepsilon = 0.5$.
The unresponsive class pools two kinds of node that the measurement cannot separate.
At $\varepsilon = 1$, $15.2\%$ of nodes never deviate at all, while $34.1\%$ deviate under some shock but never far enough to reach $\theta$.
Both are useless as reporters, so pooling them costs nothing.
An unresponsive node never crosses the threshold.
It is not necessarily a node that never moves.

Breadth and mean sensitivity agree on $97\%$ of nodes.
We have that $79\%$ of the disagreements lie within $0.05$ of the cutoff (\cref{fig:si-entropy}b).
So the reclassification sharpens the main text results.
Where the two rules differ is exactly where the mean is least informative:
in the band around $\theta$ in which a moderate average can come either from many small responses or from a few large ones.

Composition measured this way reproduces the main text result,
provided the comparison is made against the rule the paper actually argues with (\cref{fig:si-entropy}c,d).
Selecting the $8$ most sensitive nodes yields $7.96$ promiscuous members and $0.04$ dormant members per panel.
The most sensitive heuristic essentially never picks a selective node, and the panel it produces is $8$ promiscuous nodes whose profiles overlap heavily.
Evolved panels instead carry $4.34$ promiscuous, $3.40$ dormant, and $0.26$ unresponsive members at $\varepsilon = 1$, about $85$ times the dormant count of the heuristic.
The same ordering holds at $\varepsilon = 0$ and $\varepsilon = 0.5$.
Evolved panels also carry more dormant members than random panels do, $3.40$ against $2.24$.

\begin{figure}
\centering
\includegraphics[width=0.98\textwidth]{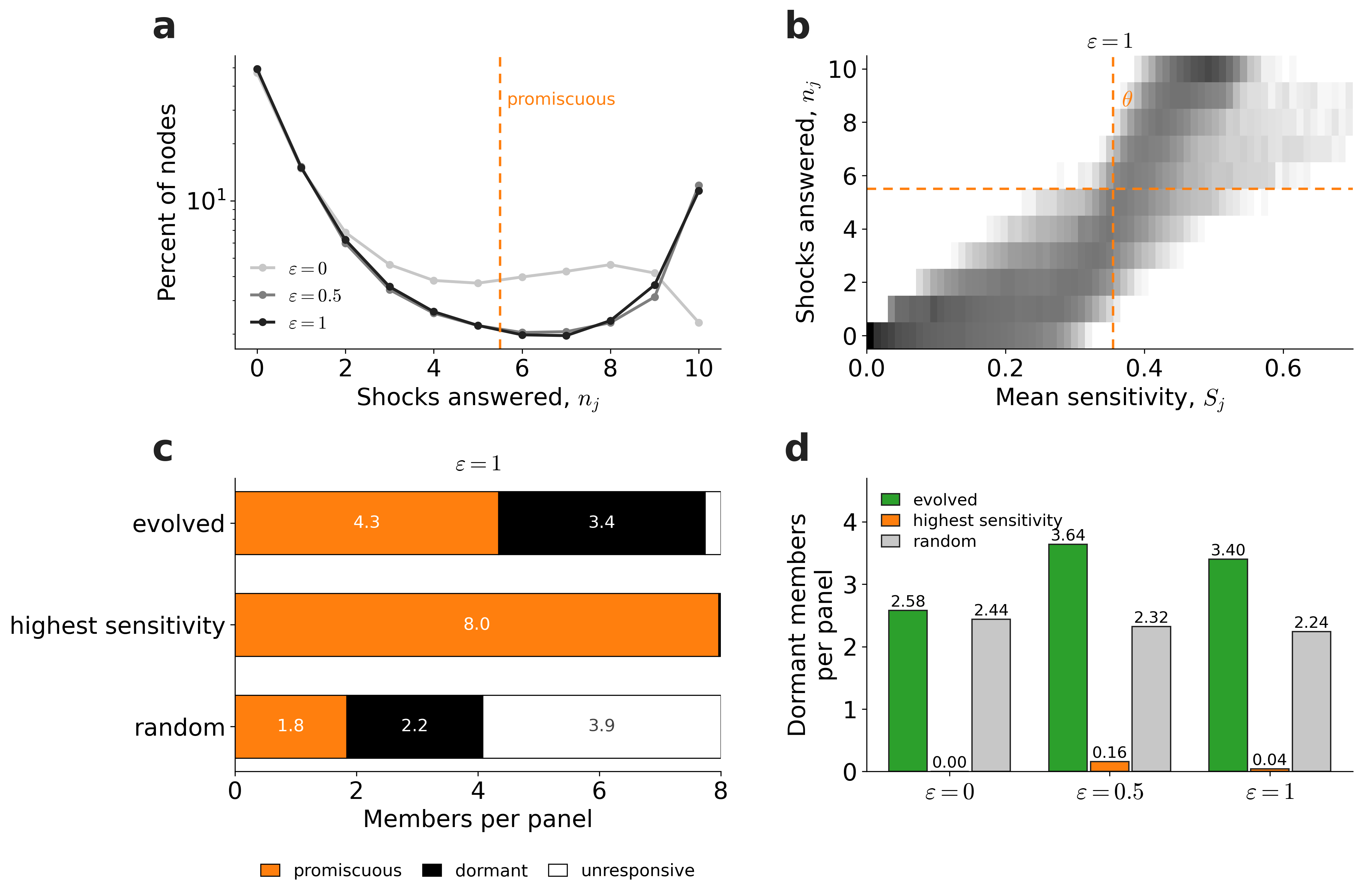}
\caption{
\textbf{Reporters classified by response breadth rather than by mean response.}
A node answers shock $q$ when $\Delta_{j,q} \geq \theta$, and $n_j$ counts the shocks it answers.
\textbf{a}, Distribution of $n_j$ across all nodes of all $50$ networks, at the $3$ noise levels, on a logarithmic vertical scale.
The distribution is ``U'' shaped and its interior minimum sits at the majority split, marked dashed, so the boundary between dormant and promiscuous falls in a valley.
\textbf{b}, Mean sensitivity against breadth at $\varepsilon = 1$, on a logarithmic density scale, with $\theta$ and the majority split marked.
The two classifications agree on $97\%$ of nodes, and the disagreements concentrate in the band around $\theta$ where a moderate mean can come either from many small responses or from a few large ones.
\textbf{c}, Mean composition of an $8$ member panel under $3$ selection rules at $\varepsilon = 1$.
The sensitivity heuristic returns almost no dormant members.
The genetic algorithm returns more than $3$ per panel.
\textbf{d}, Dormant members per panel at $3$ noise levels, for the same $3$ rules.
The genetic algorithm beats both baselines at every noise level, and the sensitivity heuristic, not random selection, is the informative one.
}
\label{fig:si-entropy}
\end{figure}

\noindent\textbf{Do dormant reporters divide the shock space?}

A dormant node answers a minority of the $10$ shocks.
So a panel of dormant nodes could in principle partition the shock space, each member covering what the others miss.
Whether the evolved panels partition the shock space depends entirely on the null comparison. %

Write $A_j$ for the set of shocks node $j$ answers, so $n_j = |A_j|$.
The Jaccard index between two node sets $A_i$ and $A_j$ is defined as
\begin{equation}
  J(A_i, A_j) := \frac{|A_i\cap A_j|}{|A_i\cup A_j|}.
\end{equation}
The Jaccard index measures the similarity or overlap between two sets.
The obvious null experiment replaces each dormant member's answered shocks by a uniformly random subset of the $10$ shocks of the same size.
Against that null nothing happens: the union and overlaps of the set of shocks are indistinguishable from chance at every noise level.
That null is not insightful.
The $10$ shocks are not interchangeable because some are answered by far more dormant nodes than others. %

We therefore compared against $2$ better nulls.
The \emph{matched shock} null draws each replacement shock with its own popularity among that network's dormant nodes.
The \emph{dormant pool} null draws replacement members from the network's own dormant pool with the same $n_j$.
Against both nulls the dormant members of an evolved panel overlap less than chance
(\cref{fig:si-coverage}c).
Coverage is larger in the evolved panels versus the nulls 
(\cref{fig:si-coverage}b).

Dormant reporters also do not target the shocks that the promiscuous members confuse.
Call a shock pair confusable when every promiscuous member of the panel assigns the two shocks the same sensitivity cutoff signature.
In other words,
if $P=\{j_1,\ldots,j_p\}\subseteq R$ are the promiscuous members of reporter panel $R$,
and $C_q:=(\bm1\{\Delta_{j_\ell,q}\geq \theta\})_{\ell=1}^{p}$ is the cutoff signature for shock $q$,
then $(q_1, q_2)$ is confusable when $C_{q_1} = C_{q_2}$.
Let the specialization index be the fraction of confusable pairs the dormant members separate minus
the fraction of already separable pairs they separate.
The specialization index is $+0.001$, $-0.001$, and $+0.027$ at $\varepsilon = 0$, $0.5$, and $1$.
In contrast, a control that picks dormant class nodes greedily to cover the confusable pairs reaches $+0.27$, $+0.13$, and $+0.16$ (respectively).

Thus dormant reporters are not a second discriminator, consistent with the reserve picture.
If dormant reporters were a second discriminator they should concentrate on the pairs the first discriminator (promiscuous reporters) cannot separate, and they do not.

\begin{figure}
\centering
\includegraphics[width=0.98\textwidth]{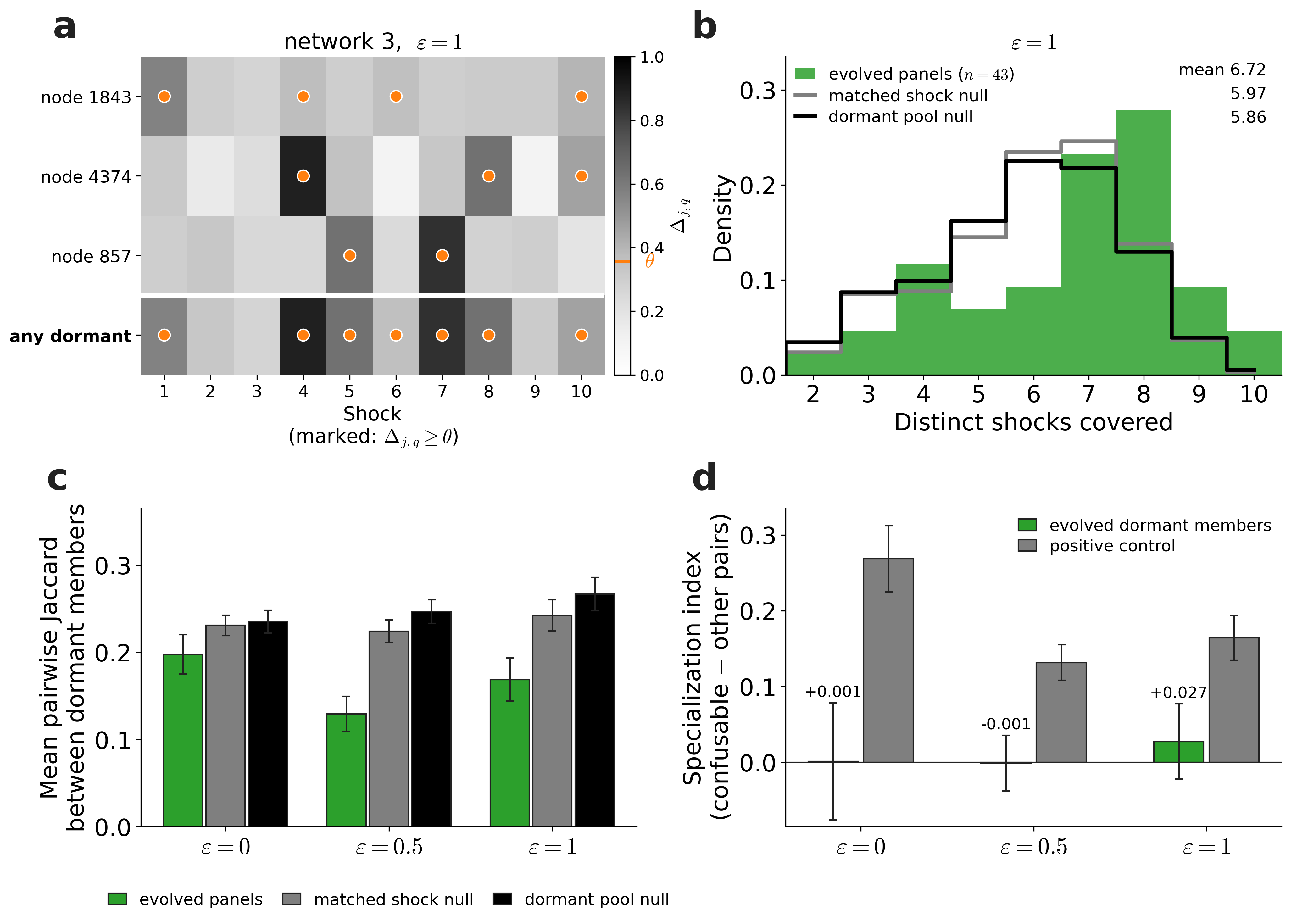}
\caption{
\textbf{Dormant reporters spread over the shock space without tiling it.}
\textbf{a}, An example network at $\varepsilon = 1$, chosen as the panel whose union and overlap are closest to the medians: its dormant members against the $10$ shocks, shaded by the deviation $\Delta_{j,q}$, with a marker on every cell reaching $\theta$.
The bottom row is the union of the rows.
\textbf{b}, Distinct shocks answered by at least $1$ dormant member, across networks, against both nulls.
\textbf{c}, Mean pairwise Jaccard between the answered sets, observed against both nulls, at the $3$ noise levels.
\textbf{d}, The specialization index against a greedy positive control.
Bars and intervals are means with $95\%$ confidence intervals across panels.
}
\label{fig:si-coverage}
\end{figure}

\subsection{How hard is choosing the panel?}

Every panel in this paper was chosen by a genetic algorithm rather than by exact optimization.
We state the selection task as a decision problem.
We prove that it is \NPcomplete.
We then say what that does and does not imply for the ensemble we study.

\medskip
\noindent\textbf{Two decision problems.}
We first recall the problem we reduce from \cite{gareyjohnson1979computers}.

\begin{quote}
\textsc{Minimum Test Collection} (\textsc{MTC}).\\
\emph{Instance}: a finite set $S$ of items, a collection $\mathcal{C} = \{C_1, \dots, C_N\}$ of subsets of $S$ called tests, and a positive integer $m \leq N$.\\
\emph{Question}: is there a subcollection $\mathcal{C}' \subseteq \mathcal{C}$ with $|\mathcal{C}'| \leq m$ such that for every pair of distinct items $s, s' \in S$ there is some $C \in \mathcal{C}'$ containing exactly one of $s$ and $s'$?
\end{quote}

\noindent A subcollection with that property is said to \emph{separate} $S$.
We now state the selection task in the same style.
The input is the response matrix, that is, the deviations a pilot experiment measures, rather than the network itself; the distinction matters and we return to it below.

\begin{quote}
\textsc{Reporter Panel}.\\
\emph{Instance}: a number of shocks $d$, a number of candidate reporters $N$, a response matrix $A \in \{0,1\}^{N \times d}$ with $A_{jq} = 1$ when node $j$ answers shock $q$, a panel size $m \leq N$, and a rational threshold $\eta \geq 0$.\\
\emph{Question}: is there a set $P \subseteq [N]$ with $|P| = m$ such that $H(Q \mid X_P) \leq \eta$, where $Q$ is uniform on the $d$ shocks and $X_P = (A_{jQ})_{j \in P}$ is what the panel observes?
\end{quote}

\noindent Conditional entropy is the quantity a classifier estimates, and its accuracy is proxy for it.

\begin{theorem}
\textup{\textsc{Reporter Panel}} is \NPcomplete.
\end{theorem}

\begin{proof}
\emph{Membership in \NP.}
Take the panel $P$ itself as the certificate, which has size $m \log N$.
To verify it, group the $d$ shocks by their response signatures (i.e. panel observations):
shocks $q$ and $q'$ fall in the same group exactly when $A_{jq} = A_{jq'}$ for every $j \in P$.
Since $Q$ is uniform and $X_P$ determines the group, the posterior on $Q$ given a response signature is uniform on that group.
So if the groups have sizes $n_1, \dots, n_k$ then using the map $G$ that maps the response signature to a shock,
\begin{align}
H(Q \mid X_P) &= H(Q \mid G(X_P)) \\
&= -\sum_{i=1}^k\sum_{q=1}^dP(G(X_p)=g_i)\times P(Q=q\mid G(X_p)=g_i)\times \log P(Q=q\mid G(X_p)=g_i)  \\
&= -\sum_{i=1}^k\sum_{q\in g_i}(n_i/d)\times (1/n_i)\times \log 1/n_i  \\
&=\frac{1}{d}\sum_{i=1}^{k} n_i \log n_i .
\label{eq:condent}
\end{align}
Forming the groups takes $O(md\log d)$ time by sorting the $d$ response signatures restricted to $P$.
Evaluating \cref{eq:condent}
and comparing it to $\eta$ takes $O(d)$ further arithmetic operations on rationals.
Verification is therefore polynomial in the input size, so the problem is in \NP.

\emph{Proof of \NPhardness.}
We reduce from \textsc{MTC}.
Let $(S, \mathcal{C}, m)$ be an instance with $|S| = d$ and $|\mathcal{C}| = N$, and write $S = \{s_1, \dots, s_d\}$.
Construct the \textsc{Reporter Panel} instance with $d$ shocks, $N$ candidate reporters, panel size $m$, threshold $\eta = 0$, and
\begin{equation}
A_{jq} = \begin{cases} 1 & s_q \in C_j, \\ 0 & \text{otherwise.} \end{cases}
\end{equation}
The construction writes $Nd$ bits and so runs in time polynomial in the size of the \textsc{MTC} instance.

By \cref{eq:condent}, $H(Q \mid X_P) = 0$ holds exactly when every group is a singleton, since each term $n_i \log n_i$ is non-negative and vanishes only for $n_i = 1$.
Every group being a singleton is exactly the statement that the map $q \mapsto (A_{jq})_{j \in P}$ is injective, which in turn says that for every pair $q \neq q'$ there is some $j \in P$ with $A_{jq} \neq A_{jq'}$.
We will show that \textsc{MTC} is a \texttt{YES} instance if and only if \textsc{Reporter Panel} is a \texttt{YES} instance.

Suppose the \textsc{MTC} instance is a \texttt{YES} instance, with $\mathcal{C}'$ separating $S$ and $|\mathcal{C}'| \leq m$.
Let $P$ be the set of reporters indexed by $\mathcal{C}'$, padded with arbitrary further indices to size exactly $m$, which is possible because $m \leq N$.
Padding is harmless because adding coordinates cannot merge two distinct response signatures.
For any $q \neq q'$ the test separating $s_q$ from $s_{q'}$ contains exactly one of them, so the corresponding row of $A$ differs in those two columns, and $H(Q \mid X_P) = 0 \leq \eta$.

Conversely suppose $P$ with $|P| = m$ achieves $H(Q \mid X_P) = 0$.
Then the map above is injective, so for every pair $s_q \neq s_{q'}$ some $j \in P$ has $A_{jq} \neq A_{jq'}$, which says that $C_j$ contains exactly one of $s_q$ and $s_{q'}$.
The subcollection indexed by $P$ therefore separates $S$ and has size $m$.

The two instances are thus equivalent, and \textsc{MTC} is \NPcomplete \cite{gareyjohnson1979computers}, so \textsc{Reporter Panel} is \NPhard.

Being both in \NP\ and \NPhard, \textsc{Reporter Panel} is \NPcomplete.
\end{proof}

\noindent We give $3$ remarks on the scope of the theorem.

Firstly, the classification task in the main text carries a control condition as an eleventh alternative, which the statement above omits.
Including it only adds the requirement that no shock produce the all-zero response signature, so it constrains the panel further and cannot make the problem easier.

Secondly, and more importantly, the input here is the measured response matrix, not the weight matrix $W$.
Taking $W$ as the input instead gives a problem that is at least as hard in practice but that we do not claim is in \NP, because verifying a candidate panel would require the distribution of trajectories induced by the noisy initial condition, an average over $2^N$ initial states.
The theorem therefore describes the layer an experimenter actually optimizes, namely which of the measured reporters to keep, and not the physical problem of reading the answer off the network.

Thirdly, in the actual problem of this paper, our instances are not adversarial.
The weights are drawn independently and uniformly on $[-1,1]$ and the shock targets uniformly at random.
So the theorem bounds what an algorithm must handle in general rather than what it faces on a typical draw from this model.
Whether a polynomial time algorithm finds a near optimal panel on such a draw is a question we do not resolve.
What we observe is that a genetic algorithm run for $30$ generations beats every polynomial time heuristic we implemented,
that the margin grows with noise,
and that the heuristics coming closest are the ones using response statistics rather than topology.

\subsection{What makes a panel good, and what does not}

A natural explanation of mixed panels is a division of labor.
Dormant nodes would report the state of the system.
Promiscuous nodes would then report the shock, given that state.
Our measurements rule this out.

Consider only which members answered, and ignore how strongly.
A panel can distinguish at most $\log_2 10 = 3.32$ bits this way.
A panel of $8$ dormant nodes chosen for separation resolves $9.78$ of the $10$ shocks.
That is $3.29$ bits, or almost complete identification.
The $8$ most sensitive nodes resolve $2.56$ shocks, or $1.36$ bits.
So dormant reporters carry nearly all the information a thresholded readout can carry.
A promiscuous node answers most shocks.
Whatever it contributes lives in how strongly it responds.

No statistic of the mean response orders the panels by accuracy (\cref{fig:si-panel-quality}a).
The panel designed for margin holds the largest worst case margin, $0.865$ against $0.403$ for the evolved panels.
It also holds the most thresholded information, $3.29$ bits against $2.80$.
It still classifies worse, $0.757$ against $0.832$.
The evolved panels beat it in $66\%$ of networks.
A second dormant design covers all $10$ shocks in every network.
It reaches only $0.550$.
Panels of $8$ dormant nodes drawn at random reach $0.507$.
Across all $300$ panels the worst case margin predicts accuracy only moderately, at $\rho = 0.50$.
The count of distinct thresholded patterns predicts it barely at all, at $\rho = 0.19$.

This rules out any account written in terms of the mean response alone.
We therefore measured what the mean response leaves out.

We first measured reproducibility, node by node.
For each node we took the spread of its reading across replicates.
We then took the standardized separation between every pair of alternatives.
Dormant reporters match the promiscuous mean separation, $0.193$ against $0.198$.
They do so with about half the trial to trial spread, $0.073$ against $0.126$.
Their standardized separation is therefore higher, $2.09$ against $1.82$.
Reliability alone is still not enough.
Combining these single node values as if members were independent puts the margin design above the evolved one, $2.50$ against $2.07$.
That is the wrong order.
The independent combination also tracks accuracy only at $\rho = 0.59$.
Redundancy does not explain the gap either.
The mean absolute correlation between members is about the same in the two designs, $0.199$ against $0.193$.

We then measured the same separation with the full member covariance.
The $10^{\text{th}}$ percentile over pairs orders all $6$ designs exactly as accuracy does.
It correlates with accuracy at $\rho = 0.84$ across the $300$ panels.
The best statistic of the mean response reaches only $\rho = 0.50$.
The result holds as the covariance shrinkage varies from $0.2$ to $0.8$.

The off diagonal terms of the covariance are the natural candidate \cite{averbeck2006neural,morenobote2014information}.
An exact attribution rules them out.
We split the off diagonal into $3$ blocks and computed each block's Shapley contribution.
The whole off diagonal contributes about $1\%$ of the statistic.
Most of it sits between pairs of dormant members.
The block between pairs of promiscuous members is the smallest.
The diagonal alone predicts accuracy exactly as well as the full covariance, $\rho = 0.835$ for both.

The effect is a normalization rather than a covariance.
The statistic that works divides each member's separation by its variance pooled across all conditions.
The statistic that fails divides by the variance of the two conditions being compared.
A dormant reporter is quiet in most conditions and variable where it flips.
Judged against its local noise it looks unreliable exactly when it acts.
Judged against its overall variability it barely fluctuates at all.
The classifier follows the pooled normalization.

What makes a panel good is not a property of its members one at a time.
It is not sensitivity, breadth, coverage, margin, or any single member's reliability.
It is whether the panel keeps every pair of alternatives separated relative to its members' overall variability.
The binding constraint is the worst separated pairs.
That constraint cannot be evaluated for one node in isolation.
The genetic algorithm only ever sees accuracy, and it finds panels that satisfy the criterion.
No rule that ranks nodes one at a time can do the same.
The set level version of the quantity does work as a selection rule.
That rule is the spring rule of main text \textbf{Figure 7}.

One detail matters at zero noise.
With a perfectly reproduced initial condition the across replicate variance is zero.
The standardized separation then degenerates.
The unregularized rule collapses to chance level selection, at $0.51$.
Including the member covariance with a small ridge regularizes the denominator.
The rule then reaches $0.97$, $0.92$, and $0.90$.
That is within a few points of the search at every noise level.

\begin{figure}
\centering
\includegraphics[width=0.98\textwidth]{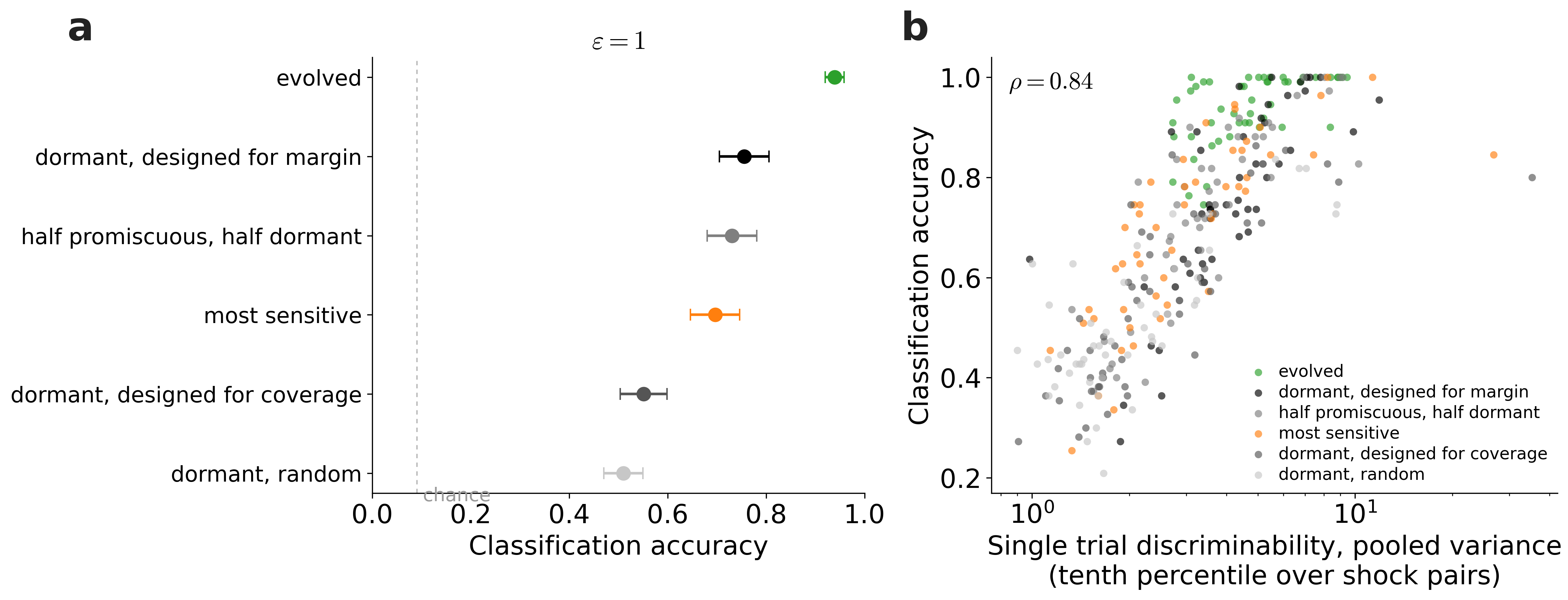}
\caption{
\textbf{Panel accuracy tracks a set level discriminability, not any member property.}
\textbf{a}, Classification accuracy at $\varepsilon = 1$ for $6$ ways of building an $8$ member panel, means with $95\%$ confidence intervals across the $50$ networks.
\textbf{b}, Accuracy against the $10^{\text{th}}$ percentile over shock pairs of the panel's single trial discriminability with each member's separation normalized by its trial to trial variance pooled across conditions, one point per network and design.
Spearman $\rho = 0.84$ across the $300$ panels, and the mean of this statistic orders the $6$ designs exactly as their accuracy does.
The off diagonal of the member covariance contributes about $1\%$ of the statistic, so the prediction is carried by the pooled normalization and not by correlations between members.
}
\label{fig:si-panel-quality}
\end{figure}

\section{Experimental considerations}

\subsection{Selection heuristics}

Main text \textbf{Figure 7} compares the leading selection strategies and compresses the remaining heuristics into a faint gray color.
Every curve reports the accuracy of finished panels rescored on fresh random splits.
The genetic algorithm averages $10$ splits per fitness evaluation and rescores its elites every generation.
So its reported best runs only about $2$ points above the rescored value at full noise.
The genetic algorithm leads at every noise level, $1.00$, $0.95$, and $0.94$ (respectively) at $m = 8$,
with greedy information gain within $2$ to $3$ points at a small fraction of the evaluations.
\Cref{fig:si-strategies-full} shows every strategy we evaluated at all $3$ noise levels and all panel sizes.

The heuristics separate into $3$ tiers.

Structural rankings that use only the topology track the random baseline or fall below it:
highest in degree, highest out degree, upstream coverage, Jaccard coverage, and greedy influence maximization.
The two degree rules take the top of the distribution,
the most regulated nodes and the largest hubs.

Rankings built from the dynamics do substantially better.
The highest sensitivity panel reaches accuracy $0.90$, $0.71$, and $0.70$ at $m=8$ for $\varepsilon = 0$, $0.5$, and $1$, respectively.
Adding a pairwise redundancy penalty buys a small further gain, to $0.92$, $0.72$, and $0.73$, respectively.
That rule---which we call ``entropy with diversity''---adds at each step the node maximizing $H_j - \beta \overline{I}_j$,
where $H_j$ is the marginal entropy of node $j$'s state in the control condition and
$\overline{I}_j$ is its mean pairwise mutual information with the nodes already chosen.
It reads the control dynamics only and never sees a shock.
A third rule of the same family---which we call ``greedy minimum mean squared error (MMSE)''---%
selects control snapshot covariance columns greedily and lands between the structural tier and the sensitivity ranking.

The decisive ingredient is conditioning,
meaning scoring a candidate by what it adds given the nodes already chosen.
``Greedy information gain'' is one conditional strategy that adds at each step the node that most reduces the conditional entropy $H(\text{shock} \mid \text{panel})$.
Note that this is a different quantity from the one the previous rule maximizes:
entropy with diversity scores the marginal entropy of a member's own control state,
whereas information gain scores the entropy of the shock label given the whole panel.
A second conditioning strategy reserves a fraction of its picks for \emph{quiet} nodes before conditioning on the rest.
A quiet pick maximizes $-S_j - \beta \overline{I}_j$, so it takes the least responsive node available while staying decorrelated from those already held. %
These two conditional strategies reach $0.97$ and $0.89$ at $\varepsilon=0$ and $0.92$ and $0.81$ at $\varepsilon=1$, each respectively.

The figure also contains an instructive failure.
The strategy ``reserve least sensitive nodes'' is the simplest implementation of the reservation idea.
At panel size $m$ it reserves a $25\%$ of the picks for nodes of minimal sensitivity $S$,
chosen with a mutual information diversity penalty,
and fills the rest with the most sensitive nodes.
It performs worse than the highest sensitivity panel alone,
reaching only $0.77$, $0.59$, and $0.56$, respectively for noise levels, at $m=8$.
The reason is that minimizing $S$ does not find dormant reporters.
About $15\%$ percent of nodes have $S$ exactly zero because they are frozen in every condition---%
these are unresponsive nodes which carry almost no information. %
The dormant nodes that evolved panels and the conditional strategies recruit are different.
They respond weakly to shocks yet remain active and reproducible in the control condition.
Low sensitivity alone is therefore not the criterion for a dormant reporter.
What makes one valuable is a reproducible and informative control state, and identifying it requires looking at the control dynamics, not only at the response magnitude.

Every heuristic in \cref{fig:si-strategies-full} is implemented in \texttt{scripts/heuristic-node-selection.py} of the code repository, selected by its \texttt{-{}-strategy} flag.
The genetic algorithm, the random baseline, and the spring rule are implemented in \texttt{scripts/genetic-algorithm-selection.py}, \texttt{scripts/random-node-selection.py}, and \texttt{scripts/rule-prefix-curves.py} (respectively).
All of them are scored through the same evaluator in \texttt{scripts/classifier.py}.

\begin{figure}
\centering
\includegraphics[width=0.98\textwidth]{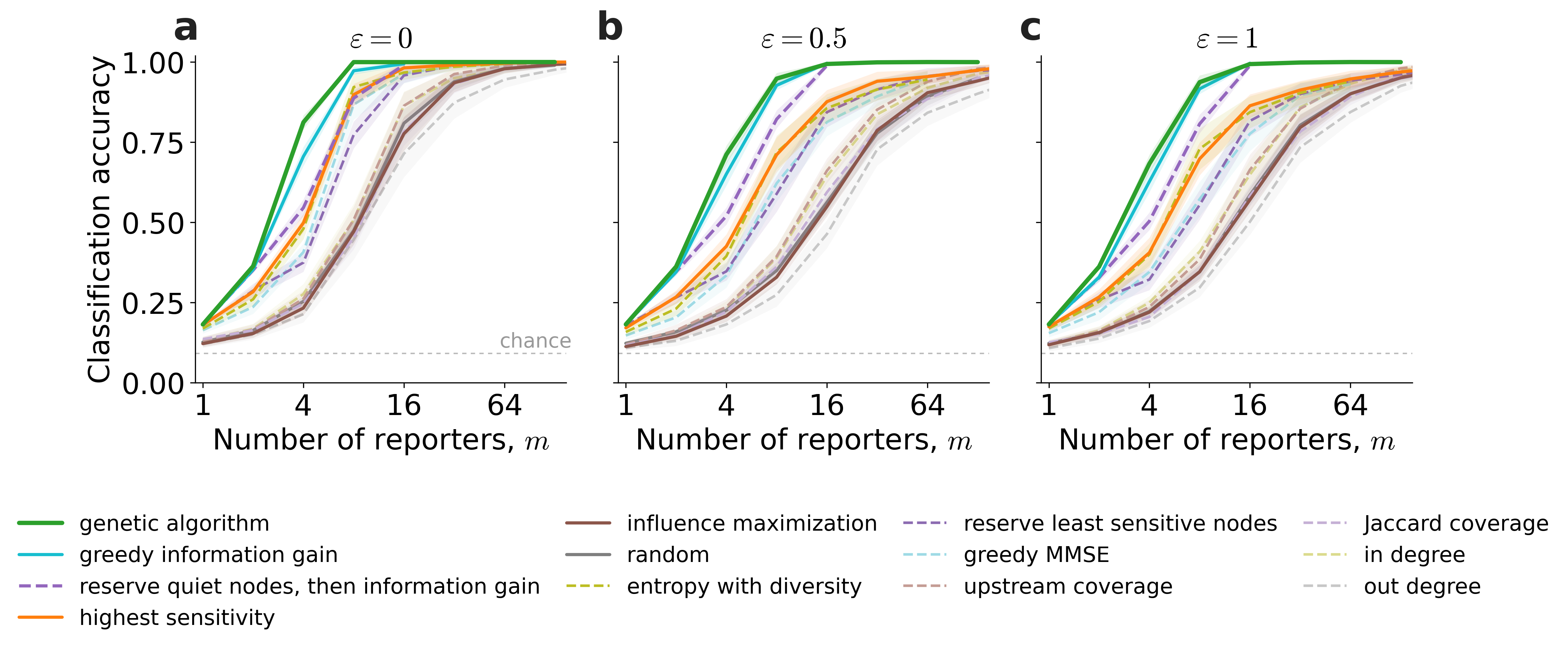}
\caption{
\textbf{All selection heuristics at all panel sizes.}
Classification accuracy against panel size $m$ at $3$ noise levels.
Curves are means over $50$ networks and bands are $95\%$ confidence intervals.
Solid curves repeat the strategies of main text \textbf{Figure 7}.
Dashed curves are the heuristics that the main figure shows in gray.
The entropy based strategies stop at $m=16$ where the estimator they optimize is reliable.
}
\label{fig:si-strategies-full}
\end{figure}

\subsection{Reporter selection in experimental practice}

In practice, reporter placement is usually settled before any quantitative comparison of candidates.
The stated rationales fall into a few recurring patterns, which we group by domain, and \cref{fig:si-practice} summarizes.
In $18$ studies across $5$ domains,
exactly one selects its monitored units by classification performance.

\begin{figure}
\centering
\includegraphics[width=0.98\textwidth]{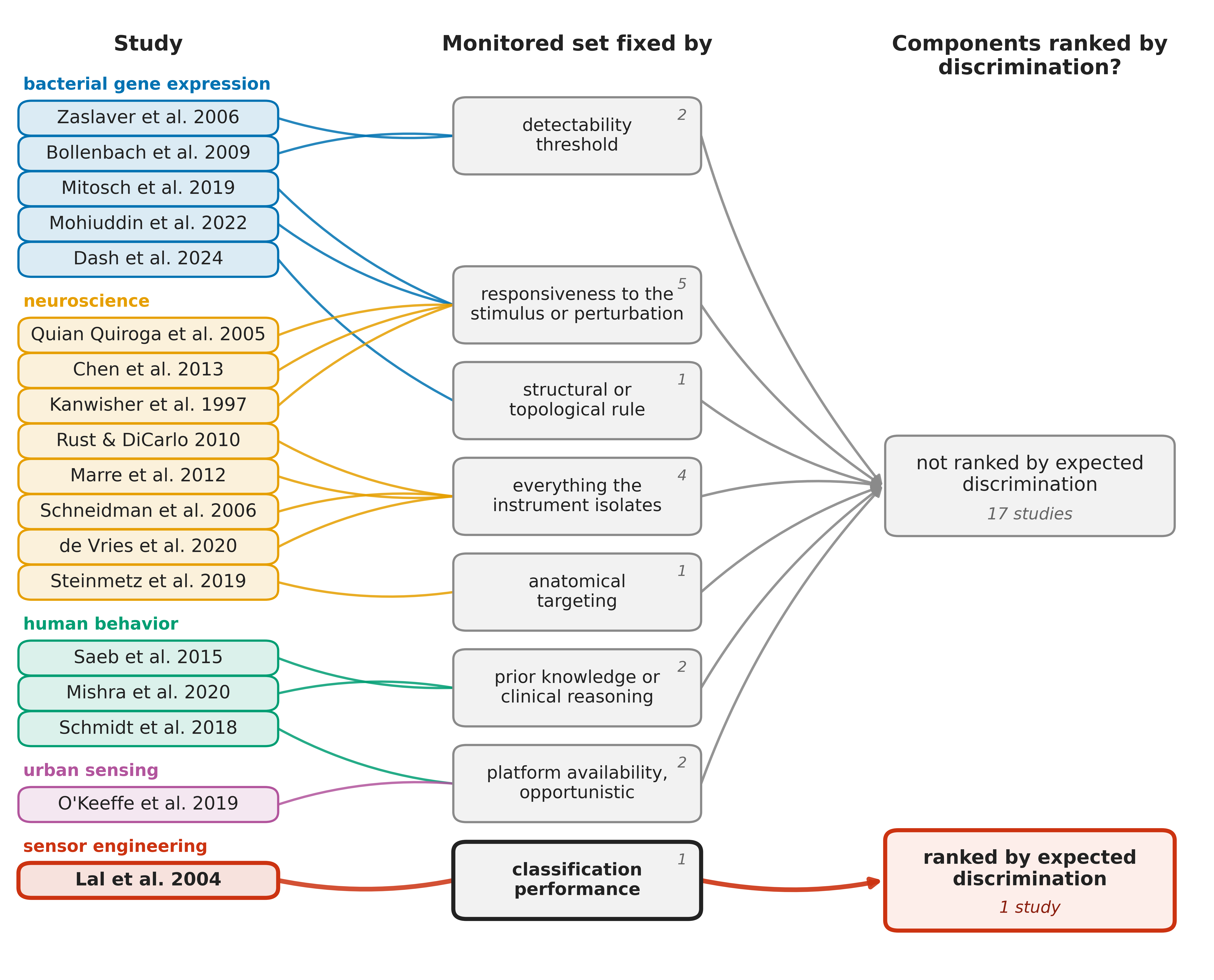}
\caption{
\textbf{How monitored sets are chosen in practice.}
We list $18$ studies across bacterial gene expression, neuroscience, human behavior, urban sensing, and sensor engineering, from top to bottom \cite{zaslaver2006library,bollenbach2009nonoptimal,mitosch2019temporal,mohiuddin2022high,dash2024library,quiroga2005invariant,chen2013ultrasensitive,kanwisher1997fusiform,rust2010selectivity,marre2012mapping,schneidman2006weak,devries2020large,steinmetz2019distributed,saeb2015mobile,mishra2020presymptomatic,schmidt2018introducing,okeeffe2019quantifying,lal2004support}.
Each study flows into the criterion that fixed its analyzed set.
Where a pipeline applies several filters in sequence, the node shown is the last one.
The full pipelines are given in the text.
The small number on each criterion node counts its studies.
The criteria then converge on the question this paper studies: whether components were ranked by their expected ability to discriminate the conditions of interest.
Only $1$ study of the $18$ ranks by discrimination.
}
\label{fig:si-practice}
\end{figure}

\medskip
\noindent\textbf{Biology.}
In bacterial gene expression,
the standard instrument is a promoter reporter library.
Zaslaver et al. built theirs by a structural rule, amplifying every intergenic region longer than $40$ bp \cite{zaslaver2006library}.
Detectability then narrows the analyzed set.
About $60\%$ of their strains gave fluorescence above background at all time points.
Bollenbach et al. took promoters from almost $200$ genes ``representing key cellular functions including DNA stress response, metabolism and ribosome regulation and synthesis'',
then kept only promoters with a clearly detectable signal, leaving $110$ \cite{bollenbach2009nonoptimal}.
Mitosch et al. selected promoters ``that were strongly induced in response to at least one of the antibiotics'',
and for chromosomal integration chose promoters expressed highly enough for reliable detection of dynamic changes \cite{mitosch2019temporal}.
Mohiuddin et al. screened the same library and defined hits by a twofold increase in expression \cite{mohiuddin2022high}.
Dash et al. used topology instead, defining their targets as the $22$ transcription factors ``with 40 or more target genes each'' \cite{dash2024library}.
On the reporter component itself, Lopreside et al. state that the choice ``was largely based on trial and error with very few systematic side-by-side investigations reported'' \cite{lopreside2019comprehensive}.

\medskip
\noindent\textbf{Neuroscience.}
In neuroscience, the recording site and the analyzed units follow different rules.
Quian Quiroga et al. worked with patients whose electrode placement ``was determined exclusively by clinical criteria'',
then screened for responsive units, and $132$ of $993$ units passed \cite{quiroga2005invariant}.
Chen et al. applied a brightness threshold and then a responsiveness threshold defined by significant stimulus driven fluorescence changes \cite{chen2013ultrasensitive}.
Kanwisher et al. used activation by faces over objects to define a region of interest in each subject, then ran new conditions there \cite{kanwisher1997fusiform}.

The opposite practice is also documented.
Rust and DiCarlo took care that ``any neuron whose waveform could be isolated would be recorded, regardless of baseline or visually elicited firing rate'' \cite{rust2010selectivity}.
Marre et al. recorded ``up to 95\% of the ganglion cells over the area of the array'' and reported cells with ``no measurable receptive field'' \cite{marre2012mapping}.
de Vries et al. decoded grating direction ``using all neurons, responsive and unresponsive'',
and report that about $34\%$ of neurons were not reliably responsive to any stimulus \cite{devries2020large}.
Steinmetz et al. targeted probes anatomically and note that ``neurons are included in analyses regardless of receptive field location'' \cite{steinmetz2019distributed}.

These $4$ studies decline a responsiveness filter.
None of them ranks components by expected discriminability.
The case against filtering has also been argued directly.
Analyses of cortical recordings report that discarding weakly tuned or untuned neurons loses information,
that decoding improves by about $10\%$ when all neurons are retained,
and that the threshold separating tuned from untuned units is largely arbitrary because the underlying distribution is unimodal \cite{zylberberg2017untuned}.
That argument and ours reach the same practical recommendation from different mechanisms.
There, untuned neurons carry no stimulus information on their own and contribute only because their variability is correlated with that of the tuned neurons,
so their contribution arrives through the ensemble's noise correlation structure rather than through any tuning of their own \cite{leavitt2017correlated}.
Here, dormant reporters carry shock specific information of their own,
because the dormant nodes that panels recruit flip completely under the one or two shocks they respond to,
and the covariance attribution above shows that correlations between members contribute almost nothing.
Thus the correlational channel is not the mechanism in this system.
Consistent with the information carrying reading, single trial analyses show that neurons lacking classical task responses are still informative about both the stimulus the animal received and the choice it made \cite{insanally2019spike}.

Even the joint analyses that come closest to our set level view take the recorded set as given.
Schneidman et al. showed that weak pairwise correlations coexist with strongly correlated states in groups of $10$ or more retinal cells,
a population level insight computed on whichever cells the electrode array isolated \cite{schneidman2006weak}.

\medskip
\noindent\textbf{Cities.}
Urban sensing states the placement tradeoff openly.
O'Keeffe et al. quantify the sensing power of taxi fleets, whose sensors ride wherever fares go.
The reach they measure is remarkable: over a single day, a random draw of just $10$ taxis passes over roughly one in three of Manhattan's street segments.
They also find a structural blind spot: segments ``never traversed by taxis'' are ``permanently out of reach'' of the fleet \cite{okeeffe2019quantifying}.
Coverage is set by traffic, not by what the city needs to distinguish. 
This is the urban version of placing reporters by response magnitude rather than by discriminability.

\medskip
\noindent\textbf{Human behavior.}
Human behavioral sensing fixes the signals in advance, from physiology and from what devices export.
Saeb et al. specified $10$ features from clinical reasoning,
focusing on ``behavioral markers related to movement through geographic space'' \cite{saeb2015mobile}.
Mishra et al. built detection on resting heart rate, citing earlier work that ``heart rate and skin temperature can be used to detect viral respiratory infections'' \cite{mishra2020presymptomatic}.
Schmidt et al. took their modality set from two commercial devices and compared $16$ combinations afterwards \cite{schmidt2018introducing}.
Formal selection does appear in Lal et al. who used recursive feature elimination to reduce the number of EEG channels \cite{lal2004support}.

\medskip
\noindent\textbf{Optimization and control.}
A separate literature treats placement as optimization.
Joshi \& Boyd choose $k$ of $m$ measurements to minimize ``the error in estimating some parameters'' \cite{joshi2009sensor}.
Summers et al. exploit submodularity of Gramian metrics over candidate placements \cite{summers2016submodularity}.
Liu et al. and Montanari et al. compute minimal sensor sets for structural and for functional observability \cite{liu2013observability,montanari2022functional}.
Haber et al. optimize a determinant of the observability Jacobian and state that they do not model measurement noise \cite{haber2018state}.
In each case the objective is reconstruction or control of the state.

Optimal experiment design in systems biology instead discriminates among rival models \cite{kreutz2009systems,melykuti2010discriminating,vanlier2014optimal}. 
Two studies come closest to our setting.
Brunton et al. optimize placement for classification, using labeled training data with a fixed set of classes \cite{brunton2016sparse}.
Cheng et al. compute ``the minimum number of sensor nodes to discriminate attractors'' under a worst case bound on flipped nodes \cite{cheng2021discrimination}.

\medskip
\noindent\textbf{The amount of prior knowledge an experimenter needs.}
Our panel design comparison also quantifies what each amount of prior knowledge buys, at $\varepsilon = 1$ and $m = 8$.
An experimenter who knows only the class labels and samples dormant reporters at random reaches $0.51$.
Knowing additionally which shocks each dormant node answers, and covering distinct shocks, reaches $0.55$.
Picking the most sensitive nodes with no class knowledge reaches $0.70$.
Mixing the classes in fixed proportion reaches $0.74$ to $0.76$.
The spring rule, which consumes all control statistics, reaches $0.90$.
The genetic algorithm which additionally uses shock statistics reaches $0.94$.
The gap between the label based recipes and the rule is the paper's central claim made practical.
Class labels are knowledge about nodes in isolation, and panel quality is not a property of isolated nodes.
Thus no amount of per-component statistics can replace an experiment that measures the candidates jointly in some capacity.

\medskip
\noindent\textbf{What we are not claiming.}
We do not claim that the reporter selection methods used in the literature we have examined are wrong.
Each study had its own goal.
Shock classification was rarely that goal.
Our claim is narrower.
Derivative work that reuses these libraries or future work that reuses the methods from these studies for the purpose of shock classification inherits a monitored set that was never ranked for that task.
Our results say such a set is unlikely to be optimal for it.

\subsection{Connection to experimental noise}

The noisy initial condition model is not meant to imply that all noise in real systems is literally a perturbation of the initial condition.
Rather, it is a minimal representation of an experimentally ubiquitous problem:
repeated trials differ before the response to the perturbation has fully developed.
In cell experiments, this variability can come from noise in gene expression, differences in cell cycle stage, variation in growth state, and microenvironmental fluctuations.
In neural recordings, it can arise from ongoing background activity and unobserved inputs.
In behavioral data, it can arise from latent context, finite sampling, and uncontrolled external factors.

The parameter $\varepsilon$ therefore measures how much of the preparation escapes control before the perturbation is applied.
The main result of the paper should be interpreted as a statement about measurement design when trials cannot be replicated exactly.
This interpretation is closely aligned with experimental settings in which the perturbation itself is well specified,
but the state before the perturbation is not perfectly reproducible
\cite{paulsson2004summing,paulsson2005models,lestas2010fundamental,golding2005realtime,wang2010robust}.

\processdelayedfloats
\clearpage

{\setlength{\bibsep}{0\baselineskip}

}

\begin{thebibliography}{10}

\bibitem{kauffman1969metabolic}
Stuart~A. Kauffman.
\newblock Metabolic stability and epigenesis in randomly constructed genetic
  nets.
\newblock {\em Journal of Theoretical Biology}, 22(3):437--467, 1969.

\bibitem{kauffman1969homeostasis}
Stuart~A. Kauffman.
\newblock Homeostasis and differentiation in random genetic control networks.
\newblock {\em Nature}, 224(5215):177--178, 1969.

\bibitem{may1972stable}
Robert~M. May.
\newblock Will a large complex system be stable?
\newblock {\em Nature}, 238(5364):413--414, 1972.

\bibitem{levin1998ecosystems}
Simon~A. Levin.
\newblock Ecosystems and the biosphere as complex adaptive systems.
\newblock {\em Ecosystems}, 1(5):431--436, 1998.

\bibitem{newman2003structure}
M.~E.~J. Newman.
\newblock The structure and function of complex networks.
\newblock {\em SIAM Review}, 45(2):167--256, 2003.

\bibitem{menczer2020first}
Filippo Menczer, Santo Fortunato, and Clayton~A. Davis.
\newblock {\em A First Course in Network Science}.
\newblock Cambridge University Press, Cambridge, 2020.

\bibitem{delriochanona2020supply}
R.~M. del Rio-Chanona, P.~Mealy, A.~Pichler, F.~Lafond, and J.~D. Farmer.
\newblock Supply and demand shocks in the {COVID-19} pandemic: An industry and
  occupation perspective.
\newblock {\em Oxford Review of Economic Policy}, 36(Supplement\_1):S94--S137,
  2020.

\bibitem{watts1998collective}
Duncan~J. Watts and Steven~H. Strogatz.
\newblock Collective dynamics of `small-world' networks.
\newblock {\em Nature}, 393(6684):440--442, 1998.

\bibitem{barabasi1999emergence}
Albert-L{\'a}szl{\'o} Barab{\'a}si and R{\'e}ka Albert.
\newblock Emergence of scaling in random networks.
\newblock {\em Science}, 286(5439):509--512, 1999.

\bibitem{strogatz2001exploring}
Steven~H. Strogatz.
\newblock Exploring complex networks.
\newblock {\em Nature}, 410(6825):268--276, 2001.

\bibitem{albert2002statistical}
R{\'e}ka Albert and Albert-L{\'a}szl{\'o} Barab{\'a}si.
\newblock Statistical mechanics of complex networks.
\newblock {\em Reviews of Modern Physics}, 74(1):47--97, 2002.

\bibitem{boccaletti2006complex}
Stefano Boccaletti, Vito Latora, Yamir Moreno, Mario Chavez, and D.-U. Hwang.
\newblock Complex networks: Structure and dynamics.
\newblock {\em Physics Reports}, 424(4--5):175--308, 2006.

\bibitem{pastor2015epidemic}
Romualdo Pastor-Satorras, Claudio Castellano, Piet Van~Mieghem, and Alessandro
  Vespignani.
\newblock Epidemic processes in complex networks.
\newblock {\em Reviews of Modern Physics}, 87(3):925--979, 2015.

\bibitem{aldana2003boolean}
Maximino Aldana.
\newblock Boolean dynamics of networks with scale-free topology.
\newblock {\em Physica D: Nonlinear Phenomena}, 185(1):45--66, 2003.

\bibitem{aldana2003natural}
Maximino Aldana and Philippe Cluzel.
\newblock A natural class of robust networks.
\newblock {\em Proceedings of the National Academy of Sciences},
  100(15):8710--8714, 2003.

\bibitem{oikonomou2006topology}
Panos Oikonomou and Philippe Cluzel.
\newblock Effects of topology on network evolution.
\newblock {\em Nature Physics}, 2(8):532--536, 2006.

\bibitem{parmer2022influence}
Thomas Parmer, Luis~M. Rocha, and Filippo Radicchi.
\newblock Influence maximization in boolean networks.
\newblock {\em Nature Communications}, 13:3457, 2022.

\bibitem{chalfie1994green}
Martin Chalfie, Yuan Tu, Ghia Euskirchen, William~W. Ward, and Douglas~C.
  Prasher.
\newblock Green fluorescent protein as a marker for gene expression.
\newblock {\em Science}, 263(5148):802--805, 1994.

\bibitem{cluzel2000ultrasensitive}
Philippe Cluzel, Michael Surette, and Stanislas Leibler.
\newblock An ultrasensitive bacterial motor revealed by monitoring signaling
  proteins in single cells.
\newblock {\em Science}, 287(5458):1652--1655, 2000.

\bibitem{golding2005realtime}
Ido Golding, Johan Paulsson, Scott~M. Zawilski, and Edward~C. Cox.
\newblock Real-time kinetics of gene activity in individual bacteria.
\newblock {\em Cell}, 123(6):1025--1036, 2005.

\bibitem{le2005realtime}
Thuc~T. Le, S{\'e}bastien Harlepp, C{\u a}lin~C. Guet, Kimberly Dittmar,
  Thierry Emonet, Tao Pan, and Philippe Cluzel.
\newblock Real-time {RNA} profiling within a single bacterium.
\newblock {\em Proceedings of the National Academy of Sciences},
  102(26):9160--9164, 2005.

\bibitem{wang2010robust}
Ping Wang, Lydia Robert, James Pelletier, Wei~Lien Dang, Fran{\c c}ois Taddei,
  Andrew Wright, and Suckjoon Jun.
\newblock Robust growth of {Escherichia coli}.
\newblock {\em Current Biology}, 20(12):1099--1103, 2010.

\bibitem{zaslaver2006library}
Alon Zaslaver, Anat Bren, Michal Ronen, Shalev Itzkovitz, Ilya Kikoin, Seagull
  Shavit, Wolfram Liebermeister, Michael~G. Surette, and Uri Alon.
\newblock A comprehensive library of fluorescent transcriptional reporters for
  {Escherichia coli}.
\newblock {\em Nature Methods}, 3(8):623--628, 2006.

\bibitem{potvintrottier2016synchronous}
Laurent Potvin-Trottier, Nathan~D. Lord, Glenn Vinnicombe, and Johan Paulsson.
\newblock Synchronous long-term oscillations in a synthetic gene circuit.
\newblock {\em Nature}, 538(7626):514--517, 2016.

\bibitem{balleza2018systematic}
Enrique Balleza, J.~Mark Kim, and Philippe Cluzel.
\newblock Systematic characterization of maturation time of fluorescent
  proteins in living cells.
\newblock {\em Nature Methods}, 15(1):47--51, 2018.

\bibitem{korobkova2004molecular}
Ekaterina Korobkova, Thierry Emonet, Jose M.~G. Vilar, Thomas~S. Shimizu, and
  Philippe Cluzel.
\newblock From molecular noise to behavioural variability in a single
  bacterium.
\newblock {\em Nature}, 428(6982):574--578, 2004.

\bibitem{paulsson2004summing}
Johan Paulsson.
\newblock Summing up the noise in gene networks.
\newblock {\em Nature}, 427(6973):415--418, 2004.

\bibitem{paulsson2005models}
Johan Paulsson.
\newblock Models of stochastic gene expression.
\newblock {\em Physics of Life Reviews}, 2(2):157--175, 2005.

\bibitem{lestas2010fundamental}
Ioannis Lestas, Glenn Vinnicombe, and Johan Paulsson.
\newblock Fundamental limits on the suppression of molecular fluctuations.
\newblock {\em Nature}, 467(7312):174--178, 2010.

\bibitem{hilfinger2011separating}
Andreas Hilfinger and Johan Paulsson.
\newblock Separating intrinsic from extrinsic fluctuations in dynamic
  biological systems.
\newblock {\em Proceedings of the National Academy of Sciences},
  108(29):12167--12172, 2011.

\bibitem{schneidman2006weak}
Elad Schneidman, Michael~J. Berry~II, Ronen Segev, and William Bialek.
\newblock Weak pairwise correlations imply strongly correlated network states
  in a neural population.
\newblock {\em Nature}, 440(7087):1007--1012, 2006.

\bibitem{watanabe2013pairwise}
Takamitsu Watanabe, Satoshi Hirose, Hiroyuki Wada, Yoshio Imai, Toru Machida,
  Ichiro Shirouzu, Seiki Konishi, Yasushi Miyashita, and Naoki Masuda.
\newblock A pairwise maximum entropy model accurately describes resting-state
  human brain networks.
\newblock {\em Nature Communications}, 4:1370, 2013.

\bibitem{insanally2019spike}
Michele~N. Insanally, Ioana Carcea, Rachel~E. Field, Chris~C. Rodgers, Brian
  DePasquale, Kanaka Rajan, Michael~R. DeWeese, Badr~F. Albanna, and Robert~C.
  Froemke.
\newblock Spike-timing-dependent ensemble encoding by non-classically
  responsive cortical neurons.
\newblock {\em eLife}, 8:e42409, 2019.

\bibitem{leavitt2017correlated}
Matthew~L. Leavitt, Florian Pieper, Adam~J. Sachs, and Julio~C.
  Martinez-Trujillo.
\newblock Correlated variability modifies working memory fidelity in primate
  prefrontal neuronal ensembles.
\newblock {\em Proceedings of the National Academy of Sciences},
  114(12):E2494--E2503, 2017.

\bibitem{wang2014studentlife}
Rui Wang, Fanglin Chen, Zhenyu Chen, Tianxing Li, Gabriella Harari, Stefanie
  Tignor, Xia Zhou, Dror Ben-Zeev, and Andrew~T. Campbell.
\newblock Studentlife: assessing mental health, academic performance and
  behavioral trends of college students using smartphones.
\newblock In {\em Proceedings of the 2014 ACM International Joint Conference on
  Pervasive and Ubiquitous Computing}, pages 3--14, New York, NY, 2014.
  Association for Computing Machinery.

\bibitem{hilfinger2016constraints}
Andreas Hilfinger, Thomas~M. Norman, Glenn Vinnicombe, and Johan Paulsson.
\newblock Constraints on fluctuations in sparsely characterized biological
  systems.
\newblock {\em Physical Review Letters}, 116(5):058101, 2016.

\bibitem{maclaren2025observing}
Neil~G. MacLaren, Baruch Barzel, and Naoki Masuda.
\newblock Observing network dynamics through sentinel nodes.
\newblock {\em Nature Communications}, 16:10211, 2025.

\bibitem{emonet2008relationship}
Thierry Emonet and Philippe Cluzel.
\newblock Relationship between cellular response and behavioral variability in
  bacterial chemotaxis.
\newblock {\em Proceedings of the National Academy of Sciences},
  105(9):3304--3309, 2008.

\bibitem{park2010interdependence}
Heungwon Park, William Pontius, Calin~C. Guet, John~F. Marko, Thierry Emonet,
  and Philippe Cluzel.
\newblock Interdependence of behavioural variability and response to small
  stimuli in bacteria.
\newblock {\em Nature}, 468(7325):819--823, 2010.

\bibitem{bollenbach2009nonoptimal}
Tobias Bollenbach, Selwyn Quan, Remy Chait, and Roy Kishony.
\newblock Nonoptimal microbial response to antibiotics underlies suppressive
  drug interactions.
\newblock {\em Cell}, 139(4):707--718, 2009.

\bibitem{wood2014uncovering}
Kevin~B. Wood, Kris~C. Wood, Satoshi Nishida, and Philippe Cluzel.
\newblock Uncovering scaling laws to infer multidrug response of resistant
  microbes and cancer cells.
\newblock {\em Cell Reports}, 6(6):1073--1084, 2014.

\bibitem{masuda2026tipping}
Naoki Masuda.
\newblock Detecting and forecasting tipping points from sample variance alone.
\newblock {\em PNAS Nexus}, 5(4):pgag126, 2026.

\bibitem{masuda2024anticipating}
Naoki Masuda, Kazuyuki Aihara, and Neil~G. MacLaren.
\newblock Anticipating regime shifts by mixing early warning signals from
  different nodes.
\newblock {\em Nature Communications}, 15(1):1086, 2024.

\bibitem{kempe2003maximizing}
David Kempe, Jon Kleinberg, and {\'E}va Tardos.
\newblock Maximizing the spread of influence through a social network.
\newblock In {\em Proceedings of the Ninth ACM SIGKDD International Conference
  on Knowledge Discovery and Data Mining}, pages 137--146, New York, NY, 2003.
  Association for Computing Machinery.

\bibitem{kitsak2010identification}
Maksim Kitsak, Lazaros~K. Gallos, Shlomo Havlin, Fredrik Liljeros, Lev Muchnik,
  H.~Eugene Stanley, and Hern{\'a}n~A. Makse.
\newblock Identification of influential spreaders in complex networks.
\newblock {\em Nature Physics}, 6(11):888--893, 2010.

\bibitem{morone2015influence}
Flaviano Morone and Hern{\'a}n~A. Makse.
\newblock Influence maximization in complex networks through optimal
  percolation.
\newblock {\em Nature}, 524(7563):65--68, 2015.

\bibitem{joshi2009sensor}
Siddharth Joshi and Stephen Boyd.
\newblock Sensor selection via convex optimization.
\newblock {\em IEEE Transactions on Signal Processing}, 57(2):451--462, 2009.

\bibitem{liu2013observability}
Yang-Yu Liu, Jean-Jacques Slotine, and Albert-L{\'a}szl{\'o} Barab{\'a}si.
\newblock Observability of complex systems.
\newblock {\em Proceedings of the National Academy of Sciences},
  110(7):2460--2465, 2013.

\bibitem{summers2016submodularity}
Tyler~H. Summers, Fabrizio~L. Cortesi, and John Lygeros.
\newblock On submodularity and controllability in complex dynamical networks.
\newblock {\em IEEE Transactions on Control of Network Systems}, 3(1):91--101,
  2016.

\bibitem{yang2016observability}
Yang Yang and Filippo Radicchi.
\newblock Observability transition in real networks.
\newblock {\em Physical Review E}, 94(3):030301, 2016.

\bibitem{pearson1901lines}
Karl Pearson.
\newblock On lines and planes of closest fit to systems of points in space.
\newblock {\em The London, Edinburgh, and Dublin Philosophical Magazine and
  Journal of Science}, 2(11):559--572, 1901.

\bibitem{jolliffe2016principal}
Ian~T. Jolliffe and Jorge Cadima.
\newblock Principal component analysis: a review and recent developments.
\newblock {\em Philosophical Transactions of the Royal Society A: Mathematical,
  Physical and Engineering Sciences}, 374(2065):20150202, 2016.

\bibitem{brunton2016sparse}
Bingni~W. Brunton, Steven~L. Brunton, Joshua~L. Proctor, and J.~Nathan Kutz.
\newblock Sparse sensor placement optimization for classification.
\newblock {\em SIAM Journal on Applied Mathematics}, 76(5):2099--2122, 2016.

\bibitem{cheng2021discrimination}
Xiaoqing Cheng, Wai-Ki Ching, Sini Guo, and Tatsuya Akutsu.
\newblock Discrimination of attractors with noisy nodes in boolean networks.
\newblock {\em Automatica}, 130:109630, 2021.

\bibitem{radicchi2017superblockers}
Filippo Radicchi and Claudio Castellano.
\newblock Fundamental difference between superblockers and superspreaders in
  networks.
\newblock {\em Physical Review E}, 95(1):012318, 2017.

\bibitem{garg2008synchronous}
Abhishek Garg, Alessandro Di~Cara, Ioannis Xenarios, Luis Mendoza, and Giovanni
  De~Micheli.
\newblock Synchronous versus asynchronous modeling of gene regulatory networks.
\newblock {\em Bioinformatics}, 24(17):1917--1925, 2008.

\bibitem{albert2014boolean}
R{\'e}ka Albert and Juilee Thakar.
\newblock Boolean modeling: a logic-based dynamic approach for understanding
  signaling and regulatory networks and for making useful predictions.
\newblock {\em WIREs Systems Biology and Medicine}, 6(5):353--369, 2014.

\bibitem{park2026succession}
Kyu~Hyong Park and R{\'e}ka Albert.
\newblock Succession-diagram-based {Markov} chains reveal the attractor
  landscape of asynchronous boolean networks.
\newblock {\em npj Systems Biology and Applications}, 2026.

\bibitem{zanudo2015cell}
Jorge G.~T. Za{\~n}udo and R{\'e}ka Albert.
\newblock Cell fate reprogramming by control of intracellular network dynamics.
\newblock {\em PLOS Computational Biology}, 11(4):e1004193, 2015.

\bibitem{gedeon2024network}
Tom{\'a}{\v s} Gedeon.
\newblock Network topology and interaction logic determine states it supports.
\newblock {\em npj Systems Biology and Applications}, 10:98, 2024.

\bibitem{kadelka2025critical}
Claus Kadelka and Kishore Hari.
\newblock Critical assessment of the ability of boolean threshold models to
  describe gene regulatory network dynamics.
\newblock {\em PNAS Nexus}, 4(8):pgaf228, 2025.

\bibitem{fraisse2025representation}
Achille Fraisse, Diego~A. Oyarz{\'u}n, and Meriem El~Karoui.
\newblock Representation learning of single-cell time-series with deep
  variational autoencoders.
\newblock {\em bioRxiv}, 2025.

\bibitem{breiman2001random}
Leo Breiman.
\newblock Random forests.
\newblock {\em Machine Learning}, 45(1):5--32, 2001.

\bibitem{pedregosa2011scikit}
Fabian Pedregosa, Ga{\"e}l Varoquaux, Alexandre Gramfort, Vincent Michel,
  Bertrand Thirion, Olivier Grisel, Mathieu Blondel, Peter Prettenhofer, Ron
  Weiss, Vincent Dubourg, Jake Vanderplas, Alexandre Passos, David Cournapeau,
  Matthieu Brucher, Matthieu Perrot, and {\'E}douard Duchesnay.
\newblock Scikit-learn: machine learning in {Python}.
\newblock {\em Journal of Machine Learning Research}, 12:2825--2830, 2011.

\bibitem{derrida1986random}
Bernard Derrida and Yves Pomeau.
\newblock Random networks of automata: a simple annealed approximation.
\newblock {\em Europhysics Letters}, 1(2):45--49, 1986.

\bibitem{rohlf2002criticality}
Thimo Rohlf and Stefan Bornholdt.
\newblock Criticality in random threshold networks: annealed approximation and
  beyond.
\newblock {\em Physica A: Statistical Mechanics and its Applications},
  310(1--2):245--259, 2002.

\bibitem{shapley1953value}
Lloyd~S. Shapley.
\newblock A value for n-person games.
\newblock In Harold~W. Kuhn and Albert~W. Tucker, editors, {\em Contributions
  to the Theory of Games, Volume II}, number~28 in Annals of Mathematics
  Studies, pages 307--317. Princeton University Press, 1953.

\bibitem{gareyjohnson1979computers}
Michael~R. Garey and David~S. Johnson.
\newblock {\em Computers and Intractability: A Guide to the Theory of
  {NP}-Completeness}.
\newblock W. H. Freeman, San Francisco, 1979.

\bibitem{averbeck2006neural}
Bruno~B. Averbeck, Peter~E. Latham, and Alexandre Pouget.
\newblock Neural correlations, population coding and computation.
\newblock {\em Nature Reviews Neuroscience}, 7:358--366, 2006.

\bibitem{morenobote2014information}
Rub{\'e}n Moreno-Bote, Jeffrey Beck, Ingmar Kanitscheider, Xaq Pitkow, Peter
  Latham, and Alexandre Pouget.
\newblock Information-limiting correlations.
\newblock {\em Nature Neuroscience}, 17:1410--1417, 2014.

\bibitem{mitosch2019temporal}
Karin Mitosch, Georg Rieckh, and Tobias Bollenbach.
\newblock Temporal order and precision of complex stress responses in
  individual bacteria.
\newblock {\em Molecular Systems Biology}, 15(2):e8470, 2019.

\bibitem{mohiuddin2022high}
Sayed~Golam Mohiuddin, Aslan Massahi, and Mehmet~A. Orman.
\newblock High-throughput screening of a promoter library reveals new persister
  mechanisms in {E}scherichia coli.
\newblock {\em Microbiology Spectrum}, 10(1):e02253--21, 2022.

\bibitem{dash2024library}
Suchintak Dash, Rahul Jagadeesan, Ines S.~C. Baptista, Vatsala Chauhan, Vinodh
  Kandavalli, Samuel M.~D. Oliveira, and Andre~S. Ribeiro.
\newblock A library of reporters of the global regulators of gene expression in
  {E}scherichia coli.
\newblock {\em mSystems}, 9(6):e00065--24, 2024.

\bibitem{lopreside2019comprehensive}
Antonia Lopreside, Xinyi Wan, Elisa Michelini, Aldo Roda, and Baojun Wang.
\newblock Comprehensive profiling of diverse genetic reporters with application
  to whole-cell and cell-free biosensors.
\newblock {\em Analytical Chemistry}, 91(23):15284--15292, 2019.

\bibitem{quiroga2005invariant}
Rodrigo Quian~Quiroga, Leila Reddy, Gabriel Kreiman, Christof Koch, and Itzhak
  Fried.
\newblock Invariant visual representation by single neurons in the human brain.
\newblock {\em Nature}, 435(7045):1102--1107, 2005.

\bibitem{chen2013ultrasensitive}
Tsai-Wen Chen, Trevor~J. Wardill, Yi~Sun, Stefan~R. Pulver, Sabine~L.
  Renninger, Amy Baohan, Eric~R. Schreiter, Rex~A. Kerr, Michael~B. Orger,
  Vivek Jayaraman, Loren~L. Looger, Karel Svoboda, and Douglas~S. Kim.
\newblock Ultrasensitive fluorescent proteins for imaging neuronal activity.
\newblock {\em Nature}, 499(7458):295--300, 2013.

\bibitem{kanwisher1997fusiform}
Nancy Kanwisher, Josh McDermott, and Marvin~M. Chun.
\newblock The fusiform face area: a module in human extrastriate cortex
  specialized for face perception.
\newblock {\em Journal of Neuroscience}, 17(11):4302--4311, 1997.

\bibitem{rust2010selectivity}
Nicole~C. Rust and James~J. DiCarlo.
\newblock Selectivity and tolerance (``invariance'') both increase as visual
  information propagates from cortical area {V4} to {IT}.
\newblock {\em Journal of Neuroscience}, 30(39):12978--12995, 2010.

\bibitem{marre2012mapping}
Olivier Marre, Dario Amodei, Nikhil Deshmukh, Kolia Sadeghi, Frederick Soo,
  Timothy~E. Holy, and Michael~J. Berry~II.
\newblock Mapping a complete neural population in the retina.
\newblock {\em Journal of Neuroscience}, 32(43):14859--14873, 2012.

\bibitem{devries2020large}
Saskia E.~J. de~Vries, Jerome~A. Lecoq, Michael~A. Buice, Peter~A. Groblewski,
  Gabriel~K. Ocker, Michael Oliver, David Feng, et~al.
\newblock A large-scale standardized physiological survey reveals functional
  organization of the mouse visual cortex.
\newblock {\em Nature Neuroscience}, 23(1):138--151, 2020.

\bibitem{steinmetz2019distributed}
Nicholas~A. Steinmetz, Peter Zatka-Haas, Matteo Carandini, and Kenneth~D.
  Harris.
\newblock Distributed coding of choice, action and engagement across the mouse
  brain.
\newblock {\em Nature}, 576(7786):266--273, 2019.

\bibitem{zylberberg2017untuned}
Joel Zylberberg.
\newblock The role of untuned neurons in sensory information coding.
\newblock bioRxiv, 2018.
\newblock Preprint, not peer reviewed.

\bibitem{okeeffe2019quantifying}
Kevin~P O'Keeffe, Amin Anjomshoaa, Steven~H Strogatz, Paolo Santi, and Carlo
  Ratti.
\newblock Quantifying the sensing power of vehicle fleets.
\newblock {\em Proceedings of the National Academy of Sciences},
  116(26):12752--12757, 2019.

\bibitem{saeb2015mobile}
Sohrab Saeb, Mi~Zhang, Christopher~J. Karr, Stephen~M. Schueller, Marya~E.
  Corden, Konrad~P. Kording, and David~C. Mohr.
\newblock Mobile phone sensor correlates of depressive symptom severity in
  daily-life behavior: An exploratory study.
\newblock {\em Journal of Medical Internet Research}, 17(7):e175, 2015.

\bibitem{mishra2020presymptomatic}
Tejaswini Mishra, Meng Wang, Ahmed~A. Metwally, Gireesh~K. Bogu, Andrew~W.
  Brooks, Amir Bahmani, Arash Alavi, Alessandra Celli, Emily Higgs, Orit
  Dagan-Rosenfeld, Bethany Fay, Susan Kirkpatrick, Ryan Kellogg, Michelle
  Gibson, Tao Wang, Erika~M. Hunting, Petra Mamic, Ariel~B. Ganz, Benjamin
  Rolnik, Xiao Li, and Michael~P. Snyder.
\newblock Pre-symptomatic detection of {COVID}-19 from smartwatch data.
\newblock {\em Nature Biomedical Engineering}, 4(12):1208--1220, 2020.

\bibitem{schmidt2018introducing}
Philip Schmidt, Attila Reiss, Robert D{\"u}richen, Claus Marberger, and Kristof
  Van~Laerhoven.
\newblock Introducing {WESAD}, a multimodal dataset for wearable stress and
  affect detection.
\newblock In {\em Proceedings of the 20th ACM International Conference on
  Multimodal Interaction (ICMI '18)}, pages 400--408, Boulder, CO, USA, 2018.

\bibitem{lal2004support}
Thomas~Navin Lal, Michael Schr{\"o}der, Thilo Hinterberger, Jason Weston,
  Martin Bogdan, Niels Birbaumer, and Bernhard Sch{\"o}lkopf.
\newblock Support vector channel selection in {BCI}.
\newblock {\em IEEE Transactions on Biomedical Engineering}, 51(6):1003--1010,
  2004.

\bibitem{montanari2022functional}
Arthur~N. Montanari, Chao Duan, Luis~A. Aguirre, and Adilson~E. Motter.
\newblock Functional observability and target state estimation in large-scale
  networks.
\newblock {\em Proceedings of the National Academy of Sciences},
  119(1):e2113750119, 2022.

\bibitem{haber2018state}
Aleksandar Haber, Ferenc Molnar, and Adilson~E. Motter.
\newblock State observation and sensor selection for nonlinear networks.
\newblock {\em IEEE Transactions on Control of Network Systems}, 5(2):694--708,
  2018.

\bibitem{kreutz2009systems}
Clemens Kreutz and Jens Timmer.
\newblock Systems biology: experimental design.
\newblock {\em The FEBS Journal}, 276(4):923--942, 2009.

\bibitem{melykuti2010discriminating}
Bence M{\'e}lyk{\'u}ti, Elias August, Antonis Papachristodoulou, and Hana
  El-Samad.
\newblock Discriminating between rival biochemical network models: three
  approaches to optimal experiment design.
\newblock {\em BMC Systems Biology}, 4:38, 2010.

\bibitem{vanlier2014optimal}
Joep Vanlier, Christian~A. Tiemann, Peter A.~J. Hilbers, and Natal A.~W. van
  Riel.
\newblock Optimal experiment design for model selection in biochemical
  networks.
\newblock {\em BMC Systems Biology}, 8:20, 2014.

\end{thebibliography}
\end{document}